\documentclass[12pt]{amsart}
\usepackage{network_dep}
\usepackage{bibunits}

\title{}
	\author{\normalsize\textsc{Denis Kojevnikov}\textsuperscript{\textasteriskcentered}}
	\address{\textsuperscript{\textasteriskcentered}\normalfont{Corresponding author. Department of Econometrics and Operations Research, Tilburg University, The Netherlands. Email: \href{mailto:D.Kojevnikov@tilburguniversity.edu}{D.Kojevnikov@tilburguniversity.edu}.}}
	\author{\normalsize\textsc{Vadim Marmer}\textsuperscript{\S}}
	\author{\normalsize\textsc{Kyungchul Song}\textsuperscript{\S}}
	\address{\noindent\textsuperscript{\S}\normalfont{Vancouver School of Economics, University of British Columbia, Canada.}}

\begin{document}

\date{\today }

\begin{center}
	\LARGE \textsc{Limit Theorems for Network Dependent Random Variables}
\end{center}


\thanks{We thank Bulat Gafarov, our co-editor Elie Tamer, the associate editor, and three anonymous referees for helpful comments. Marmer and Song gratefully acknowledge the financial support of the Social Sciences and Humanities Research Council of Canada. This research was enabled in part by support provided by Compute Canada (\href{http://www.computecanada.ca}{www.computecanada.ca}).}

%
%

\begin{bibunit}[elsart-harv]
\begin{abstract}{\footnotesize
		This paper is concerned with cross-sectional dependence arising because observations are interconnected through an observed network. Following \cite{Doukhan/Louhichi:99}, we measure the strength of dependence by covariances of nonlinearly transformed variables. We provide a law of large numbers and central limit theorem for network dependent variables. We also provide a method of calculating standard errors robust to general forms of network dependence. For that purpose, we rely on a network heteroskedasticity and autocorrelation  consistent (HAC) variance estimator, and show its consistency. The results rely on conditions characterized by tradeoffs between the rate of decay of dependence across a network and network's denseness. Our approach can accommodate data generated by network formation models, random fields on graphs, conditional dependency graphs, and large functional-causal systems of equations.
	}

	\bigskip


\end{abstract}

\begingroup
\let\MakeUppercase\relax 
\maketitle
\endgroup

\section{Introduction}

In this paper, we consider cross-sectional dependence arising because of observations' interdependence in a network. Datasets exhibiting such forms of dependence are common in economics and other disciplines, and the results derived in this paper will allow the researcher to formally argue the consistency and asymptotic normality of estimators with network dependent data. Moreover, to facilitate inference with network dependent data, we derive conditions for the consistency of the network heteroskedasticity and  autocorrelation consistent (HAC) robust variance estimator. The estimator can be used for construction of standard errors robust to general forms of network dependence.

The main results of this paper are three-fold: the Law of Large Numbers (LLN), the Central Limit Theorem (CLT), and the consistency of HAC estimators.\footnote{
	\cite{Conley:99:JOE} proposed a HAC estimator in a spatial random field model. See \cite{Kelejian/Prucha:07} and \cite{Kim/Sun:11:JOE} for spatial HAC estimators. \cite{Leung:19:WP} considers spatial and network HAC estimators in models of discrete choice with social interactions. \cite{Kojevnikov:19:WP} develops bootstrap-based alternatives to network HAC estimation.
}
We provide a unified condition for the LLN and CLT when the network is formed in a generic way such that the links are formed independently conditional on observed or unobserved variables. This includes various network formation models proposed and used in the literature. Our condition reveals an explicit tradeoff between the extensiveness of the cross-sectional dependence and the denseness of the network permitted. The condition is also simple, as it involves only the average of conditional link formation probabilities. The paper also provides generic high level conditions that can accommodate random fields on graphs, conditional dependency graphs, and large functional-causal systems of equations.

To model network dependence, we adopt the approach of $\psi$-dependence proposed by \cite{Doukhan/Louhichi:99}, and extend the notion to accommodate common shocks. The notion of $\psi$-dependence is simple and intuitive. Roughly speaking, $\psi$-dependence measures the strength of dependence between two sets of random variables in terms of the covariance between nonlinear functions of random variables.

A primary benefit of modeling through $\psi$-dependence comes when dependence among the variables is produced through a system of causal equations in which sharing of exogenous shocks creates cross-sectional dependence among the variables of interest. We give four broad classes of such examples, including those where the random variables are generated from primitive random variables through a nonlinear transform. These classes cover many sub-examples that are used in statistics and econometrics. In such examples, a traditional approach of modeling through various mixing properties is cumbersome, because it is hard to find primitive conditions that guarantee the mixing properties for the variables of interest. On the other hand, one often can write the covariance bounds of those variables in terms of the primitive exogenous shocks using the causal equations. This flexibility of the $\psi$-dependence notion, however, carries a cost. The $\psi$-dependence of a nonlinearly transformed $\psi$-dependent random variables is not necessarily ensured, if the nonlinear transform does not belong to the class in the original definition. This paper provides several auxiliary results for such situations.

Network models have been used to capture a complex form of interdependence among cross-sectional observations. These observations may represent actions by people or firms, or outcomes from industry sectors, assets or products. Random fields indexed by points in a lattice in a Euclidean space have often been adopted as a model of spatial dependence in econometrics and statistics. \cite{Conley:99:JOE} proposed using random field modeling to specify the cross-sectional dependence of observations in the context of GMM estimation. More recent contributions include \cite{Jenish/Prucha:09:JOE} and \cite{Jenish/Prucha:12:JOE}. See \cite{Jia:08:Eca} for an application in entry decisions in retail markets, and \cite{Boucher/Mourifie:17:EJ} for an inference problem for a network formation model. For limit theorems for such random fields in statistics, see \cite{Comets/Janzura:98:JAP} and the references therein. 

When the dependence ordering arises from geographic distances or their analogues, using such random fields appears natural. However, the dependence ordering often stems from pairwise relations among the sample units, which can be viewed as a form of a network. To apply the random field modeling, one would first need to transform these relations into a random field on a lattice in a Euclidean space using methods such as multidimensional scaling.\footnote{
	See, e.g., \cite{Borg/Groenen:05:MDS}. See also footnote 16 of \cite{Conley:99:JOE} on page 15.
}

However, embedding of a network into a lattice can distort the dependence ordering. In fact, we show in Section \ref{subsec:Network_Topology} that network dependence is not necessarily embedded as a random field indexed by a lattice in the Euclidean space with a fixed dimension, when the network has a maximum clique whose size increases as the network grows. Networks with a growing maximum clique size often arise from those with a power-law degree distribution and high clustering coefficients. These features are typically shared by social networks that are observed in practice. In this paper, we directly use a network as a model of dependence ordering, so that such an embedding is not required when dependence ordering comes from pairwise relations.

Associating dependence patterns with networks has been previously used in the literature. \cite{Stein:72:BerkeleySymp} introduced a notion of dependency graphs in studying the normal approximation of a sum of random variables which are allowed to be dependent only when they are adjacent in a given network. See also \cite{Janson:88:AP}, \cite{Baldi/Rinott:89:AP}, \cite{Chen/Shao:04:AP}, and \cite{Rinott/Rotar:96:JMA} for various results for normal approximation for variables with related local dependence structures, and \cite{Aronow&Samii:17}, \cite{Leung:20:ReStat}, and \cite{Song:17:ReStat} for recent applications of dependency graphs to network data. Modeling based on dependency graphs has drawbacks. In particular, it requires independence between variables that are not adjacent in the network, and hence is not adequate to model more extensive forms of dependence.

A closely related strand of the literature studies various models of Markov random fields and spatial autoregressive models. Markov random fields constitute an alternative class of models of dependence which imposes conditional independence restrictions based on the network structure.\footnote{See, e.g., \cite{Lauritzen:96:GraphicalModels} and \cite{Pearl:09:Causality}. Recently, \cite{Lee/Song:18:Bernoulli} established a central limit theorem using a more general local dependence notion that encompasses both dependency graphs and a class of Markov random fields. See also Chapter 19 of \cite{Murphy:12:ML} for applications in the literature of machine learning.} Spatial autoregressive models specify cross-sectional dependence through the weight matrix in linear simultaneous equations, and have been extensively studied in econometrics. See, among others, \cite{Lee:04:Eca} and \cite{Lee/Liu/Lin:10:EJ} and references therein. Also see \cite{Gaetan/Guyon:10:SpatialStatistics} for an extensive review of spatial modeling and limit theorems.

There is a line of recent research that pursues a general form of limit theorems in a situation where the dependence structure itself is generated through a stochastic mechanism. \cite{Kuersteiner/Pucha:15:WP} embed dependence along a network as a martingale model. Similarly, \cite{Kuersteiner:19:WP} adopted a conditional spatial mixingale modeling of cross-sectional dependence, and established limit theorems which accommodate various network formation models. \cite{Leung/Moon:19:WPb} focus on the normal approximation of network statistics when the network is formed according to a generalized version of a random geometric graph.

In contrast to the dependency graph modeling, and similarly to the recent strand of literature mentioned above, our approach permits dependence between random variables that are only indirectly linked through intermediary variables. In fact a dependency graph model can be viewed as a special case of our network dependence modeling. The approach in this paper is also distinct from Markov random fields modeling. Markov random fields are based on conditional independence restrictions among the variables. While limit theorems on Markov random fields rely on independence restrictions that come from conditioning on certain random variables, our modeling expresses the degree of stochastic dependence in terms of the distance in the network.

The rest of the paper is organized as follows. In Section \ref{sec:NetDep} of the paper, we define network dependence of stochastic processes and provide examples. In particular, Section \ref{sec:NF} describes a class of network formation models that our approach can accommodate. Condition NF in that section describes the restrictions one needs to impose in order to apply our results to data generated by network formation models. The condition ties link formation probabilities with network dependence patterns, and requires that dependence between nodes decays with network distance at a rate depending on the link formation probabilities.

In Section \ref{sec:LimThms}, we present the main results of the paper: the LLN and CLT. Condition ND in that section provides a unifying high level assumption for the asymptotic results. It also demonstrates the tradeoffs between how fast dependence decays with the network distance and network's denseness. Lemma \ref{lemma:network_form} establishes the connection between Conditions ND and NF.

Section \ref{sec:HAC} is devoted to deriving conditions for the consistency of HAC estimators. As with time series, the consistency of HAC estimators requires truncation of network autocovariances corresponding to large distances. The amount of truncation is determined by a bandwidth parameter. Our proposed bandwidth selection rule is given in equation \eqref{eq:hac_const}. While in the time series case the bandwidth parameter is typically proportional to a fractional power of the sample size, it is logarithmic in our case. More aggressive truncation (than in the time series case) is due to the fact that the time series dependence structure can be viewed as a sparse network with a fixed number of neighbors at any distance. However, in our case the number of neighbors at any distance can grow with the sample size, which can result in fast accumulation of errors in HAC estimation.

Using Monte Carlo simulations, we evaluate the finite sample performance of our HAC estimator in Section \ref{sec:MonteCarlo}. We find that HAC-based inference is accurate even in relatively dense networks. At the same, the performance of HAC-based confidence intervals can deteriorate with networks' denseness and the amount of network dependence.

The Supplemental Note to this paper contains additional proofs and simulation results.

\section{Network Dependence and Examples}\label{sec:NetDep}

\subsection{Network Topology and a Lattice in a Euclidean Space}
\label{subsec:Network_Topology}
Let $N_n=\{1,2,\ldots,n\}$ be the set of cross-sectional unit indices. Modeling cross-sectional dependence usually assumes a certain metric on $N_n$. 
In some examples, this distance can be motivated by geographic distances or economic distances measured in terms of economic outcomes. This paper focuses on the pattern of cross-sectional dependence that is shaped along a given network.

Suppose that we observe an \textit{undirected} network $G_n$ on $N_n$, where $G_n=(N_n,E_n)$, and $E_n\subseteq \{\{i,j\}: i,j \in N_n, i \ne j\}$ denotes the set of links. For $i,j\in N_n $, we define $d_n(i,j)$ to be the distance between $i$ and $j$ in $G_n$, i.e., the length of the shortest path between nodes $i$ and $j$ given $G_n$. The distance $d_n$ defines a metric on the set $N_n$. We refer to \textit{network dependence} as a stochastic dependence pattern of random variables governed by the distance $d_n$ in $G_n$.

Let $N_n(i;s)$ denote the set of the nodes that are within the distance $s$ from node $i$, and let $N_n^\partial(i;s)$ denote the set of the nodes that are exactly at the distance $s$ from
node $i$. That is,
\begin{align}
\label{eq:nodes_sets}
N_n(i;s)=\left\{ j\in N_n:d_n(i,j) \le s\right\} \qtextq{and} N_n^\partial (i;s)=\left\{ j\in N_n:d_n(i,j)=s\right\}.
\end{align}

Our first focus is on the relation between modeling dependence through network topology and that through random fields indexed by the elements of a finite subset of a metric space $(\mathcal{X},d_{\mathcal{X}})$. We denote the equilateral dimension of $\mathcal{X}$, i.e., the maximum number of equidistant points in $\mathcal{X}$ with respect to the distance $d_{\mathcal{X}}$, as $e(\mathcal{X})$. The main question here is whether any given connected network is embeddable in $\mathcal{X}$.\footnote{
	A network/graph is connected if there is a path between every pair of nodes.
}
The following definition makes the notion of embedding precise.
\begin{definition}
An isometric embedding of a network $G_n = (N_n, E_n)$ into a metric space $(\mathcal{X},d_{\mathcal{X}})$ is an injective map $b:N_n \to \mathcal{X}$ such that for all $i,j\in N_n$
\begin{equation}
\label{eq:equiv}
	d_{\mathcal{X}}(b(i),b(j))=d_n(i,j).
\end{equation}
\end{definition}

When such an isometry exists, it means that modeling cross-sectional dependence using a network topology can be viewed as a special case of modeling a random field on a finite subset of $\mathcal{X}$. The following result shows that this is not always possible when the clique number $\omega(G_n)$ of $G_n$, i.e., the number of nodes in a maximum clique in $G_n$, is large enough.\footnote{
	A clique of a graph $G$ is a subset of nodes such that every two distinct nodes are adjacent.
}

\begin{prop}
\label{prop:net_emb}
	A connected network $G_n$ is isometrically embeddable into a metric space $(\mathcal{X},d_{\mathcal{X}})$ only if $\omega(G_n)\le e(\mathcal{X})$.
\end{prop}
\begin{proof}
Suppose that $C$ is a maximum clique of $G_n$. It is obvious that there is no isometry between $C$ and $\mathcal{X}$ when $|C|>e(\mathcal{X})$.
\end{proof}

Proposition \ref{prop:net_emb} gives only a necessary condition for isometric embedding. Consider, for example, $\R^k$ equipped with the Euclidean distance, which has the equilateral dimension of $k+1$. Figure \ref{fig:net_emb} provides an example of a network with the maximum clique size of two that cannot be embedded into the Euclidean $\R^2$ space, which has the equilateral distance of three. Figure \ref{fig:net_emb_L_inf} provides an example with a non-Euclidean space. It shows a network with the maximal clique size of four that cannot be imbedded into $\R^2$ equipped with the $L_\infty$ distance, which has the equilateral dimension of four.
\begin{figure}[t]
	\hfill
	\begin{subfigure}{0.4\textwidth}
		\centering
		\vspace{.51125cm}
		\includegraphics{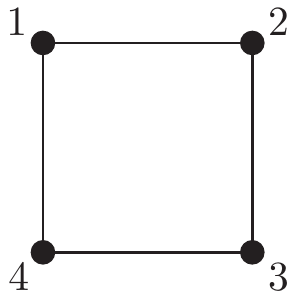}
		\vspace{.51125cm}
		\subcaption{}
	\end{subfigure}
	\hfill
	\begin{subfigure}{0.4\textwidth}
		\centering
		\includegraphics[scale=0.9]{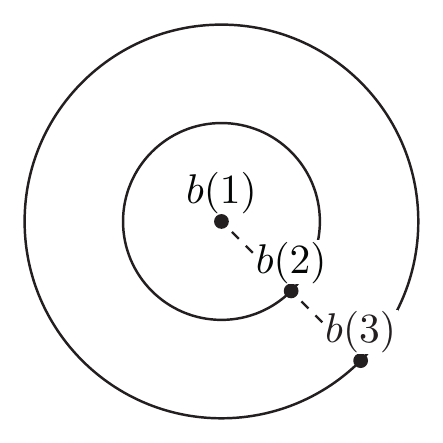}
		\subcaption{}
	\end{subfigure}
	\hfill\null
	\caption{\footnotesize An example of a network with the maximum clique size of two in panel (A) that cannot be embedded into $\R^2$ equipped with the Euclidean distance (with the equilateral dimension of three), as shown in panel (B). Node~2 has distance one from nodes~1 and~3, and node~3 has distance two from node~1. Their maps $b(1)$, $b(2)$, and $b(3)$ must be on the same line. If one maps node~4 to preserve its distance of one from nodes ~1 and~3, $b(4)$ would have zero distance from $b(2)$.}
\label{fig:net_emb}
\end{figure}

\begin{figure}[t]
 	\setbox1=\hbox{\includegraphics{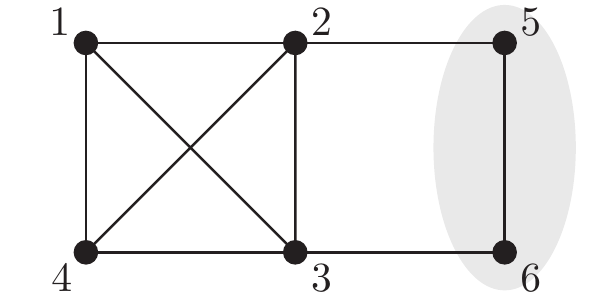}}
	\setbox2=\hbox{\includegraphics{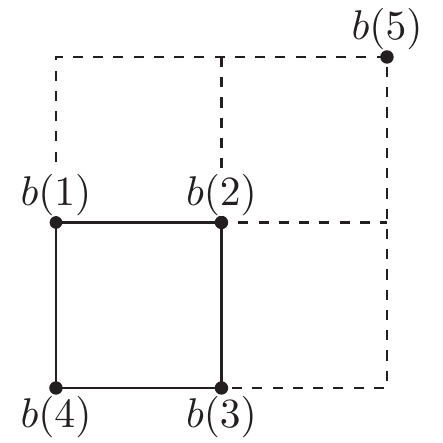}}
 	\hfill
	\begin{subfigure}{0.4\textwidth}
		\centering
		\vspace{0.6125cm}
		\includegraphics{Images/net_r2_Linf_a.pdf}
		\vspace{0.6125cm}
		\subcaption{}
	\end{subfigure}
	\hfill
	\begin{subfigure}{0.4\textwidth}
		\centering
		\includegraphics[scale=0.9]{Images/net_r2_Linf_b.pdf}
		\subcaption{}
	\end{subfigure}
	\hfill\null
	\caption{\footnotesize An example of a network with the maximum clique size of four in panel (A) that cannot be embedded into $\R^2$ equipped with the $L_\infty$ distance (with the equilateral dimension of 4), as shown in panel (B). Node~5 has distance one from node~2, and distance two from nodes~1,~3, and~4, which uniquely determine its map $b(5)$. Similarly, the distances between node~6 and nodes~1,~2,~3, and~4 uniquely determine $b(6)$; however, it would be inconsistent with distance one between nodes~5 and~6.}
\label{fig:net_emb_L_inf}
\end{figure}

An important consequence of Proposition \ref{prop:net_emb} is that when the size of the maximum cliques in the network $G_n$ grows to infinity as $n \rightarrow \infty$, the sequence of networks cannot be embedded into a metric space having a finite equilateral dimension. Examples of such spaces include a $k$-dimensional normed space $M^k$ and a sphere $\mathbb{S}^{k}$ equipped with the usual distance because $e(M^k)\le 2^{k}$ and $e(\mathbb{S}^k)= k+2$ \citep[see, e.g.,][]{Petty:71}. As a consequence, the random field models used in \cite{Conley:99:JOE} with the Euclidean distance and in \cite{Jenish/Prucha:09:JOE} with the Chebychev distance cannot include a network dependence model when the maximum clique size of the networks increases with the sample size. Indeed, there are random graphs whose degree distribution takes the form of a power law and the size of the maximum cliques grows to infinity as $n \to \infty$ \citep[see][]{Blasius/Friedrich/Krohmer:17:Algorithmica}. Such models accommodate both dense and sparse graphs, and are often motivated as a model of many real networks that we observe in practice.

The asymptotic results developed in this paper can accommodate network generating processes with the maximum clique size increasing with the sample size. However, our results impose certain restrictions on the rate of growth of the maximum clique size.

One may consider ``approximating" the network dependence ordering by a lattice in a finite dimensional Euclidean space. Multidimensional scaling (MDS) provides various ways to achieve such an approximation \citep[see][]{Borg/Groenen:05:MDS}. The dependence ordering obtained through MDS is itself dependent on the data, and is stochastic. Hence, it is generally different from the true dependence ordering of the data. Proposition \ref{prop:net_emb} tells us that there is no guarantee that the approximation error of the MDS-based dependence ordering will be small with a large sample size.

\subsection{Network Dependent Processes}
Suppose that we are given a triangular array of $\R^v$-valued random vectors, $Y_{n,i}, i \in N_n$, which are laid on a network $G_n$ whose agency matrix we denote by $A_n$. That is, the $(i,j)$-th entry $A_{n,ij}$ of matrix $A_n$ is one if $i$ and $j$ are adjacent in $G_n$ and zero otherwise. The $(i,i)$-th entry of $A_n$ is zero for $i\in N_n$. We adapt the $\psi$-dependence notion of \cite{Doukhan/Louhichi:99} to our setup. We define $\N = \{1,2,3,\ldots\}$, and for any $v,a \in \N$, we endow $\R^{v \times a}$ with the distance
\begin{equation}
\label{eq:distance}
	\mathss{d}_a(\vec{x},\vec{y})\eqdef\sum_{l=1}^a\norm{x_l-y_l},
\end{equation}
where $\vec{x}=(x_1,\ldots,x_a)$ and $\vec{y}=(y_1,\ldots,y_a)$ are points in $\R^{v\times a}$, and $\| \cdot \|$ denotes the Euclidean norm, i.e., $\|w\| = \sqrt{w^{\top}w}$, for $w \in \R^a$. Let
\begin{equation}
	\label{eq:Lv}
	\L_v \eqdef \left\{\L_{v,a}: a \in \N \right\},
\end{equation}
where $\L_{v,a}$ denotes the collection of bounded Lipschitz real functions on $\R^{v \times a}$, i.e.,
\begin{equation}
	\label{eq:Lva}
	\L_{v,a}\eqdef \left\{f:\R^{v\times a}\to \R : \norm{f}_{\infty}<\infty,\Lip(f)<\infty \right\},
\end{equation}
with $\Lip(f)$ denoting the Lipschitz constant of $f$,\footnote{\label{footnote:Lipschitz_constant}
	The Lipschitz constant for a function $f:\R^{v \times a} \rightarrow \R$ is the smallest constant $C$ such that $|f(\vec{x}) - f(\vec{y})| \le C \mathss{d}_a(\vec{x},\vec{y})$, for all $\vec{x},\vec{y} \in \R^{v \times a}$.
}
and $\|\cdot\|_\infty$ the sup-norm of $f$, i.e., $\|f\|_\infty = \sup_x|f(x)|$. For any positive integers $a,b,s$, consider two sets of nodes (of size $a$ and $b$) with distance between each other of at least $s$. Let $\P_n(a,b;s)$ denote the collection of all such pairs:
\begin{equation}
	\label{eq:P(a,b;s)}
	\P_n(a,b;s) = \{(A,B): A,B \subset N_n, |A| = a, |B| = b, \textnormal{ and } d_n(A,B) \ge s \},
\end{equation}
where
\begin{equation}
	\label{eq:d_n}
	d_n(A,B) = \min_{i \in A} \min_{i' \in B} d_n(i,i'),
\end{equation}
and $d_n(i,i')$ denotes the distance between nodes $i$ and $i'$ in $G_n$, i.e., the length of the shortest path between $i$ and $i'$ in $G_n$. For each set $A$ of positive integers, we write
\begin{equation}
	\label{eq:Y_A}
	Y_{n,A}=(Y_{n,i})_{i \in A}.
\end{equation}
We take $\{\C_n\}_{n \ge 1}$ to be a given sequence of $\sigma$-fields such that for each $n \ge 1$, the adjacency matrix $A_n$ of graph $G_n$ is $\C_n$-measurable. Below we introduce a notion of conditional $\psi$-dependence for a triangular array $\{Y_{n,i}\}_{i \in N_n}, n \ge 1, Y_{n,i} \in \R^v$. From here on, we write triangular arrays simply as $\{Y_{n,i}\}$, and sequences $\{\C_n\}_{n \ge 1}$ as $\{\C_n\}$.

\begin{definition}
\label{def:psi_dep}
A triangular array $\{Y_{n,i}\}, n \ge 1, Y_{n,i} \in \R^v$, is called \textit{conditionally} $\psi$-\textit{dependent} given $\{\C_n\}$, if for each $n \in \N$, there exist a $\C_n$-measurable sequence $\theta_n=\{\theta_{n,s}\}_{s\ge 0}$, $\theta_{n,0}=1$, and a collection of nonrandom functions $(\psi_{a,b})_{a,b \in \N}$, $\psi_{a,b}: \L_{v,a} \times \L_{v,b} \rightarrow [0,\infty)$, such that for all $(A,B) \in \P_n(a,b;s)$ with $s>0$ and all $f\in \L_{v,a}$ and $g\in \L_{v,b}$,
\begin{equation}
\label{eq:psi_dep}
	\abs{\Cov\left(f(Y_{n,A}),g(Y_{n,B})\mid \C_n\right)}\le \psi_{a,b} (f,g) \theta_{n,s} \qtext{a.s.}
\end{equation}
In this case, we call the sequence $\theta_n$ \textit{the dependence coefficients} of $\{Y_{n,i}\}$.
\end{definition}

In a typical set-up that we consider in this paper, $\{\theta_{n,s}\}$ approaches zero as $s$ grows. The $\sigma$-field $\C_n$ can be thought of as a ``common shock" such that when we condition on it, the cross-sectional dependence of triangular array $\{Y_{n,i}\}$ becomes substantially weaker. However, we do not have to think of $\C_n$ as being originated from a variable that affects every node in the network. In many network set-ups, $\C_n$ can be thought of as having been generated by some characteristics or actions of multiple central nodes which affect many other nodes through their many links. For example, consider a star network, where node $1$ is adjacent to the other $n-1$ nodes. Suppose that $Y_{n,1}=U_1$ corresponds to the central node, and for the remaining nodes ($i\geq 2$), 
\[
	Y_{n,i}=U_1+U_i,
\]
where $\{U_i:i=1,\ldots,n\}$ are independent. In that case, we can take $\C_n=\sigma(U_1)$. Then, conditionally on $\C_n$, $Y_{n,2},\ldots, Y_{n,n}$ are i.i.d., and $\PM\{\theta_{n,2}=0\mid \C_n\}=1$.

Unlike the unconditional version of $\psi$-dependence of \cite{Doukhan/Louhichi:99}, in our definition the dependence coefficients $\{\theta_n\}$ are random, due to our accommodation of the common shocks, $\C_n$. We make the following assumption.

\begin{assumption} The triangular array $\{Y_{n,i}\}$ is conditionally $\psi$-dependent given $\{\C_n\}$ with the dependence coefficients $\{\theta_n\}$ satisfying the following conditions.
	\label{assu:psi_dep}
	\begin{enumerate}[leftmargin=*]
	\item[(a)] For some constant $C>0$,
	\begin{align*}
		\psi_{a,b}(f,g)\leq C \times a b \left(\norm{f}_{\infty}+\Lip(f)\right) \left(\norm{g}_{\infty}+\Lip(g)\right).
	\end{align*}
	\item[(b)] $\sup_{n \ge 1} \max_{s \ge 1} \theta_{n,s} < \infty$ a.s.
	\end{enumerate}
\end{assumption}

Assumption \ref{assu:psi_dep} will be maintained throughout the paper. It is shown to be satisfied by all the examples we present in the next subsection. The following lemma shows that $\psi$-dependence of random vectors carries over to linear combinations of their elements.

\begin{lemma}
	\label{lemma:linear_comb}
	 Suppose that a triangular array $\{Y_{n,i}\}$, $Y_{n,i}\in\R^v$, satisfies Assumption \ref{assu:psi_dep}(a) with the dependence coefficients $\{\theta_n\}$. For each $n\ge 1$, let $\{c_{n,i}\}_{i \in N_n}$ be a sequence of $\C_n$-measurable vectors in $\R^v$ such that $\max_{i \in N_n} \norm{c_{n,i}} \le 1$ a.s. Then the array $\{Z_{n,i}\}$ defined by $Z_{n,i}\eqdef c_{n,i}^{\top}Y_{n,i}$ is conditionally $\psi$-dependent given $\{\C_n\}$ with the dependence coefficients $\{\theta_n\}$.
\end{lemma}

A result similar to Lemma \ref{lemma:linear_comb} holds for nonlinear transforms of random variables, under certain conditions for the nonlinear transforms. See Appendix \ref{subsec:covariance_inequalities} for details.

\subsection{An Overview of the Limit Theorems}
\subsubsection{A Motivating Example}
Limit theorems in this paper focus on the asymptotic behavior of the following sum
\begin{equation}
\label{eq:sum}
	\sum_{i \in N_n} (Y_{n,i} - \E[Y_{n,i}\mid \C_n]),
\end{equation}
where $\C_n$ is a certain $\sigma$-field. The asymptotic behavior of such a sum often arises in the network-based interactions models. For example, let us consider the following linear interaction model:
\[
	y_i = \beta \overline{y}_i + \gamma' X_i + v_i,
\]
where
\begin{align}
	\label{ybar}
	\overline{y}_i = \frac{1}{\abs{N_n^\partial(i;1)}}\sum_{j \in N_n^\partial(i;1)} y_j,
\end{align}
and $N_n^\partial(i;1)$ is defined in \eqref{eq:nodes_sets}. Such a model has been widely studied and used in the literature on social interactions \citep[see, e.g.,][and the references therein]{Blume/Brock/Durlauf/Jayaraman:15:JPE}. This literature usually assumes that the error term $v_i$ in the outcome equation is uncorrelated with the network $G_n$. In other words, the network $G_n$ is exogenously formed. A recent paper by \cite{Johnsson/Moon:19:ReStatForthComing} extends the framework to accommodate a situation where $G_n$ is endogenously formed, by introducing an explicit yet generic network formation model and proposing a control function approach. Their network formation model is given as follows:
\[
	j \in N_n^\partial(i;1) \qtext{if and only if}\quad f_{ij}(X,T_i,T_j) \ge u_{ij},
\]
where $f_{ij}$ is a nonstochastic map, $X = (X_i)_{i \in N_n}$, $T_i$'s are unobserved individual heterogeneity affecting the network formation process, and $u_{ij}$'s are link-specific error terms. \cite{Johnsson/Moon:19:ReStatForthComing} introduce a set of assumptions which imply the following conditions:
\begin{itemize}[wide=0pt,leftmargin=*]
	\item[(Condition A)] $v_i$'s are conditionally i.i.d. given $\C_n$, and
	\item[(Condition B)] $\E[v_i \mid \C_n] = \E[v_i\mid T_i]$, for all $i \in N_n$,
\end{itemize}
where $\C_n$ is the $\sigma$-field generated by $(X,T,u)$, $X = (X_i)_{i \in N_n}$, $T = (T_i)_{i \in N_n}$, and $u = (u_{ij})_{i,j \in N_n}$. Then the normal approximation of the distribution of estimators for $\beta$ stems from the limit distribution of the sum of the form:
\[
	\sum_{i \in N_n} (v_i - \E[v_i\mid T_i])\varphi_i,
\]
where $\varphi_i$ is a random variable that is constructed as a function of $X$ and $G_n$ (such as instrumental variables). This sum becomes \eqref{eq:sum} if we take $Y_{n,i} = (v_i - \E[v_i\mid T_i])\varphi_i$, and by Condition B, we have $\E[Y_{n,i}\mid \C_n] = 0$. Our limit theorems in this paper can be used to relax the conditional i.i.d. assumption in Condition A to accommodate the case where the terms $v_i - \E[v_i\mid T_i]$ exhibit network dependence along $G_n$, for instance, through a data generating process as in one of the examples in Section 2.4.

\subsubsection{Network Dependence with Network Formation}
\label{sec:NF}

In contrast to time series dependence, limit theorems for network dependent processes depend not only on the strength of the dependence but also on the shape of the network itself. In this section, we consider a class of network formation models and give a sufficient condition for the network formation process. First, consider a generic network $G_n=(N_n,E_n)$ for which the link between each pair of nodes $\{i,j\}$, $i\ne j$, is realized randomly as follows: $i$ and $j$ are linked if and only if $A_{n,ij} = 1$, where
\begin{equation}
\label{eq:graph}
	A_{n,ij} = 1\{\varphi_{n,ij} \ge \varepsilon_{ij}\},
\end{equation}
$\varphi_{n,ij}$'s and $\varepsilon_{ij}$'s are random variables such that $\varphi_{n,ij}=\varphi_{n,ji}$, $\varepsilon_{ij}=\varepsilon_{ji}$, and $\{\varepsilon_{ij}:i<j\}$ are i.i.d. and independent of $\varphi_n=(\varphi_{n,ij})_{i<j}$. This random graph model can be viewed as a generalization of the Erd\"{o}s--R\'{e}nyi graph model in the sense that conditional on $\varphi_n$, the link formation probabilities can be heterogeneous across all pairs of nodes. Many network formation models used in the literature take this form, where
\[
	\varphi_{n,ij} = f_n(W_{ij},T_i,T_j),
\]
for some function $f_n$, $W_{ij}$ is observable, and $T_i$ is an unobservable node-specific component. For example, \cite{Graham:17} specified $f_n$ as follows:
\[
	f_n(W_{ij},T_i,T_j) = W_{ij}'\beta_0 + T_i + T_j.
\]
\cite{Ridder/Sheng:19:WP} considered an endogenous network formation model where the payoff depends not only on the neighbors' characteristics and the characteristics of their 2-neighbors. They find that the model yields the following ``reduced-form" for the formation of the network with
\[
	f_n(W_{ij},T_i,T_j) = V_{n,ij}(W_{ij}),
\]
where we take $W_{ij}$ to be the covariates $(X_i)_{i \in N_n}$ and $V_{n,ij}$ is a nonstochastic map. \cite{Leung:19:JOE} studied a network model which is contained in a graph generated through
\[
	f_n(W_{ij},T_i,T_j) = \sup_s V(r_n^{-1}\|X_i - X_j\|,s,W_{ij}),
\]
where $(X_i,X_j)$ is a part of the vector $W_{ij}$, $V$ is a map, and $r_n = (\kappa/n)^{1/d}$, with $d$ representing the dimension of $X_i$ and $\kappa$ is a positive constant. This latter graph can be viewed as a generalized version of a random geometric graph.

Suppose that $\{Y_{n,i}\}$ is conditionally $\psi$-dependent given $\{\C_n\}$ with the dependence coefficients $\{\theta_{n}\}$, such that
\begin{equation}
\label{eq:moment_cond}
	\sup_{n \ge 1} \max_{i \in N_n}\E\left[|Y_{n,i}|^p\mid \C_n\right] < \infty \qtext{a.s.}
\end{equation}
for some $p > 4$, and the functional $\psi_{a,b}$ satisfies Assumption \ref{assu:psi_dep}(a). Our main interest is in the LLN and CLT of the following form:
\begin{align}
\label{eq:LLN_CLT}
	\begin{aligned}
		\frac{1}{n}\sum_{i \in N_n} (Y_{n,i} - \E[Y_{n,i}\mid \C_n]) &\rightarrow_p 0, \text{ and } \\
		\frac{1}{\sigma_n}\sum_{i \in N_n} (Y_{n,i} - \E[Y_{n,i}\mid \C_n]) &\rightarrow_d \mathcal{N}(0,1),
	\end{aligned}
\end{align}
where $\sigma_n^2 = \Var(\sum_{i \in N_n} Y_{n,i}\mid \C_n)$. The sparsity of the graph $G_n$ necessary for the limit theorems in this paper is summarized by the asymptotic behavior of the maximal expected degree $\pi_n$, where
\begin{equation}
	\label{eq:pi_n}
	\pi_n = \max_{ i \in N_n} \sum_{j \in N_n\setminus \{i\}} \PM\{\varphi_{n,ij} \ge \varepsilon_{ij} \mid \varphi_n\}.
\end{equation}
In this paper we show that these limit theorems hold if the following condition holds.\medskip

\noindent \textbf{Condition NF.} There exist $\varepsilon>0$ and $q > \max\{p/(p-4),3p/(p-1)\}$ for $p >4$ in \eqref{eq:moment_cond} and a positive random variable $M$ such that for all $n \ge 1$,
\[
	\theta_{n,s} \le M((\pi_n \vee 1) + \varepsilon)^{-qs}, \quad \quad 1 \le s \le n,
\]
holds eventually with probability one.\footnote{
	A sequence of events $\{E_n\}_{n\ge 1}$ holds eventually with probability one if $\PM(\bigcup_{n\ge 1}\bigcap_{m\ge n} E_m)=1$.
}
\medskip

For example, Condition NF is satisfied if there exists $\gamma \in (0,1)$ such that $\theta_{n,s} \le \gamma^s, 1 \le s \le n$, and $\gamma < ((\pi_n \vee 1) + \varepsilon)^{-q}$ for some $q > \max\{p/(p-4),3p/(p-1)\}$.

\subsection{Examples}
In this section, we consider four broad classes of examples of conditionally $\psi$-dependent random vectors.

\subsubsection{Strong-Mixing Processes}
\label{sec:example_mixing}
Let $(\Omega,\mathcal{F},\PM)$ be an underlying probability space. For sub $\sigma$-fields $\mathcal{G}$, $\mathcal{H}$, $\C$ of $\mathcal{F}$, let
\[
	\alpha(\mathcal{G},\mathcal{H}\mid \C)\eqdef\sup_{G\in\mathcal{G},H\in \mathcal{H}}\abs{\Cov(\ind_G,\ind_H\mid \C)}.
\]
For a triangular array $\{Y_{n,i}\}$ and a sequence of $\sigma$-fields $\{\C_n\}$ we define the strong mixing coefficients by\footnote{
	These coefficients are different from those given in \cite{Jenish/Prucha:09:JOE} because our $\theta_n$ coefficients do not depend on $\abs{A}$ and $\abs{B}$.
}
\begin{align}
\label{alpha_ns}
	\alpha_{n,s}\eqdef \sup\left\{\alpha\left(\sigma(Y_{n,A}),\sigma(Y_{n,B})\mid \C_n\right):A,B\subset N_n, d_n(A,B)\ge s\right\}.
\end{align}
The proposition below provides a conditional covariance inequality that is due to Theorem 9 of \cite{PrakasaRao:13}.
\begin{prop}
\label{prop:example_mixing}
For $f\in \L_{v,a}$, $g\in \L_{v,b}$, and $(A,B)\in\P_n(a,b;s)$,
\[
	\abs{\Cov(f(Y_{n,A}),g(Y_{n,B})\mid \C_n)}\le 4\norm{f}_{\infty}\norm{g}_{\infty}\alpha_{n,s} \qtext{a.s.}
\]
\end{prop}

Hence, the array $\{Y_{n,i}\}$ is conditionally $\psi$-dependent given $\{\C_n\}$ with $\psi_{a,b}(f,g)=4\norm{f}_{\infty}\norm{g}_{\infty}$, and the dependence coefficients $\{\theta_{n,s}\}_{s \ge 1}$ are given by the strong mixing coefficients $\{\alpha_{n,s}\}_{s\ge 1}$.

The proof of Proposition \ref{prop:example_mixing} follows by adapting the proof of Theorem A.5. of \cite{HallHeyde:80:MLT} to the conditional settings and noticing that the strong mixing coefficients can be equivalently defined by replacing $\alpha(\mathcal{G},\mathcal{H}\mid \C)$ with $\alpha(\mathcal{G}\vee \C,\mathcal{H}\vee \C\mid \C)$.

\subsubsection{Conditional Dependency Graphs}
Suppose that $\{Y_{n,i}\}$ is a given collection of random vectors and $G_n = (N_n,E_n)$ is a graph on the index set $N_n$. Let $\C_n$ be a given $\sigma$-field. We say that $\{Y_{n,i}\}$ has $G_n$ as a \textit{conditional dependency graph given} $\C_n$, if for any set $A \subset N_n$, $Y_{n,A}$ and $\{Y_{n,i}: i \in N_n\setminus N_n(A)\}$ are conditionally independent given $\C_n$, where $N_n(A) = \bigcup_{i\in A}N_n(i;1)$. The notion of a conditional dependency graph is a conditional variant of a dependency graph introduced by \cite{Stein:72:BerkeleySymp}. It is not hard to see that when $\{Y_{n,i}\}_{i\in N_n}$ has $G_n$ as a conditional dependency graph given $\C_n$ for each $n\ge 1$, the array $\{Y_{n,i}\}$ is conditionally $\psi$-dependent given $\{\C_n\}$ with
\[
	\psi_{a,b}(f,g) = 4 \|f\|_\infty \|g\|_\infty,
\]
and $\theta_n$ is such that $\theta_{n,s} = 0$ for all $s \ge 1$.

\subsubsection{Functional Dependence on Independent Variables}
\label{sec:example_eta}

Consider a triangular array of $\R^k$-valued random vectors $\{\varepsilon_{n,i}\}_{i \in N_n}$ which is row-wise independent given $\C_n$. For $\R^v$-valued measurable functions $\{\bm\phi_{n,i}\}_{i \in N_n}$, let
\[
	Y_{n,i}\eqdef\bm\phi_{n,i}(\varepsilon_n), \quad i\in N_n,
\]
where $\varepsilon_n = (\varepsilon_{n,j} : j\in N_n)$. Further, define a modified version of $Y_{n,i}$, which replaces too distant shocks $\varepsilon_{n,j}$ with zeros:
\[
	Y_{n,i}^{(s)}\eqdef\bm\phi_{n,i}\big(\varepsilon_n^{(s)}\big),
\]
where $\varepsilon_n^{(s,i)} = (\varepsilon_{n,j}\ind\{j\in N_n(i;s)\} : j\in N_n)$ and $N_n(i;s)$ is defined in \eqref{eq:nodes_sets}.\footnote{
	Zero can be replaced with another constant if the functions $\bm\phi_{n,i}$ are undefined at zero.
} Now, for any $A,B\subset N_n$ with $d_n(A,B)>2s$, $Y_{n,A}^{(s)}$ and $Y_{n,B}^{(s)}$ are conditionally independent given $\C_n$.

\begin{prop}
\label{prop:example_eta}
Let $\{Y_{n,i}\}$ be as described above. Then for any $(A,B)\in \P_n(a,b;2s+1)$ and $f\in\L_{v,a},g\in\L_{v,b}$,
\[
	\absin{\Cov(f(Y_{n,A}),g(Y_{n,B})\mid \C_n)} \le \big(a\norm{g}_{\infty}\Lip(f)+b\norm{f}_{\infty}\Lip(g)\big)\theta_{n,s} \qtext{a.s.},
\]
where $\theta_{n,s}=2\max_{i\in N_n}\E[\normin{Y_{n,i}-Y_{n,i}^{(s)}}\mid \C_n]$.
\end{prop}

It follows from Proposition \ref{prop:example_eta} that $\{Y_{n,i}\}$ is conditionally $\psi$-dependent given $\{\C_n\}$, where the $\psi$ function is given by
\[
	\psi_{a,b}(f,g) = a\norm{g}_{\infty}\Lip(f)+b\norm{f}_{\infty}\Lip(g).
\]
This functional $\psi_{a,b}(f,g)$ satisfies Assumption \ref{assu:psi_dep}(a).

Proposition \ref{prop:example_eta} can be extended to the case where $\varepsilon_n = (\varepsilon_{n,j}: j \in N_n)$ is $\psi$-dependent, as shown below. Let $\phi_{n,ir}$ denote the $r$-th component of $\bm\phi_{n,i}$, and we endow the domain $\R^{k \times n}$ of $\phi_{n,ir}$ with the norm $\mathss{d}_n$ defined in \eqref{eq:distance} with $a = n$, so that $\Lip(\phi_{n,ir})$ represents the Lipschitz constant of $\phi_{n,ir}$ with respect to $\mathss{d}_n$.

\begin{prop}
\label{prop:example_eta2}
Let $(A,B) \in \P_n(a,b;3s)$, and let $\{Y_{n,i}\}$, $f$, and $g$ be as in Proposition \ref{prop:example_eta}, except that $\varepsilon_n = (\varepsilon_{n,j}: j \in N_n)$ is $\psi$-dependent with coefficient $\theta_{n,s}^\varepsilon$ and $\psi$ equal to a functional $\psi_{a,b}^\varepsilon$ which satisfies Assumption \ref{assu:psi_dep}(a). Then
\[
	\absin{\Cov(f(Y_{n,A}),g(Y_{n,B})\mid \C_n)}\le \psi_{a,b}(f,g) \theta_{n,s} \qtext{a.s.},
\]
where
\begin{align}
\label{eq:psi2}
	\begin{aligned}
		\theta_{n,s} &= 2\max_{i\in N_n}\E[\normin{Y_{n,i}-Y_{n,i}^{(s)}}\mid \C_n] + D^2_n(s) \theta_{n,s}^\varepsilon, \text{ and } \\
		\psi_{a,b}(f,g) &= a\norm{g}_{\infty}\Lip(f)+b\norm{f}_{\infty}\Lip(g) \\
		&\quad+ C \times ab\left( \|f\|_\infty +\sqrt{k} \Lip(f) \bar \phi \right)\left( \|g\|_\infty + \sqrt{k} \Lip(g) \bar \phi \right)
	\end{aligned}
\end{align}
with $D_n(s)=\max_{i\in N_n}\vert N_n(i;s)\vert$,
\[
	\bar \phi = \sup_{n \ge 1} \max_{i \in N_n} \sum_{r=1}^k \Lip(\phi_{n,ir}),
\]
and $C>0$ is the constant in Assumption \ref{assu:psi_dep}.
\end{prop}

As a concrete example, consider a simple linear case in which
\[
	Y_{n,i}=\sum_{m\ge 0}\gamma_{m,n}\sum_{j\in N_n^{\partial}(i;m)}\varepsilon_{n,j},
\]
where $\varepsilon_{n,j}$'s are $\psi$-dependent as in Proposition \ref{prop:example_eta2} such that $D_{n}^2(s) \theta_{n,s}^\varepsilon \rightarrow_{a.s.} 0$, and $N_n^\partial(i;m)$ is as defined in \eqref{eq:nodes_sets}. Let us take $\|\cdot \|$ to be the Euclidean norm. Since
\[
	\normin{Y_{n,i}-Y_{n,i}^{(s)}}\le \sum_{m>s}\abs{\gamma_{m,n}}\sum_{j\in N_n^{\partial}(i;m)}\normin{\varepsilon_{n,j}},
\]
setting $\alpha_n\eqdef \max_{i\in N_n}\E[\norm{\varepsilon_{n,i}}\mid \C_n]$, we find that
\begin{equation}
\label{eq:theta_bound}
	\theta_{n,s}\le 2\alpha_n\sum_{m>s}\abs{\gamma_{m,n}}\max_{ i \in N_n} \abs{N_n^\partial(i;m)} + D_n^{2}(s) \theta_{n,s}^\varepsilon \qtext{a.s.}
\end{equation}

As for the functional $\psi_{a,b}$ in the $\psi$-dependence of $\{Y_{n,i}\}_{i \in N_n}$, observe that for any $x, \tilde x \in \R^{n \times k}$ (with their $(j,k)$-th entries denoted by $x_{j,r}$ and $\tilde x_{j,r}$),
\begin{align*}
	\sum_{r=1}^k \left| \sum_{m \ge 0} \gamma_{m,n} \sum_{j \in N_n^\partial(i;m)}(x_{j,r} - \tilde x_{j,r}) \right| &\le \sum_{m \ge 0}|\gamma_{m,n}|\sum_{j \in N_n} \sum_{r=1}^k |x_{j,r} - \tilde x_{j,r}|\\
	&\le \sqrt{k} \sum_{m \ge 0}|\gamma_{m,n}| \sum_{j \in N_n}\|x_j - \tilde x_j\|\\
	&= \sqrt{k} \sum_{m \ge 0}|\gamma_{m,n}| \mathss{d}_n(x, \tilde x),
\end{align*}
where $x_j = (x_{j,r})_{r=1}^k$ and $\tilde x_j = (\tilde x_{j,r})_{r=1}^k$. Hence,
\[
	\sum_{r=1}^k \Lip(\phi_{ni,r}) =\sqrt{k} \sum_{m \ge 0}|\gamma_{m,n}|.
\]
Thus, we can find the functional $\psi$ in the $\psi$-dependence of $\{Y_{n,i}\}$ as $\psi_{a,b}$ in \eqref{eq:psi2} with
\[
	\bar \phi = \sqrt{k} \sup_{n \ge 1} \sum_{m \ge 0}|\gamma_{m,n}|.
\]
The functional $\psi_{a,b}$ in \eqref{eq:psi2} satisfies Assumption \ref{assu:psi_dep}(a) if $\bar \phi < \infty$ in this model.

\subsubsection{Functional Dependence on Associated or Gaussian Variables}\label{subsec:associated_procs}
Let us consider the following process:
\[
	Y_{n,i}= \varphi_{n,i}(\varepsilon_n),
\]
where $\varepsilon_n = (\varepsilon_{n,i} : i\in N_n)$ is a positively associated process, $\varepsilon_{n,i} \in \R$, conditional on certain $\sigma$-field $\C_n$, i.e., for all coordinatewise non-decreasing real-valued measurable functions $f$ and $g$ and all finite subsets $A$ and $B$ of $N_n$,
\[
	\Cov(f(\varepsilon_{n,A}), g(\varepsilon_{n,B})\mid \C_n) \ge 0 \qtext{a.s.}
\]
When the above inequality is reversed for all finite subsets $A$ and $B$ of $N_n$, we say that $\{\varepsilon_{n,i}\}_{i \in N_n}$ is negatively associated. When a set of random variables is positively or negatively associated, independence between two random variables in the set is equivalent to their being uncorrelated. The following result follows as a consequence of a covariance inequality due to Theorem 3.1 of \cite{Birkel:88:AP} and Lemma 19 of \cite{Doukhan/Louhichi:99}.

\begin{prop}
	\label{prop: associated proc}
	Suppose that for each $i \in N_n$, $\varphi_{n,i}\in \mathscr{C}_b^1$ (i.e., $\varphi_{n,i}$ is continuously differentiable with bounded derivatives). Let $A,B\in\P_n(a,b;s)$ and let $f:\R^a \to \R$ and $g:\R^b \to \R$ be differentiable with bounded derivatives. Suppose further that either (i) $\{\varepsilon_{n,i}\}_{i \in N_n}$ is conditionally positively or negatively associated given $\C_n$ or (ii) $\{\varepsilon_{n,i}\}_{i \in N_n}$ is conditionally Gaussian given $\C_n$ and $f$ and $g$ are bounded.
	Then
	\[
		\abs{\Cov(f(Y_{n,A}),g(Y_{n,B})\mid \C_n)} \le ab \Lip(f)\Lip(g) \theta_{n,s} \qtext{a.s.},
	\]
	where
	\begin{equation}
	\label{eq:theta_assoc_proc}
		\theta_{n,s} = \max_{(A',B') \in \mathcal{P}_n(a,b;s)}\max_{k_1 \in A',k_2 \in B'} \sum_{i,j\in N_n} \left\|\frac{\partial \varphi_{n,k_1}}{\partial \varepsilon_{n,i}}\right\|_{\infty} \left\|\frac{\partial \varphi_{n,k_2}}{\partial \varepsilon_{n,j}}\right\|_{\infty} \abs{\Cov(\varepsilon_{n,i},\varepsilon_{n,j}\mid \C_n)}.
	\end{equation}
\end{prop}

The above proposition clearly shows that the dependence structure of $\{Y_{n,i}\}$ is determined by the (conditional) local dependence structure of $\varepsilon_{n,i}$'s and $\varphi_{n,i}$'s. In the special case where $\varepsilon_{n,i}$'s are all conditionally independent given $\C_n$, the sequence $\theta_{n,s}$ is reduced to the following:
\[
	\theta_{n,s} = \max_{(A',B') \in \mathcal{P}_n(a,b;s)} \max_{k_1 \in A',k_2 \in B'} \left\|\frac{\partial \varphi_{n,k_1}}{\partial \varepsilon_{n,i}}\right\|_{\infty} \left\|\frac{\partial \varphi_{n,k_2}}{\partial \varepsilon_{n,i}}\right\|_{\infty}.
\]
Suppose further that $N_n$ is endowed with a graph $G_n$ such that $\partial \varphi_{n,k}/\partial \varepsilon_{n,i} = 0$, whenever $i$ is at least $m$-edges away from $k$ in $G_n$. Then $Y_{n,i}$'s have a graph $G_n'$ as a conditional dependency graph given $\C_n$, where $i$ and $j$ are adjacent in $G_n'$ if and only if $i$ and $j$ are within $2m$ edges away. Hence, it follows that $\theta_{n,s} = 0$, for all $s \ge 2m$.

The following corollary shows that the array $\{Y_{n,i}\}$ is conditionally $\psi$-dependent given $\{\C_n\}$.

\begin{corollary}
Suppose that $\varphi_{n,i}\in \mathscr{C}_b^1$ for all $i\in N_n$ and $n\ge 1$. Then the triangular array $\{Y_{n,i}\}$ is conditionally $\psi$-dependent given $\{\C_n\}$ with the coefficients given by \eqref{eq:theta_assoc_proc} and
\[
	\psi_{a,b}(f,g)=ab\Lip(f)\Lip(g).
\]
\end{corollary}
\begin{proof}
The result follows from the fact that for any $\epsilon>0$, a Lipschitz function $f$ admits an approximation by a continuously differentiable function $f_{\epsilon}$ s.t.\ $\norm{f-f_{\epsilon}}_{\infty}\le \epsilon$ and $\Lip(f_{\epsilon})\le \Lip(f)$ \citep[see, e.g.,][p.~174]{JimenezSevilla:11}.
\end{proof}
\section{Limit Theorems for Network Dependent Processes}
\label{sec:LimThms}

\subsection{Network Dependence Condition}

In this section, we provide a sufficient condition for the shape of the network that ensures our limit theorems (i.e., the LLN and CLT) hold. The crucial aspect of the network which matters for the limit theorem is the properties of the neighborhood shells. For the limit theorems to hold, the number of the neighbors at distance $s$ should not grow too fast as $s$ increases. The precise condition for such neighborhood shells depends on the dependence coefficients $\theta_{n,s}$, so that if $\theta_{n,s}$ decreases fast as $s$ increases, the requirement for the neighborhood shells can be weakened.

To introduce sufficient conditions, let
\begin{equation}
\label{eq:delta_n}
	\delta_n^{\partial}(s;k)= \frac{1}{n}\sum_{i\in N_n}\absin{N_n^{\partial}(i;s)}^k,
\end{equation}
where $N_n^\partial(i;s)$ is defined in \eqref{eq:nodes_sets}. When $k=1$, we simply write $\delta_n^\partial (s;1) = \delta_n^\partial(s)$. This quantity measures the denseness of a network. Let us introduce further notation. Define
\begin{equation}
\label{eq:Delta_n}
	\Delta_n(s,m;k) = \frac{1}{n}\sum_{i \in N_n} \max_{j \in N_n^\partial(i;s)}\abs{N_n(i;m)\setminus N_n(j;s-1)}^{k},
\end{equation}
where $N_n(i;s)$ is defined in \eqref{eq:nodes_sets}, and we take $N_n(j;s-1) = \varnothing$ if $s = 0$. We also define
\begin{equation}
\label{eq:c_n}
	c_n(s,m;k)= \inf_{\alpha>1}\left[\Delta_n(s,m;k\alpha)\right]^{\frac{1}{\alpha}} \left[\delta_n^\partial\left(s;\frac{ \alpha}{\alpha-1}\right)\right]^{1 - \frac{1}{\alpha}}.
\end{equation}
The quantity $c_n(s,m;k)$ is easy to compute when a network is given, and it captures the network properties that are relevant for the limit theorems. It consists of two components: $\Delta_n(s,m;k\alpha)$ and $\delta_n^\partial(s;\alpha/(\alpha-1))$. They capture the denseness of the network through the average neighborhood sizes and the average neighborhood shell size. We summarize a sufficient condition for the network and the weak dependence coefficient as follows.\medskip

\noindent \textbf{Condition ND.} There exist $p>4$ and a sequence $m_n \to \infty$ such that
\begin{enumerate}[leftmargin=*]
	\item[(a)] $\theta_{n,m_n}^{1-1/p} = o_{a.s}(n^{-3/2})$,
	\item[(b)] for each $k\in\{1,2\}$,
		\[
		\frac{1}{n^{k/2}}\sum_{s\ge 0} c_n(s,m_n;k) \theta_{n,s}^{1-\frac{k+2}{p}} =o_{a.s.}(1), \text{ and}
		\]
	\item[(c)] $\sup_{n \ge 1} \max_{i \in N_n} \E[ \vert Y_{n,i} \vert^p \mid \C_n] <\infty$ a.s.
\end{enumerate}
\medskip

Later we show that Condition ND is sufficient for the LLN and CLT in \eqref{eq:LLN_CLT}. Note that $\Delta_n(s,m_n;k)$ tends to decrease fast to zero as $s$ goes beyond a certain level, because the set $N_n(j;s-1)$ quickly becomes large.

The following lemma shows that in the case of the network formation model in \eqref{eq:graph}, Condition NF implies Conditions ND(a) and (b).

\begin{lemma}
\label{lemma:network_form}
Suppose that network $G_n$ is generated as in \eqref{eq:graph}, and let $\mathcal{C}_n$ be a $\sigma$-field such that the adjacency matrix of network $G_n$ is $\mathcal{C}_n$-measurable. Suppose further that $\{Y_{n,i}\}$ is conditionally $\psi$-dependent given $\{\C_n\}$, with dependence coefficients $\{\theta_n\}$ satisfying Condition NF. Then Conditions ND(a) and (b) hold.
\end{lemma}

Thus, Condition NF is a sufficient condition on the network formation model for the LLN and CLT in \eqref{eq:LLN_CLT} for $\{Y_{n,i}\}$. The proof of Lemma \ref{lemma:network_form} is found in the Supplemental Note to this paper. The proof is built on a bound on the tail probability of $\delta_n^\partial(s;k)$. This bound is obtained using similar arguments in \cite{Chung/Lu:01:AAM} for the case of Erd\"{o}s--R\'{e}nyi graphs.

\subsection{Law of Large Numbers}

Let $\{Y_{n,i}\}$ be conditionally $\psi$\hyp{}dependent given $\{\C_n\}$.
Since a LLN can be applied element-by-element in the vector case, without loss of generality we can assume that $Y_{n,i}\in\R$ in this section, i.e., $v=1$.

Let $\norm{Y_{n,i}}_{\C_n,p}=(\E[ \vert Y_{n,i} \vert^p \mid \C_n])^{1/p}$. We assume the following moment condition.
%
\begin{assumption} For some $\epsilon>0$,
	\label{assu:Uniform_L1_Integrability}
		$\sup_{n\geq1}\max_{i\in N_n} \norm{Y_{n,i}}_{\C_n,1+\epsilon}<\infty$ a.s.
\end{assumption}

%
The next assumption puts a restriction on the denseness of the network and the rate of decay of dependence with the network distance.
\begin{assumption}
	\label{assu:LLN_suff_delta}
	$n^{-1}\sum_{s\ge 1}\delta^\partial_n(s)\theta_{n,s}\to_{a.s.}0$.
\end{assumption}
Note that the above assumption is implied by Condition ND(b) for $k=2$: $\theta_{n,s}\leq \theta_{n,s}^{1-4/p}$ for $p>4$, and $ \delta_n^\partial(s) \leq 4c_n(s,m;2)$.\footnote{The second inequality follows from \eqref{eq:H_bound} and \eqref{eq:delta_H_bound} in the supplement.}

Assumption \ref{assu:LLN_suff_delta}
can fail, for example, if there is a node connected to almost every other node in the network as in the following example. Consider a network with the \emph{star} topology, which has a central node or hub connected to every other node. In this case, the distance between any two nodes does not exceed 2: $\delta_n^\partial(1)=2(n-1)/n$, $\delta_n^\partial(2)=(n-2)(n-1)/n$, and $\delta_n^\partial(s)=0$ for $s\geq 3$. Hence, Assumption \ref{assu:LLN_suff_delta} fails for a star network, unless $\theta_{n,2}=0$, i.e., unless there is no network dependence at the distance $s>1$.

Alternatively, consider a network with the ring topology, where nodes are connected in a circular fashion to form a loop, see Figure \ref{fig:ring_net}(A) in \ref{sec:HAC}. In that case, $\delta_n^\partial(s)\leq 2$, and Assumption \ref{assu:LLN_suff_delta} holds when $n^{-1}\sum_{s\ge 1}\theta_{n,s}\to_{a.s.} 0$.

The following theorem establishes a conditional LLN.
\begin{theorem}
	\label{thm:pointwise_LLN}Suppose that $\{Y_{n,i}\}$ is conditionally $\psi$-dependent given $\{\C_n\}$ and Assumptions \ref{assu:psi_dep}(a), \ref{assu:Uniform_L1_Integrability}, and \ref{assu:LLN_suff_delta} hold. Then as $n\to\infty$,
	\[
		\norm{\frac{1}{n}\sum_{i \in N_n} \left(Y_{n,i}-\E[Y_{n,i} \mid \C_n ]\right)}_{\C_n,1}\to_{a.s.} 0.
	\]
\end{theorem}
An unconditional version of the result, which replaces the conditional norm in Theorem \ref{thm:pointwise_LLN} with the unconditional norm, can be established in a similar manner by replacing the conditional moment in Assumption \ref{assu:Uniform_L1_Integrability} with the unconditional moment.



Next, we discuss LLNs for nonlinear functions of $\{Y_{n,i}\}$. When $f\in \L_{v,1}$, a LLN for a nonlinear transformation $f(Y_{n,i})$ follows immediately from the definition of the $\psi$-dependence in Definition \ref{def:psi_dep}.\footnote{
	Note that compositions of bounded Lipschitz functions are also bounded and Lipschitz; hence, $f(Y_{n,i})$ is also $\psi$-dependent when $f\in\L_v$.
}
In that case,
\[
	\norm{\frac{1}{n}\sum_{i \in N_n} \left(f(Y_{n,i})-\E[f(Y_{n,i}) \mid \C_n ]\right)}_{\C_n,2}^2 \leq \frac{2}{n}  \norm{ f}_{\infty}^2 +\psi_{1,1}(f,f) \frac{1}{n} \sum_{s\ge 1}\delta^\partial_n(s)\theta_{n,s}.
\]
We have the following result.
\begin{prop}
Suppose that $\{Y_{n,i}\}$ is conditionally $\psi$-dependent given $\{\C_n\}$, Assumption \ref{assu:LLN_suff_delta} holds, and $f\in \L_{v,1}$. Then as $n\to\infty$,
\[
	\norm{\frac{1}{n}\sum_{i \in N_n}(f(Y_{n,i})-\E[f(Y_{n,i}) \mid \C_n])}_{\C_n,2} \to_{a.s.} 0.
\]
\end{prop}

However, in general nonlinear transformations of $\psi$-dependent processes are not necessarily $\psi$-dependent. In such cases, LLNs for nonlinear transformations can be established using the covariance inequalities for transformation functions presented in Appendix \ref{subsec:covariance_inequalities} in this paper.
For example, suppose that the assumptions of Corollary \ref{cor:cov_ineq_simple} in the appendix hold for some nonlinear function $h(\cdot)$ of a $\psi$-dependent process $\{Y_{n,i}\}$, and that $\theta_{n,s}$ is bounded by a constant uniformly over $s \ge 1$ and $n \ge 1$. In that case for some constants $C>0$ and $p>2$, the conditional covariance given $\C_n$ between $h(Y_{n,i})-\E[h(Y_{n,i}) \mid \C_n]$ and $h(Y_{n,j})-\E[h(Y_{n,j}) \mid \C_n]$ is bounded by
\[
 	C\cdot \sup_{n,i}\norm{h(Y_{n,i})}_{\C_n,p}^2 \cdot \theta_{n,d_n(i,j)}^{1-\frac{2}{p}}.
\]
Therefore, as $n\to\infty$,
\[
	 \norm{\frac{1}{n}\sum_{i \in N_n}(h(Y_{n,i})-\E[h(Y_{n,i}) \mid \C_n])}_{\C_n,2} \to_{a.s.} 0,
\]
provided that $\sup_{n,i}\norm{h(Y_{n,i})}_{\C_n,p}<\infty$ a.s., and a condition similar to that in Assumption \ref{assu:LLN_suff_delta} holds:
\[
	\frac{1}{n} \sum_{s\ge 1} \delta_n^\partial(s) \theta_{n,s}^{1-\frac{2}{p}}=o_{a.s.}(1).
\]
Cases not covered by Corollary \ref{cor:cov_ineq_simple} can be handled in a similar manner using the covariance inequality of Theorem \ref{thm:cov_ineq2} in Appendix \ref{subsec:covariance_inequalities}. We use such a strategy to show the consistency of the HAC estimator in Section \ref{sec:HAC}.
\subsection{Central Limit Theorem}

In this section, we study the CLT for a sum of random variables that are conditionally $\psi$-dependent. Define
\begin{equation}
\label{eq:sigma_n2}
	\sigma_n^2 = \Var(S_n\mid \C_n),
\end{equation}
where $S_n = \sum_{i \in N_n} Y_{n,i}$. The assumption below presents a moment condition.
\begin{assumption}
	\label{assu:moment}
	For some $p>4$, $\sup_{n\ge 1}\max_{i \in N_n} \|Y_{n,i}\|_{\C_n,p} < \infty$ a.s.
\end{assumption}
While the moment condition in Assumption \ref{assu:moment} is more restrictive than those conditions known for the CLT for special cases of $\psi$-dependence, such a moment condition is widely used in many models in practice. The following assumption limits the extent of the cross-sectional dependence of the random variables through restrictions on the network.
\begin{assumption}
	\label{assu:theta_bound}
	There exists a positive sequence $m_n \to \infty$ such that for $k=1,2$,
	\begin{align*}
		&\frac{n}{\sigma_n^{2+k}}\sum_{s\ge 0} c_n(s,m_n;k) \theta_{n,s}^{1-\frac{2+k}{p}} \rightarrow_{a.s.} 0, \text{ and } \\
		&\frac{n^2 \theta_{n,m_n}^{1-(1/p)}}{\sigma_n} \to_{a.s.} 0,
	\end{align*}
	as $n \rightarrow \infty$, where $p>4$ is that appears in Assumption \ref{assu:moment}.
\end{assumption}
It is not hard to see that Condition ND is a sufficient condition for this assumption, when $\sigma_n \ge c \sqrt{n}$ with probability one, for some constant $c>0$ that does not depend on $n$. The latter condition is satisfied if the ``long-run variance", $\Var(S_n\mid \C_n)/n$ is bounded away from $c^2>0$ for all $n \ge 1$.

The theorem below establishes the CLT for the normalized sum $S_n/\sigma_n$.

\begin{theorem}
\label{thm:CLT}
Suppose that Assumptions \ref{assu:psi_dep}, \ref{assu:moment}-\ref{assu:theta_bound} hold, and that $\E[Y_{n,i}\mid \C_n]=0$ a.s. Then
\[
	\sup_{t \in \R}\abs{\PM\left\{\frac{S_n}{\sigma_n} \le t\mid \C_n\right\} - \Phi(t)} \to_{a.s.} 0, \text{ as } n \to \infty,
\]
where $\Phi$ denotes the distribution function of $\mathcal{N}(0,1)$.
\end{theorem}

The proof of the CLT uses Stein's Lemma \citep{Stein:86}. The CLT immediately gives a stable convergence of a normalized sum of random variables under appropriate conditions. More specifically, suppose that
\[
	\sigma_n^2/(nv^2) \rightarrow_{a.s.} 1,
\]
where $v^2$ is a random variable that is $\C$-measurable and $\C$ is a sub $\sigma$-field of $\C_n$ for all $n \ge 1$. Then it follows that $S_n/\sqrt{n}$ converges stably to a mixture normal random variable.

\section{Network HAC Estimation}\label{sec:HAC}

In this section, we develop network HAC estimation of the conditional variance of $S_n/\sqrt{n}$ given $\C_n$, where $S_n\eqdef\sum_{i\in N_n}Y_{n,i}$. First, we assume that $\E[Y_{n,i}\mid \C_n]=0$ a.s.\ for all $i\in N_n$. Let
\begin{equation}
\label{eq:Omega_ns}
	\Omega_n(s)\eqdef n^{-1}\sum_{i\in N_n}\sum_{j\in N_n^{\partial}(i;s)}\E[Y_{n,i}Y_{n,j}^{\top}\mid \C_n].
\end{equation}
Then the conditional variance of $S_n/\sqrt{n}$ given $\C_n$ is given by 
\begin{equation}\label{eq:trueV}
	V_n\eqdef\Var(S_n/\sqrt{n}\mid \C_n)=\sum_{s\ge 0}\Omega_n(s) \qtext{a.s.}
\end{equation}

Similarly to the time-series case, the asymptotic consistency of an estimator of $V_n$ requires a restriction on weights given to the estimated ``autocovariance" terms $\Omega_n(\csdot)$. Consider a kernel function $\omega:\bar{\R}\to [-1,1]$ such that $\omega(0)=1$, $\omega(z)=0$ for $\abs{z}>1$, and $\omega(z)=\omega(-z)$ for all $z \in \bar{\R}$.

Let $b_n$ denote the bandwidth or the lag truncation parameter. Then the kernel HAC estimator of $V_n$ is given by
\begin{equation}
\label{eq:HAC1}
	\tilde{V}_n=\sum_{s\ge 0}\omega_n(s)\tilde\Omega_n(s),
\end{equation}
where $\omega_n(s)\eqdef\omega(s/b_n)$, and
\begin{equation}
\label{eq:Omega_tilde}
	\tilde \Omega_n(s)\eqdef n^{-1}\sum_{i\in N_n}\sum_{j\in N_n^{\partial}(i;s)}Y_{n,i}Y_{n,j}^{\top}.
\end{equation}
The weight given for each sample covariance term $\tilde\Omega_n(s)$ is a function of distance $s$ implied by the structure of a network. Also notice that if nodes $i$ and $j$ are disconnected then $d_n(i,j)=\infty$ so that $\omega_n(d_n(i,j))=0$. 

Unlike the time series case, the number of terms included in the double sum in \eqref{eq:Omega_tilde} depends on the shape of the network. Hence, if there are many empty neighborhood shells, a large value of the bandwidth can still produce a HAC estimator that performs well in finite samples.

Next, assume that $\E[Y_{n,i}\mid \C_n]=\Lambda_n$ a.s.\ for all $i\in N_n$ and the sequence of common conditional expectations $\{\Lambda_n\}$ is unknown.\footnote{If random vectors $\{Y_{n,i}\}_{i\in N_n}$ do not share a common expectation, it is hard to justify plugging the sample mean into $\hat{\Omega}_n(\csdot)$ because $\bar{Y}_n$ is not a consistent estimator of $\E[Y_{n,i}\mid \C_n]$.} By Theorem \ref{thm:pointwise_LLN}, $\bar Y_n\eqdef S_n/n$ is a consistent estimator of $\Lambda_n$ in the sense that $\E[\normin{\bar{Y}_n-\Lambda_n}\mid \C_n]\to_{a.s.}0$. We redefine the kernel HAC estimator given in \eqref{eq:HAC1} as follows:
\begin{equation}
\label{eq:HAC2}
	\hat{V}_n=\sum_{s\ge 0}\omega_n(s)\hat\Omega_n(s),
\end{equation}
where
\begin{equation}
\label{eq:Omega_hat}
	\hat\Omega_n(s)\eqdef n^{-1}\sum_{i\in N_n}\sum_{j\in N_n^{\partial}(i;s)}\left(Y_{n,i}-\bar Y_n\right)\left(Y_{n,j}-\bar Y_n\right)^{\top}.
\end{equation}

\subsection{Consistency}\label{subsec:consistency_HAC}

We establish the consistency of the estimators \eqref{eq:HAC1} and \eqref{eq:HAC2} by imposing suitable conditions on the moments of the array $\{Y_{n,i}\}$, the denseness of a sequence of networks, and the rate of growth of the bandwidth parameter.

\begin{assumption}
\label{assu:HAC1}
There exists $p>4$ such that
\begin{enumerate}[label=(\roman*),leftmargin=*,align=left]
	\item $\sup_{n \ge 1}\max_{i\in N_n}\norm{Y_{n,i}}_{\C_n,p}<\infty$ a.s.,
	\item $\lim_{n\to\infty}\sum_{s\ge 1}\abs{\omega_n(s)-1}\delta_n^{\partial}(s) \theta_{n,s}^{1-(2/p)}=0$ a.s., and
	\item $\lim_{n\to\infty}n^{-1}\sum_{s\ge 0}c_n(s,b_n;2) \theta_{n,s}^{1-(4/p)}=0$ a.s.\label{it:weights}
\end{enumerate}
\end{assumption}

The assumption demonstrates the tradeoff between the conditional moments of $\{Y_{n,i}\}$ given $\{\C_n\}$ and the magnitude of the network dependence. For a given sequence of networks, a stronger network dependence requires the finiteness of higher conditional moments, i.e., a larger value of $p$. On the other hand, sparse networks allow for either weaker moments conditions or a stronger dependence along the network. Note that Assumptions \ref{assu:HAC1}(i) and (iii) are implied by Condition ND. 

Assumption \ref{assu:HAC1}(ii) is a high-level condition, which requires that the kernel weights $\omega_n(s)$ converge to one sufficiently fast as $n\to\infty$. Proposition \ref{prop:kernel} below provides primitive conditions for Assumption \ref{assu:HAC1}(ii) in the case of models satisfying Condition NF.

Assumption \ref{assu:HAC1}(iii) determines the admissible rate of growth of the sequence of bandwidths $\{b_n\}$. In particular, it strongly depends on the network topology. In case of models satisfying Condition NF, the following bandwidth selection rule is motivated by equation \eqref{eq:m_n} in the proof of Lemma \ref{lemma:network_form} in the Supplemental Note:
\begin{equation}\label{eq:hac_const}
	b_n=\textrm{constant}\times \frac{1}{\log(\textrm{avg.deg}\vee (1+\varepsilon))}\times \log n,
\end{equation}
where ``avg.deg" is the average degree $\delta_n^\partial(1)$ of the observed network and used to approximate $\pi_n$ in Condition NF. For example, in the case of the Parzen kernel, we found through extensive Monte Carlo simulations that setting the constant in \eqref{eq:hac_const} to $2.0$ and $\varepsilon = 0.05$ works well, see Section \ref{sec:MonteCarlo} for the details.

We define
\begin{equation}
\label{eq:delta2_n}
	\delta_n(b_n)= n^{-1} \sum_{i=1}^n \vert N_n(i;b_n) \vert.
\end{equation}
Note that $\delta_n(b_n)$ measures the denseness of a network in terms of the average size of $b_n$-neighborhoods. Let $\norm{\csdot}_F$ denote the Frobenius norm.\footnote{\label{footnote:Frobenius}
	For a real matrix $A$, $\norm{A}_F\eqdef \sqrt{\trace(A^{\top}A)}$.
}

\begin{prop}
\label{prop:HAC}
Suppose that Assumptions \ref{assu:psi_dep} and \ref{assu:HAC1} hold. Then as $n\to\infty$,
\[
	\E[\normin{\tilde{V}_n-V_n}_F\mid \C_n]\to_{a.s.} 0.
\]
If, in addition, $\delta_n(b_n)=o_{a.s.}(n)$, and $\{\theta_{n,s}/s^{p/(p-4)}\}$ are non-increasing in $s\geq 1$ a.s., then as $n\to\infty$,
\[
	\E[\normin{\hat{V}_n-V_n}_F\mid \C_n]\to_{a.s.} 0.
\]
\end{prop}

In the second part of the proposition, the non-increasing in $s$ condition for $\{\theta_{n,s}/s^{p/(p-4)}\}$ is a mild requirement consistent with the notion of weak dependence. For example, in the linear model below Proposition \ref{prop:example_eta2} with conditionally independent $\varepsilon_{n,i}$'s, we can take $\theta_{n,s}$ to be the bound on the right hand side of \eqref{eq:theta_bound} to satisfy this monotonicity condition.

Next, we provide primitive conditions for Assumption \ref{assu:HAC1}(ii) in the case of networks satisfying Condition NF.
\begin{prop}\label{prop:kernel}
Suppose that Condition NF holds, and for some constants $C,\eta>0$,
\begin{equation}\label{eq:kernel}
	\vert \omega(x)-1 \vert \leq C\vert x\vert^{1+\eta}.
\end{equation}
Suppose further that $b_n\to\infty$ a.s.\ and $(\log{n})/b_n=O_{a.s.}(\pi_n\vee 1)$. Then Assumption \ref{assu:HAC1}(ii) is satisfied.
\end{prop}
The bandwidth condition in Proposition \ref{prop:kernel} is consistent with the bandwidth selection rule in \eqref{eq:hac_const}. The condition in \eqref{eq:kernel} is satisfied by many commonly used kernels such as the truncated kernel $\omega(x)=1\{\abs{x}\leq 1\}$, Parzen, and Tukey--Hanning kernels \cite[see][p. 824]{Andrews:91}. However, 
\eqref{eq:kernel} does not hold for the Bartlett kernel.

\begin{figure}[t]
	\hfill
	\begin{subfigure}[b]{0.25\textwidth}
		\includegraphics[width=\textwidth]{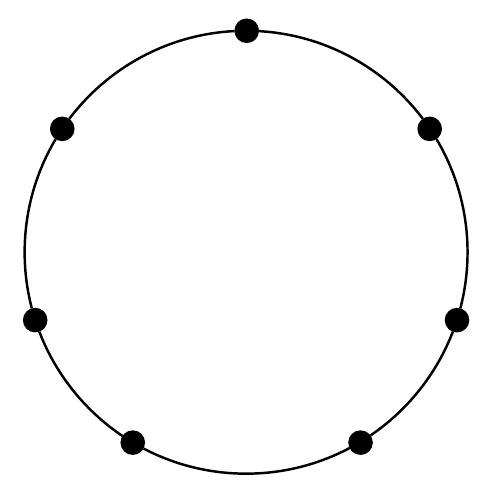}
	\caption{}
	\end{subfigure}
	\hfill
	\begin{subfigure}[b]{0.25\textwidth}
		\includegraphics[width=\textwidth]{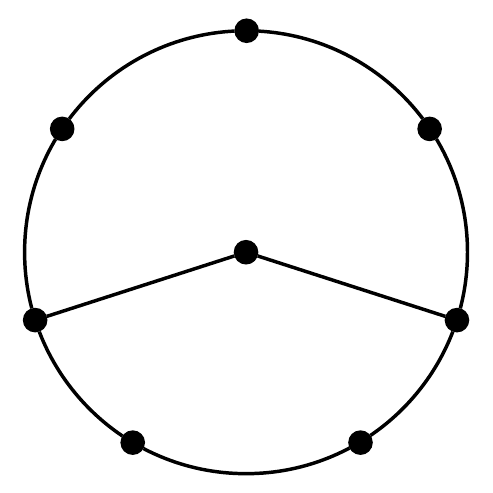}
	\caption{}
	\end{subfigure}
	\hfill\null
\caption{\footnotesize An example of networks for which the corresponding weighting matrices $W=[\omega(d_n(i,j)/2)]_{i,j\in N_n}$ are either positive semidefinite (A) or indefinite (B) for the same positive definite kernel function $\omega(z)=\ind\{\abs{z}\le 1\}(1-\abs{z})$.}
\label{fig:ring_net}
\end{figure}

While according to Proposition \ref{prop:HAC} the proposed HAC estimators are consistent, they are not necessarily positive semidefinite. The following example provides a simple case in which positive definiteness of the kernel function does not automatically imply positive semidefiniteness of the estimated covariance matrix.

\begin{example}
\label{ex:circ_network}
Consider a ring network (an example is shown in Figure \ref{fig:ring_net}(A)), where $N_n^{\partial}(i;s)=2$ for $1\le s \le \lfloor(n-1)/2\rfloor$ and all $i\in N_n$. Suppose that $\Lambda_n=0$ a.s., and let $\omega(z)=(1-|z|)\ind\{|z|\le 1\}$ (Barlett kernel). Then for an integer $m_n<(n-1)/2$ and a vector $c\in \R^v$,
\[
	c^{\top}(\tilde{V}_n-\hat{V}_n)c=2\bar y_n^2\sum_{s=0}^{m_n}\left(1-\frac{s}{m_n+1}\right)=\bar y_n^2(2+m_n)\ge 0,
\]
where $\bar{y}_n=c^{\top}\bar{Y}_n$ and we take $b_n=m_n+1$. Hence, $\tilde V_n-\hat V_n$ is positive semidefinite. In particular, $[\tilde V_n-\hat V_n]_{k,k}\ge 0$ for all $1\le k\le v$ so that the estimator $\hat V_n$ yields lower variances in finite samples.

In addition, it is easy to verify that given the network in Figure \ref{fig:ring_net}(A) and the Barlett kernel, each estimator yields a positive semidefinite covariance matrix. Generally, if the weighting matrix $W\eqdef [\omega_n(d_n(i,j))]_{i,j\in N_n}$ is positive semidefinite, there exists a matrix $L$ with $W=LL^{\top}$ so that
\[
	\tilde V_n=n^{-1}(\tilde YL)(\tilde YL)^{\top} \qtextq{and} \hat V_n=n^{-1}(\hat YL)(\hat YL)^{\top},
\]
where $\tilde Y$ and $\hat Y$ are $d\times n$ matrices whose columns are given by $(Y_{n,i}-\E[Y_{n,i}\mid \C_n])$ and $(Y_{n,i}-\bar Y_n)$, respectively. Hence, both $\tilde V_n$ and $\hat V_n$ are positive semidefinite. Consequently, in a context, in which the distance measure corresponds to the Euclidean norm on $\R^p$, $p\ge 1$, i.e., $d(i,j)=\|x_i-x_j\|$ for some vectors of characteristics $x_i,x_j\in\R^p$, positive definiteness of the kernel function implies that $W$ is positive semidefinite (see, e.g., \citealp{Kelejian/Prucha:07} and \citealp[Chapter~6]{Wendland:04:ScatteredDataApprox}).

This result, however, is not applicable to our case, and positive semidefiniteness of the weighting matrix strongly depends on the network topology. For example, while $W$ is positive semidefinite for the ring network in Figure \ref{fig:ring_net}(A) and the Barlett kernel with $m_n=2$, it becomes indefinite after a slight modification shown in Figure \ref{fig:ring_net}(B).
\end{example}

\subsection{Partially Observed Networks}

In the context of spatial models, \citet{Conley/Molinary:07:JOE} and \citet{Kim/Sun:11:JOE} show that the network HAC estimator can be consistent despite measurement errors in locations. Below, we show that our network HAC estimators have a similar property when the network is only partially observed.

Suppose that the true network is given by $G_n=(N_n,E_n)$. However, the econometrician observes $G^*_n=(N_n,E^*_n)$, where the observed set of links $E^*_n$ is a subset of the true set of links $E_n$. Thus, the links are only partially observed by the econometrician. We continue to use $d_n(i,j)$ to denote the distance between $i$ and $j$ in $G_n$. Let $d^*_n(i,j)$ denote the distance between $i$ and $j$ in $G^*_n$. The immediate consequence of $E^*_n\subset E_n$ is that $d^*_n(i,j)\geq d_n(i,j)$. Hence, because some of the links are unobserved, the true network can be denser than the observed one. The HAC estimator is now defined similarly to \eqref{eq:HAC1}, however, with $\tilde \Omega_n(s)$ replaced by $\tilde \Omega^*_n(s)$, where
\[
	\tilde \Omega^*_n(s) =n^{-1} \sum_{i\in N_n} \sum_{j\in N^{*\partial}_n(i,s)} Y_{n,i}Y_{n,j}^\top,
\]
and $N_n^{*\partial}(i;s)=\{j\in N_n: d^*_n(i,j)=s\}$ is the set of nodes of distance $s$ from $i$ according to the observed network $G^*_n$. We denote the resulting HAC estimator as
\[
	\tilde V^*_n=\sum_{s\ge 0}\omega_n(s)\tilde \Omega^*_n(s).
\]

The implications for the HAC estimator are two-fold: (a) some terms $Y_{n,i}Y_{n,j}^{\top}$ would appear in the covariance term $\tilde{\Omega}^*_n(s)$ with a larger distance $s$ than the true distance; (b) some terms $Y_{n,i}Y_{n,j}^{\top}$ would be missing from the estimator because there is no observed path between $i$ and $j$ in $G^*_n$. The direct consequence of (a) is that such terms would be assigned smaller weights $\omega_n(s)$ compared to those one would assign if the true network was observed. However, since the weights must converge to one according to Assumption \ref{assu:HAC1}(ii), the effect of (a) would be asymptotically negligible provided that the true unobserved network satisfies the rest of the conditions in Assumption \ref{assu:HAC1}. From the expression for $V_n$ in \eqref{eq:trueV}, one can also see that the effect of (b) is asymptotically negligible if the number of missing $Y_{n,i}Y_{n,j}^{\top}$ terms in the HAC estimator is of a smaller order than $n$.

We define $\delta_n^\partial(s\mid d^*_n=\infty)$ as the \emph{average} number of $s$-neighbors that are isolated in the partially observed network:
\[
	\delta_n^\partial(s\mid d^*_n=\infty) = \frac{1}{n} \sum_{i\in N_n} | \{j\in N_n^\partial(i;s): d^*_n(i,j)=\infty\}|.
\]
\begin{assumption}
\label{assu:partial}
	$\lim\sup_{n\to\infty} \sum_{s\ge 1}\delta_n^\partial(s\mid d^*_n=\infty) \theta_{n,s}^{1-(2/p)} =0$ a.s.\ for the same $p>4$ as in Assumption \ref{assu:HAC1}.
\end{assumption}
Assumption \ref{assu:partial} controls the share of nodes that appear isolated due to missing links. For example, the assumption holds if the total number of such nodes is $o_{a.s.}(n)$. For the consistency of the HAC estimator with partially observed networks, we also assume that the true network satisfies the conditions in Assumption \ref{assu:HAC1}. 

In the case of non-zero means, the estimator is defined similarly to \eqref{eq:HAC2}:
\begin{align*}
\hat{V}^*_n &=\sum_{s\ge 0}\omega_n(s)\hat\Omega^*_n(s), \; \text{where}\\
	\hat\Omega^*_n(s)&\eqdef n^{-1}\sum_{i\in N_n}\sum_{j\in N_n^{*\partial}(i;s)}\left(Y_{n,i}-\bar Y_n\right)\left(Y_{n,j}-\bar Y_n\right)^{\top}.
\end{align*}
Let $N_n^{*}(i;s)=\{j\in N_n: d^*_n(i,j)\leq s\}$ be the set of nodes within distance $s$ from $i$ according to the observed network $G^*_n$, and define $\delta_n^*(s)=n^{-1}\sum_{i\in N_n} \vert N_n^{*}(i;s) \vert$. We have the following result.

\begin{prop}\label{cor:HAC_partially_observed}
Suppose that the true network satisfies Assumptions \ref{assu:psi_dep} and  \ref{assu:HAC1}, and Assumption \ref{assu:partial} holds for the partially observed network. Suppose further that $\vert \omega(s)-1\vert \leq \vert \omega(s')-1\vert$ for $s\geq s'$. Then as $n\to\infty$,
\[
	\E[\normin{\tilde{V}^*_n-V_n}_F\mid \C_n]\to_{a.s.} 0.
\]
If, in addition, $\delta_n^*(b_n)=o_{a.s.}(n)$, and $\{\theta_{n,s}/s^{p/(p-4)}\}$ are non-increasing in $s\geq 1$ a.s., then as $n\to\infty$,
\[
	\E[\normin{\hat{V}^*_n-V_n}_F\mid \C_n]\to_{a.s.} 0.
\]
\end{prop}
The monotonicity condition for $\vert \omega(s)-1\vert $ holds, for example, for the truncated, Parzen, and Tukey-Hanning kernels.
\section{Monte Carlo Study}
\label{sec:MonteCarlo}

For our simulation study, we use a version of the network formation model described in Section \ref{sec:NF}. For each sample size $n=500,1000$, and $5000$, we randomly sample $n$ points, $\{X_1,\ldots,X_n\}$ from the uniform distribution on $[0,1]^2$. These points represent the nodes of a random graph $G_n$. Two nodes $i,j\in N_n$ become connected with probability that is inversely proportional to the Euclidean distance between $X_i$ and $X_j$, that is,
\[
	\PM(\{i,j\}\in E_n \mid X_i,X_j)=\exp\left(-\norm{X_i-X_j} \sqrt{2\pi n/\lambda}\right),
\]
where $\lambda$ is a positive constant that determines the average degree of the resulting graph. To reduce the dependence of the results on a particular realization of the latent process $\{X_1,\ldots,X_n\}$, a new network was drawn in each Monte Carlo repetition. We use the values $\lambda$=$1$, $2$, $3$, $4$, and $5$.

To generate $\{Y_{n,i}\}$, we consider a special case of the network dependent processes presented in Section \ref{sec:example_eta}. Specifically, we generate samples using the following linear model:
\begin{equation}
\label{eq:linear_model}
	Y_{n,i}=\sum_{m\ge 0}\frac{\gamma^m}{\absin{N_n^{\partial}(i;m)}}\sum_{j\in N_n^{\partial}(i;m)}\varepsilon_{n,j}, \quad i\in N_n,
\end{equation}
where $\{\varepsilon_{n,i}\}$ are independent $\mathcal{N}(0,1)$ random variables, and $\gamma=0.0$, $0.1$, $0.2$, $0.3$, $0.4$, and $0.5$.

We use the $\hat V_n$ version of the network HAC estimator that does not assume a known mean. To compute the HAC estimator, the bandwidth is chosen according to the rule in \eqref{eq:hac_const} with $\varepsilon=0.05$ and the constant equal to $1.7$, $1.8$, $1.9$, $2.0$, $2.1$, and $2.2$. We use the Parzen kernel given by
\[
	\omega(x)\eqdef\begin{cases}
		1-6x^2+6|x|^3 & \text{for }0\le |x|\le 1/2, \\
		2(1-|x|)^3 & \text{for }1/2<|x|\le 1, \\
		0 & \text{otherwise}.
	\end{cases}
\]

The number of Monte Carlo repetitions is set to 10,000.\footnote{
	The simulations were performed on Compute Canada clusters in Julia using 640 CPUs and 3GB of memory per CPU. The total computation time was 10.5 hours.
} In each Monte Carlo repetition, we compute the average $\bar Y_{n}$, the HAC estimator $\hat V_n$, and construct the 95\% asymptotic confidence interval for the mean of $Y_{n,i}$'s as $\bar Y_n \pm z_{0.975} \times (\hat V_n/n)^{1/2}$, where $z_{0.975} $ is the $0.975$-th percentile of the standard normal distribution.

According to our simulations with the Parzen kernel, setting the bandwidth constant to $2.0$ provides the most accurate coverage in terms of the average squared distance from the nominal coverage probability of $0.95$, where the average is computed across the all considered data generating processes. The distance exhibits a U-shape pattern across the considered constant values.

In Table \ref{tbl:sim_results} we report the simulated coverage probabilities obtained with the bandwidth constant equal to $2.0$. The results for the other bandwidth constants are reported in Appendix \ref{sec:more_simulations} of the Supplemental Note. The Supplemental Note also reports the simulated  rejection probabilities of the corresponding HAC-based $t$-test.

Table \ref{tbl:sim_results} also reports networks statistics such as the diameter, average degree, maximum degree, and average connected distance. One can see from the table that the average degree of the simulated networks is very close to the value of $\lambda$, and that larger values of $\lambda$ correspond to denser networks. Note also that, for example, in the case of $\lambda=3$, $n=1000$, and the bandwidths constant equal to $2.0$, our bandwidth selection rule \eqref{eq:hac_const} produces the bandwidth of approximately $12.58$. In this case, the average simulated diameter is $41.70$, and the average connected distance is $15.89$. Hence, a non-trivial amount of truncation is applied when computing the HAC estimator (except for $\lambda=1$).

Note that when $\lambda=1$, the simulated coverage probabilities do not vary with the constant in the bandwidth selection rule in \eqref{eq:hac_const}, see Tables \ref{tab:coverage_1_5_1_7} and \ref{tab:coverage_1_8_2_2} in the Supplemental Note. As reported in Table \ref{tbl:sim_results}, in that case the simulated average degree is below one, and the bandwidth values resulting from  \eqref{eq:hac_const} exceed the diameters of the simulated networks even for the smallest considered value of the constant. Nevertheless, the coverage of the confidence intervals remains accurate because the networks generated with $\lambda=1$ are sparse with the average connected distance of $2.75-3.01$.

Table  \ref{tbl:sim_results} shows that while in the majority of the cases the simulated coverage probabilities are close to the nominal coverage of $0.95$, the performance of the HAC-based confidence intervals deteriorates for larger values of the denseness parameter $\lambda$ and the dependence parameter $\gamma$. Nevertheless, the coverage improves with the sample size. For example, when $n=5000$ the simulated coverage probabilities are between $0.931-0.949$ even for denser graphs with $\lambda=5$ as long as the dependence parameter $\gamma$ does not exceed $0.3$.

\begin{table}[h!]
  \centering
\footnotesize
  \caption{Average across simulations network statistics and simulated coverage probabilities of the 95\% HAC-based confidence intervals for the bandwidth constant $2.0$ and different values of the denseness parameter $\lambda$, sample size $n$, and the dependence parameter $\gamma$.}
\scalebox{0.95}{
  \begin{threeparttable}
    \begin{tabular}{rrrrrrrrrrrrrrr}
\toprule        &       &       &       &       &       &       &       &     \multicolumn{6}{c}{Simulated Coverage}        &  \\
          &       &       &       &       &       & \multicolumn{1}{c}{Avg.} &       & \multicolumn{6}{c}{$\gamma$} \\
\cmidrule{9-14}    \multicolumn{1}{c}{$\lambda$} & \multicolumn{1}{c}{$n$} &       & \multicolumn{1}{c}{Diam.\tnote{1}} & \multicolumn{1}{c}{Avg.Deg.\tnote{2}} & \multicolumn{1}{c}{Max.Deg.\tnote{3}} & \multicolumn{1}{c}{Dist.\tnote{4}} &       & \multicolumn{1}{c}{0.0} & \multicolumn{1}{c}{0.1} & \multicolumn{1}{c}{0.2} & \multicolumn{1}{c}{0.3} & \multicolumn{1}{c}{0.4} & \multicolumn{1}{c}{0.5} \\
          &       &       &       &       &       &       &       &       &       &       &       &       &  \\
    \midrule
          &       &       &       &       &       &       &       &       &       &       &       &       &  \\
    \multicolumn{1}{c}{1} & \multicolumn{1}{c}{500} &       & \multicolumn{1}{c}{9.68} & \multicolumn{1}{c}{0.95} & \multicolumn{1}{c}{4.99} & \multicolumn{1}{c}{2.75} &       & \multicolumn{1}{c}{0.948} & \multicolumn{1}{c}{0.944} & \multicolumn{1}{c}{0.948} & \multicolumn{1}{c}{0.944} & \multicolumn{1}{c}{0.947} & \multicolumn{1}{c}{0.944} \\
          &       &       & \multicolumn{1}{c}{(2.37)} & \multicolumn{1}{c}{(0.06)} & \multicolumn{1}{c}{(0.77)} & \multicolumn{1}{c}{(0.55)} &       &       &       &       &       &       &  \\
    \multicolumn{1}{c}{1} & \multicolumn{1}{c}{1000} &       & \multicolumn{1}{c}{11.16} & \multicolumn{1}{c}{0.97} & \multicolumn{1}{c}{5.42} & \multicolumn{1}{c}{2.87} &       & \multicolumn{1}{c}{0.946} & \multicolumn{1}{c}{0.949} & \multicolumn{1}{c}{0.947} & \multicolumn{1}{c}{0.947} & \multicolumn{1}{c}{0.949} & \multicolumn{1}{c}{0.944} \\
          &       &       & \multicolumn{1}{c}{(2.43)} & \multicolumn{1}{c}{(0.04)} & \multicolumn{1}{c}{(0.73)} & \multicolumn{1}{c}{(0.47)} &       &       &       &       &       &       &  \\
    \multicolumn{1}{c}{1} & \multicolumn{1}{c}{5000} &       & \multicolumn{1}{c}{14.66} & \multicolumn{1}{c}{0.99} & \multicolumn{1}{c}{6.29} & \multicolumn{1}{c}{3.01} &       & \multicolumn{1}{c}{0.949} & \multicolumn{1}{c}{0.948} & \multicolumn{1}{c}{0.952} & \multicolumn{1}{c}{0.951} & \multicolumn{1}{c}{0.950} & \multicolumn{1}{c}{0.947} \\
          &       &       & \multicolumn{1}{c}{(2.52)} & \multicolumn{1}{c}{(0.02)} & \multicolumn{1}{c}{(0.68)} & \multicolumn{1}{c}{(0.26)} &       &       &       &       &       &       &  \\
          &       &       &       &       &       &       &       &       &       &       &       &       &  \\
    \multicolumn{1}{c}{2} & \multicolumn{1}{c}{500} &       & \multicolumn{1}{c}{30.19} & \multicolumn{1}{c}{1.87} & \multicolumn{1}{c}{7.21} & \multicolumn{1}{c}{10.34} &       & \multicolumn{1}{c}{0.936} & \multicolumn{1}{c}{0.931} & \multicolumn{1}{c}{0.933} & \multicolumn{1}{c}{0.931} & \multicolumn{1}{c}{0.924} & \multicolumn{1}{c}{0.912} \\
          &       &       & \multicolumn{1}{c}{(8.23)} & \multicolumn{1}{c}{(0.09)} & \multicolumn{1}{c}{(0.92)} & \multicolumn{1}{c}{(3.16)} &       &       &       &       &       &       &  \\
    \multicolumn{1}{c}{2} & \multicolumn{1}{c}{1000} &       & \multicolumn{1}{c}{42.01} & \multicolumn{1}{c}{1.91} & \multicolumn{1}{c}{7.73} & \multicolumn{1}{c}{13.89} &       & \multicolumn{1}{c}{0.943} & \multicolumn{1}{c}{0.942} & \multicolumn{1}{c}{0.940} & \multicolumn{1}{c}{0.938} & \multicolumn{1}{c}{0.933} & \multicolumn{1}{c}{0.922} \\
          &       &       & \multicolumn{1}{c}{(11.80)} & \multicolumn{1}{c}{(0.06)} & \multicolumn{1}{c}{(0.89)} & \multicolumn{1}{c}{(4.51)} &       &       &       &       &       &       &  \\
    \multicolumn{1}{c}{2} & \multicolumn{1}{c}{5000} &       & \multicolumn{1}{c}{82.59} & \multicolumn{1}{c}{1.96} & \multicolumn{1}{c}{8.87} & \multicolumn{1}{c}{24.84} &       & \multicolumn{1}{c}{0.947} & \multicolumn{1}{c}{0.947} & \multicolumn{1}{c}{0.945} & \multicolumn{1}{c}{0.946} & \multicolumn{1}{c}{0.944} & \multicolumn{1}{c}{0.938} \\
          &       &       & \multicolumn{1}{c}{(23.28)} & \multicolumn{1}{c}{(0.03)} & \multicolumn{1}{c}{(0.83)} & \multicolumn{1}{c}{(8.47)} &       &       &       &       &       &       &  \\
          &       &       &       &       &       &       &       &       &       &       &       &       &  \\
    \multicolumn{1}{c}{3} & \multicolumn{1}{c}{500} &       & \multicolumn{1}{c}{31.59} & \multicolumn{1}{c}{2.76} & \multicolumn{1}{c}{9.09} & \multicolumn{1}{c}{12.04} &       & \multicolumn{1}{c}{0.936} & \multicolumn{1}{c}{0.926} & \multicolumn{1}{c}{0.926} & \multicolumn{1}{c}{0.918} & \multicolumn{1}{c}{0.903} & \multicolumn{1}{c}{0.869} \\
          &       &       & \multicolumn{1}{c}{(4.97)} & \multicolumn{1}{c}{(0.11)} & \multicolumn{1}{c}{(1.06)} & \multicolumn{1}{c}{(1.46)} &       &       &       &       &       &       &  \\
    \multicolumn{1}{c}{3} & \multicolumn{1}{c}{1000} &       & \multicolumn{1}{c}{41.70} & \multicolumn{1}{c}{2.83} & \multicolumn{1}{c}{9.75} & \multicolumn{1}{c}{15.89} &       & \multicolumn{1}{c}{0.938} & \multicolumn{1}{c}{0.938} & \multicolumn{1}{c}{0.936} & \multicolumn{1}{c}{0.923} & \multicolumn{1}{c}{0.915} & \multicolumn{1}{c}{0.889} \\
          &       &       & \multicolumn{1}{c}{(5.26)} & \multicolumn{1}{c}{(0.08)} & \multicolumn{1}{c}{(1.01)} & \multicolumn{1}{c}{(1.45)} &       &       &       &       &       &       &  \\
    \multicolumn{1}{c}{3} & \multicolumn{1}{c}{5000} &       & \multicolumn{1}{c}{80.48} & \multicolumn{1}{c}{2.93} & \multicolumn{1}{c}{11.06} & \multicolumn{1}{c}{30.65} &       & \multicolumn{1}{c}{0.944} & \multicolumn{1}{c}{0.947} & \multicolumn{1}{c}{0.943} & \multicolumn{1}{c}{0.937} & \multicolumn{1}{c}{0.932} & \multicolumn{1}{c}{0.919} \\
          &       &       & \multicolumn{1}{c}{(5.21)} & \multicolumn{1}{c}{(0.04)} & \multicolumn{1}{c}{(0.94)} & \multicolumn{1}{c}{(1.17)} &       &       &       &       &       &       &  \\
          &       &       &       &       &       &       &       &       &       &       &       &       &  \\
    \multicolumn{1}{c}{4} & \multicolumn{1}{c}{500} &       & \multicolumn{1}{c}{21.85} & \multicolumn{1}{c}{3.64} & \multicolumn{1}{c}{10.78} & \multicolumn{1}{c}{8.56} &       & \multicolumn{1}{c}{0.929} & \multicolumn{1}{c}{0.921} & \multicolumn{1}{c}{0.918} & \multicolumn{1}{c}{0.905} & \multicolumn{1}{c}{0.885} & \multicolumn{1}{c}{0.833} \\
          &       &       & \multicolumn{1}{c}{(2.12)} & \multicolumn{1}{c}{(0.13)} & \multicolumn{1}{c}{(1.15)} & \multicolumn{1}{c}{(0.49)} &       &       &       &       &       &       &  \\
    \multicolumn{1}{c}{4} & \multicolumn{1}{c}{1000} &       & \multicolumn{1}{c}{28.38} & \multicolumn{1}{c}{3.74} & \multicolumn{1}{c}{11.53} & \multicolumn{1}{c}{11.02} &       & \multicolumn{1}{c}{0.936} & \multicolumn{1}{c}{0.931} & \multicolumn{1}{c}{0.927} & \multicolumn{1}{c}{0.923} & \multicolumn{1}{c}{0.900} & \multicolumn{1}{c}{0.860} \\
          &       &       & \multicolumn{1}{c}{(2.13)} & \multicolumn{1}{c}{(0.09)} & \multicolumn{1}{c}{(1.11)} & \multicolumn{1}{c}{(0.45)} &       &       &       &       &       &       &  \\
    \multicolumn{1}{c}{4} & \multicolumn{1}{c}{5000} &       & \multicolumn{1}{c}{55.04} & \multicolumn{1}{c}{3.89} & \multicolumn{1}{c}{13.05} & \multicolumn{1}{c}{21.00} &       & \multicolumn{1}{c}{0.943} & \multicolumn{1}{c}{0.943} & \multicolumn{1}{c}{0.941} & \multicolumn{1}{c}{0.937} & \multicolumn{1}{c}{0.921} & \multicolumn{1}{c}{0.898} \\
          &       &       & \multicolumn{1}{c}{(2.48)} & \multicolumn{1}{c}{(0.04)} & \multicolumn{1}{c}{(1.04)} & \multicolumn{1}{c}{(0.43)} &       &       &       &       &       &       &  \\
          &       &       &       &       &       &       &       &       &       &       &       &       &  \\
    \multicolumn{1}{c}{5} & \multicolumn{1}{c}{500} &       & \multicolumn{1}{c}{17.41} & \multicolumn{1}{c}{4.50} & \multicolumn{1}{c}{12.39} & \multicolumn{1}{c}{6.95} &       & \multicolumn{1}{c}{0.928} & \multicolumn{1}{c}{0.919} & \multicolumn{1}{c}{0.910} & \multicolumn{1}{c}{0.894} & \multicolumn{1}{c}{0.864} & \multicolumn{1}{c}{0.810} \\
          &       &       & \multicolumn{1}{c}{(1.43)} & \multicolumn{1}{c}{(0.15)} & \multicolumn{1}{c}{(1.24)} & \multicolumn{1}{c}{(0.28)} &       &       &       &       &       &       &  \\
    \multicolumn{1}{c}{5} & \multicolumn{1}{c}{1000} &       & \multicolumn{1}{c}{22.54} & \multicolumn{1}{c}{4.64} & \multicolumn{1}{c}{13.24} & \multicolumn{1}{c}{8.87} &       & \multicolumn{1}{c}{0.935} & \multicolumn{1}{c}{0.929} & \multicolumn{1}{c}{0.930} & \multicolumn{1}{c}{0.913} & \multicolumn{1}{c}{0.889} & \multicolumn{1}{c}{0.842} \\
          &       &       & \multicolumn{1}{c}{(1.48)} & \multicolumn{1}{c}{(0.10)} & \multicolumn{1}{c}{(1.20)} & \multicolumn{1}{c}{(0.27)} &       &       &       &       &       &       &  \\
    \multicolumn{1}{c}{5} & \multicolumn{1}{c}{5000} &       & \multicolumn{1}{c}{43.53} & \multicolumn{1}{c}{4.84} & \multicolumn{1}{c}{14.91} & \multicolumn{1}{c}{16.72} &       & \multicolumn{1}{c}{0.949} & \multicolumn{1}{c}{0.942} & \multicolumn{1}{c}{0.942} & \multicolumn{1}{c}{0.931} & \multicolumn{1}{c}{0.920} & \multicolumn{1}{c}{0.885} \\
          &       &       & \multicolumn{1}{c}{(1.75)} & \multicolumn{1}{c}{(0.05)} & \multicolumn{1}{c}{(1.10)} & \multicolumn{1}{c}{(0.28)} &       &       &       &       &       &       &  \\
          &       &       &       &       &       &       &       &       &       &       &       &       &  \\
    \bottomrule
    \end{tabular}%
 \begin{tablenotes}
    \item[1] Diameter: the shortest distance between two most distant nodes.
    \item[2] Average Degree: the average number of adjacent nodes.
    \item[3] Maximum Degree: the maximum number of adjacent nodes.
    \item[4] Average Connected Distance.
     \item Standard deviations in the parentheses.
    \end{tablenotes}
    \end{threeparttable}
    }
  \label{tbl:sim_results}%
\end{table}%
    \clearpage


\putbib[network_dep]

\newpage
\appendix
\section{}

\subsection{Auxiliary Results for \texorpdfstring{$\psi$}{\textpsi}-Dependent Processes}
\label{subsec:covariance_inequalities}

In this section, we present covariance inequalities for functions of general $\psi$-dependent processes. The proofs of the results in this section are found in the Supplemental Note. Let $\F_a$ and $\G_a$ be some classes of functions on $\R^{v\times a}$ with $v,a\ge 1$, and let $\F\eqdef \bigcup_{a\ge 1}\F_a$ and $\G\eqdef \bigcup_{a\ge 1}\G_a$.
\begin{definition}
\label{def:psi_dep_gen}
A triangular array $\{Y_{n,i}\}$, $Y_{n,i} \in \R^v$, is \textit{conditionally $(\F,\G,\psi)$-dependent} given $\{\C_n\}$, if for each $n \in \N$, there exist a $\C_n$-measurable sequence $\theta_n\eqdef\{\theta_{n,s}\}_{s\ge 0}$, $\theta_{n,0}=1$, and a collection of nonrandom functions $(\psi_{a,b})_{a,b \in \N}$, $\psi_{a,b}: \F_a \times \G_b \rightarrow [0,\infty)$, such that for all $(A,B) \in \P_n(a,b;s)$ with $s>0$ and all $f\in \F_a$ and $g\in \G_b$,
\begin{equation}
\label{eq:psi_dep_gen}
	\abs{\Cov\left(f(Y_{n,A}),g(Y_{n,B})\mid \C_n\right)}\le \psi_{a,b} (f,g) \theta_{n,s} \quad \text{a.s.}
\end{equation}
\end{definition}

Let $\{Y_{n,i}\}$ be a triangular array of random vectors in $\R^v$, and let $(A,B)\in \P_n(a,b;s)$ with $s>0$. We first provide a result of a covariance inequality that permits the nonlinear transforms to be random functions. Suppose that $Z_j$, $j=1,2$ is a random element taking values in a separable metric space $(\Z_j,\rho_j)$ equipped with the Borel $\sigma$-algebra $\B(\Z_j)$ and $f$ and $g$ are real-valued, measurable functions defined on $\R^{v\times a}\times \Z_1$ and $\R^{v\times b}\times \Z_2$, respectively. Let $f^z$ be the $z$-section of $f$, i.e., $f^z(y)\eqdef f(y,z)$ (the $z$-section $g^z$ of $g$ is defined similarly) and note that if $f^{z_1}\in \F_a$ and $g^{z_2}\in \G_b$, then $\psi_{a,b}(f^{z_1},g^{z_2})$ is well defined. In addition, let $\bar{f}(y)\eqdef\sup_{z\in \Z_1}\abs{f(y,z)}$ and $\bar{g}(y)\eqdef\sup_{z\in \Z_2}\abs{g(y,z)}$.

\begin{lemma}
\label{lemma:psi_dep}
Suppose that $\{Y_{n,i}\}$ is conditionally $(\F,\G,\psi)$-dependent given $\{\C_n\}$ with the dependence coefficients $\{\theta_n\}$. Suppose further that $f^{z_1}\in \F_a$ and $g^{z_2}\in\G_b$ for all $z_j\in \Z_j$, $f$ and $g$ are continuous in the second arguments, and the function $F(z_1,z_2)\eqdef\psi_{a,b}(f^{z_1},g^{z_2})$ is continuous on $\Z_1\times \Z_2$.\footnote{
	Note that the continuity of $F$ implies that it is Borel measurable. Moreover, if $\Z_j=\Z_{j,1}\times \Z_{j,2}$, $j=1,2$, where each $\Z_{j,k}$ is a separable metric space, the supremum of $F$ taken over $\Z_{1,1}$ and $\Z_{2,1}$ is also Borel measurable. The last observation is essential for other result presented in this section.
} If $Z_1$ and $Z_2$ are $\C_n$-measurable, and $\E[\bar{f}^2(Y_{n,A})+\bar{g}^2(Y_{n,B})]<\infty$ for $(A,B)\in \P_n(a,b;s)$ with $s>0$, then
\[
	\abs{\Cov\left(f(Y_{n,A},Z_1),g(Y_{n,B},Z_2)\mid \C_n\right)}\le F(Z_1,Z_2)\theta_{n,s} \qtext{a.s.}
\]
\end{lemma}

The continuity requirement of the function $F(z_1,z_2)$ in Lemma \ref{lemma:psi_dep} can be relaxed by considering a continuous function $\tilde{F}$ such that for all $(z_1, z_2)\in\Z_1\times \Z_2$, $F(z_1,z_2)\le \tilde{F}(z_1,z_2)$. Consider, for example, the case when $h:\R\to\R$ is piece-wise linear and $f(x,z)=\varphi_z(h(x))$. If the $\psi$ function depends on the Lipschitz constant of $f^{z_1}$ as in Assumption \ref{assu:psi_dep}, then the corresponding $F(z_1, z_2)$ is not continuous in $z_1$. It is clear, however, that the result of Lemma \ref{lemma:psi_dep} holds if we replace $F$ with a continuous dominating function.

\begin{assumption} The triangular array  $\{Y_{n,i}\}$ is conditionally $(\F,\G,\psi)$-dependent given $\{\C_n\}$ with the dependence coefficients $\{\theta_n\}$ satisfying the following conditions.
\label{assu:psi_dep_gen}
\begin{enumerate}[leftmargin=*]
	\item[(a)] $\F$ and $\G$ are stable under multiplication by constants, that is, if $f\in \F$, $g\in \G$ and $c\in \R$, then $cf\in \F$ and $cg\in \G$.
	\item[(b)] If $f\in \F_a, g\in \G_b$ and $c_1,c_2\in \R$, then $\psi_{a,b}(c_1f,c_2g)=\abs{c_1c_2}\cdot \psi_{a,b}(f,g)$.
\end{enumerate}
\end{assumption}

Consider measurable functions $f:\R^{v\times a}\to \R$ and $g:\R^{v\times b}\to \R$ such that $f\notin \F_a$ and $g\notin \G_b$, and define
\begin{equation}
\label{eq:xi_zeta}
	\xi\eqdef f(Y_{n,A}) \qtext{and}\quad \zeta\eqdef g(Y_{n,B}).
\end{equation}
Let $\mu_{\xi,p}\eqdef\norm{\xi}_{\C_n,p}$ and $\mu_{\zeta,p}\eqdef\norm{\zeta}_{\C_n,p}$, $p>0$, and let $\varphi_K$ with $K\in [0,\infty)$ denote the element-wise censoring function, i.e., for an indexed family of real numbers $\mathbf{x}=(x_i)_{i\in I}$,
\begin{equation}
\label{eq:censoring_fn}
	[\varphi_K(\mathbf{x})]_i=(-K)\vee(K\wedge x_i), \quad i\in I,
\end{equation}
where $[A]_i$ denotes the $i$-th element of an indexed family $A$. Finally, we define
\[
	\underline{\theta}_{n,s}=\theta_{n,s}\wedge 1 \qtextq{and} \overline{\theta}_{n,s}=\theta_{n,s}\vee 1.
\]

The following result establishes a bound for the conditional covariance between $\xi$ and $\zeta$ given $\C_n$ in the case in which the censored functions $\varphi_K \circ f$ and $\varphi_L\circ g$, $K,L>0$, belong to the classes $\F$ and $\G$, respectively. The result, therefore, does not require truncation of the domains of the transformation functions. We apply the definition of $\psi$-dependence to the censored counterparts of $f$ and $g$.

\begin{theorem}
\label{thm:cov_ineq1}
Suppose that Assumption \ref{assu:psi_dep_gen} holds, and let $f$, $g$, $\xi$, and $\zeta$ be as in \eqref{eq:xi_zeta}. Suppose further that
\begin{enumerate}[label=\textnormal{(\roman*)},leftmargin=*,align=left]
	\item $\mu_{\xi,p}< \infty$ and $\mu_{\zeta,q}< \infty$ a.s.\ for some $p,q>1$ with $p^{-1}+q^{-1}<1$;
	\item $\varphi_K\circ f\in \F_a$ and $\varphi_K\circ g \in \G_b$ for all $K\in (0,\infty)$;
	\item $(K,L)\mapsto \psi_{a,b}(\varphi_K\circ f, \varphi_L\circ g)$ is continuous on $(0,\infty)^2$.
\end{enumerate}
Then
\begin{equation}
\label{eq:cov_ineq1}
	\absin{\Cov(\xi,\zeta\mid \C_n)}\le \left(\overline{\theta}_{n,s} \overline{\psi}_{a,b}(\mu_{\xi,p},\mu_{\zeta,q})+16\mu_{\xi,p}\mu_{\zeta,q}\right)\underline{\theta}_{n,s}^{1-\frac{1}{p}-\frac{1}{q}} \qtext{a.s.},\footnote{
		Following convention, we take $0 \cdot \infty$ to be $0$ throughout this paper.
	}
\end{equation}
where for $z_1,z_2\in (0,\infty)$,
\[
	\overline{\psi}_{a,b}(z_1,z_2)\eqdef \sup_{K,L\ge 1}(KL)^{-1}\psi_{a,b}(\varphi_{Kz_1}\circ f,\varphi_{Lz_2}\circ g).
\]
\end{theorem}

It is not hard to check that under Assumption \ref{assu:psi_dep_gen} the bound in \eqref{eq:cov_ineq1} preserves the scale-equivariance property because for any $c_1,c_2\in\R$,
\[
	\psi_{a,b}(\varphi_{K c_1 z_1}\circ (c_1 f),\varphi_{Lc_2 z_2}\circ (c_2 g))=\abs{c_1c_2}\psi_{a,b}(\varphi_{Kz_1}\circ f,\varphi_{Lz_2}\circ g).
	\]

\begin{corollary}\label{cor:cov_ineq_simple}
Suppose that the assumptions of Theorem \ref{thm:cov_ineq1} hold. Then
\[
	\absin{\Cov(\xi,\zeta\mid \C_n)}\le (\overline{\theta}_{n,s} \check{\psi}_{a,b} + 16) \mu_{\xi,p}\mu_{\zeta,q}\underline{\theta}_{n,s}^{1-\frac{1}{p}-\frac{1}{q}} \qtext{a.s.},
\]
where
\[
	\check{\psi}_{a,b} \eqdef \sup_{K,L\in (0,\infty)}(KL)^{-1}\psi_{a,b}(\varphi_K\circ f,\varphi_L\circ g).
\]
\end{corollary}

The latter result applies trivially to the strong mixing processes and any measurable functions $f$ and $g$ satisfying relevant moment conditions because $\psi_{a,b}(f,g)=4\norm{f}_{\infty}\norm{g}_{\infty}$. However, for some types of $\psi$-dependence, Condition (ii) of Theorem \ref{thm:cov_ineq1} may not be satisfied. Consider, for example, the case in which $\F=\L_v$ and $f(x,y)=xy$ with $x,y\in \R$. For any $K>0$, the set $\{\absin{f}\le K\}$ is unbounded so that $\varphi_K\circ f$ is not Lipschitz.\footnote{Since $\partial(xy)/\partial x=y$, one can choose $x=0$ so that the function is bounded by any $K>0$, but the partial derivative is unbounded.} To handle such cases we use truncated domains in addition to censoring of the transformation functions.

\begin{theorem}
\label{thm:cov_ineq2}
Suppose that Assumption \ref{assu:psi_dep_gen} holds, and let $f$, $g$, $\xi$, and $\zeta$ be as in \eqref{eq:xi_zeta}. Suppose further that
\begin{enumerate}[label=\textnormal{(\roman*)},leftmargin=*,align=left]
	\item the functions $f$ and $g$ are continuous, and
	\item $\mu_{\xi,p}< \infty$ and $\mu_{\zeta,q}< \infty$ a.s.\ for some $p,q>1$ s.t.\ $p^{-1}+q^{-1}<1$.
\end{enumerate}
Furthermore, there exist increasing continuous functions $h_1,h_2:[0,\infty] \to [0,\infty]$ such that
\begin{enumerate}[resume,label=\textnormal{(\roman*)},leftmargin=*,align=left]
	\item $\gamma_1\eqdef\max_{i\in A}\max_{1\le k\le v}\normin{h_1^{-1}(\absin{[Y_{n,i}]_k})}_{\C_n,p}< \infty$ a.s.\ and \\
	$\gamma_2\eqdef\max_{i\in B}\max_{1\le k\le v}\normin{h_2^{-1}(\absin{[Y_{n,i}]_k})}_{\C_n,q}< \infty$ a.s.;
	\item $f_K\eqdef \varphi_{K_1}\circ f\circ\varphi_{h_1(K_2)}\in \F_a$ and $g_K\eqdef \varphi_{K_1}\circ g\circ\varphi_{h_2(K_2)}\in \G_b$, for all $K\in (0,\infty)^2$;
	\item $(K,L)\mapsto\psi_{a,b}(f_K,g_L)$ is continuous on $(0,\infty)^4$.
\end{enumerate}
Then
\begin{equation}
\label{eq:cov_ineq2}
	\absin{\Cov(\xi,\zeta\mid \C_n)}\le \left(\overline{\theta}_{n,s}\widetilde{\psi}_{a,b}(\mu_{\xi,p},\mu_{\zeta,q},\gamma_1,\gamma_2)+16(abv^2+1)\mu_{\xi,p}\mu_{\zeta,q}\right)\underline{\theta}_{n,s}^{1-\frac{1}{p}-\frac{1}{q}} \qtext{a.s.},
\end{equation}
where for $(z_j,w_j)\in (0,\infty)^2$, $j=1,2$,
\[
	\widetilde{\psi}_{a,b}(z_1,z_2,w_1,w_2)\eqdef \sup_{K,L\ge 1}(KL)^{-1}\psi_{a,b}(f_{(Kz_1,Kw_1)},g_{(Lz_2,Lw_2)}).
\]
\end{theorem}

It can be seen from the proof that when $\varphi_K\circ f\in \F_a$ for all $K>0$ and $g$ satisfies the conditions of Theorem \ref{thm:cov_ineq2}, there is no need to truncate the domain of $f$. In such a case we do not require the continuity of $f$, and the covariance inequality becomes
\[
	\absin{\Cov(\xi,\zeta\mid \C_n)}\le \left(\overline{\theta}_{n,s} \widetilde{\psi}_{a,b}(\mu_{\xi,p},\mu_{\zeta,q},0,\gamma_2)+4(bv+4)\mu_{\xi,p}\mu_{\zeta,q}\right)\underline{\theta}_{n,s}^{1-\frac{1}{p}-\frac{1}{q}} \qtext{a.s.},
\]
where $h_1\equiv \infty$. Similarly, if both $\varphi_K\circ f\in \F_a$ and $\varphi_K\circ g\in \G_b$ for all $K>0$, we are back to the result of Theorem \ref{thm:cov_ineq1}.

Condition (iii) is a moment condition on the original process, where the required moments are defined through the functions $h_1$ and $h_2$. In the special case in which Assumption \ref{assu:psi_dep} holds (i.e., $\F=\G=\L_v$ and the $\psi$ functions are of a certain form), and $f$ and $g$ are the product functions on $\R^{1\times a}$ and $\R^{1\times b}$, respectively, with $a,b\ge 2$, i.e.,
\[
	f(Y_{n,A})=\prod_{i\in A} Y_{n,i} \qtext{and}\quad g(Y_{n,B})=\prod_{i\in B} Y_{n,i},
\]
it suffices to choose $h_1(x)=x^{\frac{1}{a-1}}$ and $h_2(x)=x^{\frac{1}{b-1}}$ in order to guarantee that $\widetilde{\psi}_{a,b}$ is finite valued. Indeed, with this choice of functions $h_1$ and $h_2$ it is not hard to see that $\Lip(f_{(K_1,K_2)})$ and $\Lip(g_{(K_1,K_2)})$ are bounded by $K_2$.

\begin{corollary}
\label{corr:cov_ineq_prod}
Let $\{Y_{n,i}\}$ be an array of random variables satisfying Assumption \ref{assu:psi_dep}, $\xi=\prod_{i\in A} Y_{n,i}$, and $\zeta=\prod_{i\in B} Y_{n,i}$. Let $\{p_i:i\in A\}$ and $\{q_i:i\in B\}$ be collections of positive numbers such that $p^{-1}+q^{-1}<1$, where $p\eqdef\left(\sum_{i\in A}1/p_i\right)^{-1}$ and $q\eqdef\left(\sum_{i\in B}1/q_i\right)^{-1}$. Suppose that
$\norm{Y_{n,i}}_{\C_n,p^*}+\norm{Y_{n,j}}_{\C_n,q^*}<\infty$ a.s.\ for $p^*=\max_{i\in A} p_i$, $q^*=\max_{i\in B}q_i$ and all $i\in A$, $j\in B$. Then
\[
	\absin{\Cov(\xi,\zeta\mid \C_n)}\le 2\overline{\theta}_{n,s}(C+16)\times ab\left(\pi_1+\tilde{\gamma}_1\right)\left(\pi_2+\tilde{\gamma}_2\right)\underline{\theta}_{n,s}^{1-\frac{1}{p}-\frac{1}{q}} \qtext{a.s.},
\]
where
\begin{align*}
	&\pi_1=\prod_{i\in A}\norm{Y_{n,i}}_{\C_n,p_i}, \quad \tilde{\gamma}_1=\max_{i\in A}\norm{Y_{n,i}^{a-1}}_{\C_n,p}, \\
	&\pi_2=\prod_{i\in B}\norm{Y_{n,i}}_{\C_n,q_i}, \quad \tilde{\gamma}_2=\max_{i\in B}\norm{Y_{n,i}^{b-1}}_{\C_n,q},
\end{align*}
and $C>0$ is the constant in Assumption \ref{assu:psi_dep}.
\end{corollary}

\subsection{Proofs of the Main Results}

\begin{proof}[\textnbf{Proof of Proposition \ref{prop:example_eta}}]
Let $\xi\eqdef f(Y_{n,A})$, $\zeta\eqdef g(Y_{n,B})$,
\[
	\xi^{(s)}\eqdef f(Y_{n,i}^{(s)}: i \in A), \qtext{and} \quad \zeta^{(s)}\eqdef g(Y_{n,i}^{(s)}: i \in B).
\]
Then, since $\xi^{(s)}$ and $\zeta^{(s)}$ are conditionally independent given $\C_n$, we find that
\begin{align*}
	\absin{\Cov(\xi,\zeta\mid \C_n)}&\le \absin{\Cov((\xi-\xi^{(s)}),\zeta\mid \C_n)}+\absin{\Cov(\xi^{(s)},(\zeta-\zeta^{(s)})\mid \C_n)} \\[1em]
	&\le 2\norm{g}_{\infty}\E[\absin{\xi-\xi^{(s)}}\mid \C_n]+ 2\norm{f}_{\infty}\E[\absin{\zeta-\zeta^{(s)}}\mid \C_n] \\[1em]
	&\le 2\norm{g}_{\infty}\Lip(f)\sum_{i\in A}\E[\normin{Y_{n,i}-Y_{n,i}^{(s)}}\mid \C_n] \\
	&\quad+ 2\norm{f}_{\infty}\Lip(g)\sum_{i\in B}\E[\normin{Y_{n,i}-Y_{n,i}^{(s)}}\mid \C_n] \\
	&\le \big(a\norm{g}_{\infty}\Lip(f)+b\norm{f}_{\infty}\Lip(g)\big)\times \theta_{n,s} \qtext{a.s.}
\qedhere
\end{align*}
\end{proof}

\begin{proof}[\textnbf{Proof of Proposition \ref{prop:example_eta2}}]
Define $\xi$, $\zeta$, $\xi^{(s)}$, and $\zeta^{(s)}$ as in the proof of Proposition \ref{prop:example_eta}. With the choice of $(A,B) \in \P_n(a,b;3s)$,
\begin{align}
	\begin{aligned}\label{eq:ext}
		\absin{\Cov(\xi,\zeta\mid \C_n)}&\le \absin{\Cov((\xi-\xi^{(s)}),\zeta\mid \C_n)}+\absin{\Cov(\xi^{(s)},(\zeta-\zeta^{(s)})\mid \C_n)} \\
		&\quad + \absin{\Cov(\xi^{(s)},\zeta^{(s)} \mid \C_n)} \qtext{a.s.}
	\end{aligned}
\end{align}
The first two terms on the right hand side are bounded by
\[
	2 \big(a\norm{g}_{\infty}\Lip(f)+b\norm{f}_{\infty}\Lip(g)\big)\times \max_{i \in N_n} \E[ \| Y_{n,i} - Y_{n,i}^{(s)}|\C_n ] \qtext{a.s.}
\]
from the proof of Proposition \ref{prop:example_eta}. We focus on the last conditional covariance in \eqref{eq:ext}. Define
\[
	N_n(A;s) = \bigcup_{i \in A} N_n(i;s).
\]
By the $\psi$-dependence of $\varepsilon_{n,i}$, for any $(A,B) \in \P_n(a,b;3s)$, and any $f \in \L_{k,| N_{n}(A;s)|}$ and $g \in \L_{k,| N_{n}(B;s)|}$,
\begin{equation}
\label{eq:bd32}
	\abs{\Cov ( f(\varepsilon_{n, N_{n}(A;s)}), f(\varepsilon_{n, N_{n}(B;s)})\mid \C_n )} \le \psi_{| N_{n}(A;s)|,| N_{n}(B;s)|}^\varepsilon(f,g) \theta_{n,s}^\varepsilon.\qtext{a.s.}
\end{equation}

Now, let $f,g$ be as given in the proposition. Let $f_1,g_1$ be maps such that
\begin{align*}
	f_1(\varepsilon_{n, N_{n}(A;s)}) &= f\left( (\bm \phi_{n,i}(\varepsilon_{n}^{(s,i)}))_{i \in A} \right), \text{ and } \\
	g_1(\varepsilon_{n, N_{n}(A;s)}) &= g\left( (\bm \phi_{n,i}(\varepsilon_{n}^{(s,i)}))_{i \in B} \right).
\end{align*}
Then by \eqref{eq:bd32} and Assumption \ref{assu:psi_dep}(a), we find that
\begin{align*}
	\abs{\Cov\left( f(Y_{n,A}^{(s)}), g(Y_{n,B}^{(s)}) \mid \C_n \right)}
	&\le \psi_{| N_{n}(A;s)|,|\overline N_{n}(B;s)|}^\varepsilon(f_1,g_1) \theta_{n,s}^\varepsilon\\
	&\le C| N_n(A;s)| | N_n(B;s)| \\
	&\quad\times \left( \|f_1\|_\infty + \Lip(f_1)\right)\left( \|g_1\|_\infty + \Lip(g_1)\right) \theta_{n,s}^\varepsilon \qtext{a.s.}
\end{align*}
Note that $ \|f_1\|_\infty = \|f\|_\infty$ and $ \|g_1\|_\infty = \|g\|_\infty$, and
\begin{align*}
	|N_n(A;s)||N_n(B;s)| \le D_n^{2}(s) a b.
\end{align*}
Let us compute $\Lip(f_1)$. For $x_n^{(s,i)},\tilde x_n^{(s,i)} \in \R^{n \times k}$ such that their $j$-th rows $x_{n,j}$ and $\tilde x_{n,j}$ are equal to zero for all $j \in N_n \setminus N_n(i;s)$. Then by the definition of $\mathss{d}_a$ in \eqref{eq:distance},
\begin{align*}
	&\left| f\left( (\bm \phi_{n,i}(x_n^{(s,i)}))_{i \in A} \right) - f\left( (\bm \phi_{n,i}(\tilde x_n^{(s,i)}))_{i \in A} \right) \right| \\[1em]
	&\qquad\le \Lip(f) \mathss{d}_a\left((\bm \phi_{n,i}(x_n^{(s,i)}))_{i \in A},(\bm \phi_{n,i}(\tilde x_n^{(s,i)}))_{i \in A} \right).
\end{align*}
The last term is equal to
\begin{align*}
	&\Lip(f) \sum_{i \in A} \sqrt{\sum_{r=1}^k (\phi_{n,ir}(x_n^{(s,i)}) - \phi_{n,ir}(\tilde x_n^{(s,i)}))^2}\\
	&\qquad\le \sqrt{k} \Lip(f) \max_{i \in N_n} \sum_{r=1}^{k} \sum_{i \in A} |\phi_{n,ir}(x_n^{(s,i)}) - \phi_{n,ir}(\tilde x_n^{(s,i)})|\\
	&\qquad\le \sqrt{k} \Lip(f) \max_{i \in N_n} \sum_{r=1}^{k} \Lip(\phi_{n,ir})\sum_{i \in N_n} \|x_n^{(s,i)} - \tilde x_n^{(s,i)}\|\\
	&\qquad\le \sqrt{k} \Lip(f) \max_{i \in N_n} \sum_{r=1}^{k} \Lip(\phi_{n,ir})\mathss{d}_n(x_n^{(s)},\tilde x_n^{(s)}),
\end{align*}
where $x_n^{s}, \tilde x_n^{(s)} \in \R^{n \times k}$ are matrices whose $i$-th row is given by $x_n^{s,i}$ and $\tilde x_n^{s,i}$. Hence, we can take
\[
	\Lip(f_1) = \sqrt{k} \Lip(f) \bar \phi.
\]
Similarly, $\Lip(g_1) = \sqrt{k}\Lip(g) \bar \phi$. We find that
\begin{align*}
	&\abs{\Cov\left( f(Y_{n,A}^{(s)}), g(Y_{n,B}^{(s)}) \mid \C_n \right)} \\
	&\qquad\le C\times ab \left( \|f\|_\infty + \sqrt{k}\Lip(f) \bar \phi \right)\left( \|g\|_\infty + \sqrt{k}\Lip(g)\bar \phi \right) D_n^{2}(s) \theta_{n,s}^\varepsilon \qtext{a.s.}
\end{align*}
Combining this with the bounds for the two terms on the right hand side of \eqref{eq:ext}, we obtain the desired result.
\end{proof}

\begin{proof}[\textnbf{Proof of Theorem \ref{thm:pointwise_LLN}}]
We follow the approach of \cite{Jenish/Prucha:09:JOE}, see the proof of Theorem 3 therein. However, instead of the truncation used therein, we rely on censoring functions $\varphi_k(x)$ defined in \eqref{eq:censoring_fn} in order to be able to use the notion of $\psi$-dependence.\footnote{Unlike discontinuous truncation functions $x\cdot 1\{|x|\leq k\}$, censoring functions $\varphi_k(x)$ are continuous and have a finite Lipschitz constant: $\Lip(\varphi_k)=1$.}
Consider a censored version of $Y_{n,i}$: for some $k>0$, let
\begin{align*}
	Y_{n,i} & =Y_{n,i}^{(k)}+\tilde{Y}_{n,i}^{(k)},\mbox{ where}\\
	Y_{n,i}^{(k)} & =\varphi_k(Y_{n,i}),\\
	\tilde{Y}_{n,i}^{(k)} & =Y_{n,i}-\varphi_k(Y_{n,i})= (Y_{n,i}-\sgn(Y_{n,i})k)1\left(|Y_{n,i}|>k\right).
\end{align*}
We have
\begin{align}
	\norm{\frac{1}{n}\sum_{i \in N_n}\left(Y_{n,i}-\E[Y_{n,i} \mid \C_n ]\right)}_{\C_n,1}
	&\leq \norm{\frac{1}{n}\sum_{i \in N_n}\left(Y_{n,i}^{(k)}-\E[Y_{n,i}^{(k)} \mid \C_n ]\right)}_{\C_n,1} \notag \\
	&\quad +\norm{\frac{1}{n}\sum_{i \in N_n}\left(\tilde Y_{n,i}^{(k)}-\E[\tilde Y_{n,i}^{(k)} \mid \C_n ]\right)}_{\C_n,1}. \label{eq:LLNtail}
\end{align}
Note that $\E[\vert\tilde Y_{n,i}^{(k)}\vert \mid \C_n]=\E[\vert\tilde Y_{n,i}^{(k)}\vert 1\left(\vert Y_{n,i}\vert >k\right) \mid \C_n]\leq 2\E[|Y_{n,i}|1\{|Y_{n,i}|>k\}\mid \C_n]$, which holds since $\tilde Y_{n,i}^{(k)}=0$ when $|Y_{k,i}|\leq k$. Hence, by the triangle inequality, the second term on the right-hand side of \eqref{eq:LLNtail} is bounded by $2 \sup_{n\geq 1}\max_{i\in N_n} \E[\vert\tilde Y_{n,i}^{(k)}\vert \mid \C_n]\leq
4\sup_{n\geq 1}\max_{i\in N_n} \E[\vert Y_{n,i}\vert 1\{\vert Y_{n,i}\vert >k\} \mid \C_n]$.

By Assumption \ref{assu:Uniform_L1_Integrability} and using the same arguments as in \citet[Theorem 12.10]{Davidson:94:StochasticLimitTheory}, $\lim_{k\to\infty}\sup_{n\geq 1}\max_{i\in N_n} \E[\vert Y_{n,i}\vert 1\{\vert Y_{n,i}\vert >k\} \mid \C_n]=0$ a.s.

By the norm inequality,
\begin{align}
	\norm{\frac{1}{n}\sum_{i \in N_n}\left(Y_{n,i}^{(k)}-\E[Y_{n,i}^{(k)} \mid \C_n ]\right)}_{\C_n,1}&\leq \frac{\sigma_{n,k}}{n},\label{eq:Bound_sum_X^(k)}
\end{align}
where
\begin{align*}
	\sigma_{n,k}^{2} &= \E\left[\left(\sum_{i \in N_n}\left(Y_{n,i}^{(k)}-\E[Y_{n,i}^{(k)} \mid \C_n ]\right)\right)^{2}\mid \C_n\right]\\
	&= \sum_{i \in N_n}\sum_{j=1}^n\E\left[(Y_{n,i}^{(k)}-\E[Y_{n,i}^{(k)} \mid \C_n ])(Y_{n,j}^{(k)}-\E[Y_{n,j}^{(k)} \mid \C_n ])\mid C_n \right]\\
	&\leq 4nk^2+\sum_{i \in N_n}\sum_{s\ge 1}\sum_{j\in N_n^{\partial}(i;s)}\left|\Cov(Y_{n,i}^{(k)},Y_{n,j}^{(k)} \mid \C_n )\right|.
\end{align*}
In view of Definition \ref{def:psi_dep} and Assumption \ref{assu:psi_dep}(a),
we have for $d_n(i,j)=s$,
\begin{equation}\label{eq:HS_covariance}
	\abs{\Cov(Y_{n,i}^{(k)},Y_{n,j}^{(k)} \mid \C_n )}
	\leq \psi_{1,1}(\varphi_k,\varphi_k) \cdot \theta_{n,s} \quad \text{a.s.},
\end{equation}
where $\varphi_k$ is bounded and $\text{Lip}(\psi_k)=1$.
Using the definitions of $N_n^\partial(i;s)$ and $\delta_n^\partial(s)$, we obtain that $\sigma_{n,k}^{2}$ is bounded by
\begin{align}\label{eq:Bound_sigma^2_(n,k)}
	\begin{aligned}
	 4nk^2+\psi_{1,1}(\varphi_k,\varphi_k) \sum_{s\ge 1} \theta_{n,s}\sum_{i \in N_n}|N^\partial_n(i;s)|
	= n\left(4k^2+\psi_{1,1}(\varphi_k,\varphi_k) \sum_{s\ge 1}\delta_n^\partial(s) \theta_{n,s}\right).
	\end{aligned}
\end{align}
By \eqref{eq:Bound_sum_X^(k)} and \eqref{eq:Bound_sigma^2_(n,k)},
\[
	\norm{\frac{1}{n}\sum_{i \in N_n}\left(Y_{n,i}^{(k)}-\E[Y_{n,i}^{(k)} \mid C_n]\right)}_{\C_n,1}\leq \left(\frac{4k^2}{n} +\psi_{1,1}(\varphi_k,\varphi_k) \frac{1}{n}\sum_{s\ge 1}\delta^\partial_n(s) \theta_{n,s} \right)^{1/2}.
\]
The result now follows from Assumption \ref{assu:LLN_suff_delta}.
\end{proof}
For each $p \ge 1$,
\[
	\mu_{n,p} \eqdef 1\vee \max_{i \in N_n} \left(\E[|Y_{n,i}|^p \mid \C_n]\right)^{1/p}, \quad \overline{\theta}_n\eqdef\max_{s\ge 1}\overline{\theta}_{n,s}, \qtextq{and} \tilde S_n \eqdef \sum_{j \in N_n} \tilde Y_{n,j},
\]
where $\tilde Y_{n,j} = Y_{n,j}/\sigma_n$, assuming that $\sigma_n>0$ a.s. Define
\begin{equation}
\label{eq:H}
	\quad H_n(s,m) = \left\{(i,j,k,l) \in N_n^4: j \in N_n(i;m),l \in N_n(k;m), d_n(\{i,j\},\{k,l\}) = s \right\}.
\end{equation}
Then note that
\begin{align*}
	|H_n(s,m)|&\le 4\sum_{i\in N_n}\sum_{j\in N_n^{\partial}(i;s)}\abs{N_n(i;m)\setminus N_n(j;s - 1)}^2\\
	&\le 4\sum_{i\in N_n}\max_{j\in N_n^{\partial}(i; s)}\abs{N_n(i;m)\setminus N_n(j;s - 1)}^2\times |N_n^{\partial}(i;s)| \\
	&\le 4n\left[\Delta_n(s,m;2\alpha)\right]^{\frac{1}{\alpha}}\left[\delta_n^{\partial}(s;\alpha/(1-\alpha))\right]^{1-\frac{1}{\alpha}},
\end{align*}
by H\"{o}lder's inequality, where $\alpha \in (1,\infty)$. Recall the definition of $\Delta_n(s,m;k)$ and $c_n(s,m;k) $ in \eqref{eq:Delta_n} and \eqref{eq:c_n}. By taking the infimum over $\alpha>1$ of the last term, we obtain that
\begin{equation}
\label{eq:H_bound}
	|H_n(s,m)| \le 4 n c_n(s,m;2).
\end{equation}

The following lemma is used for the central limit theorem.

\begin{lemma}
	\label{lemma:CLT}
	Suppose that Assumption \ref{assu:psi_dep} holds, and that $\E[Y_{n,i}\mid \C_n]=0$ a.s. Let $g: \R \rightarrow \R$ be a twice continuously differentiable bounded function with bounded derivatives. Then there exists a constant $C>0$ such that for any positive integer $m_n$ and any $n \ge 1$,
	\begin{equation}
		|\E[g'(\tilde S_n)-\tilde S_n g(\tilde S_n)\mid \C_n]| \le \bar \Delta_n(g) + \frac{C n \overline{\theta}_n \| g''\|_\infty\mu_{n,p}^3}{\sigma_n^3} \sum_{s\ge 0} c_n(s,m_n;1) \theta_{n,s}^{1-(3/p)} \qtext{a.s.},
	\end{equation}
	where
	\begin{align*}
		\bar \Delta_n(g)&= \frac{C\sqrt{n\overline{\theta}_n}\|g'\|_\infty \mu_{n,p}^2}{\sigma_n^2}\sqrt{\sum_{s\ge 0}c_n(s,m_n;2) \theta_{n,s}^{1-(4/p)}}\\
		&\quad + C n^2\overline{\theta}_n\left( \frac{\|g'\|_\infty \mu_{n,p}^2}{\sigma_n^2} \theta_{n,m_n}^{1-(2/p)} + \frac{\left( \|g\|_\infty + \sigma_n^{-1}\|g'\|_\infty\right)\mu_{n,p}}{\sigma_n} \theta_{n,m_n}^{1-(1/p)}\right).
	\end{align*}
\end{lemma}

\begin{proof}[\textnbf{Proof}]
We set an increasing sequence of positive integers $m_n$, and define for each $i \in N_n$,
\[
	\tilde S_{n,i} = \sum_{j \in N_n \setminus N_n(i;m_n)} \tilde Y_{n,j}.
\]
We write
\[
	g'(\tilde S_n) - \tilde S_n g(\tilde S_n) = A_{n,1} + A_{n,2} + A_{n,3},
\]
where
\begin{align*}
	A_{n,1} &= g'(\tilde S_n)\left(1- \sum_{i \in N_n} \tilde Y_{n,i} ( \tilde S_n - \tilde S_{n,i} ) \right) \\
	A_{n,2} &= \sum_{i \in N_n} \tilde Y_{n,i} \left(g'(\tilde S_n)(\tilde S_n - \tilde S_{n,i}) - (g(\tilde S_n) - g(\tilde S_{n,i})) \right), \qtext{and}\\
	A_{n,3} &= -\sum_{i \in N_n} \tilde Y_{n,i} g(\tilde S_{n,i}).
\end{align*}
Let us turn to $A_{n,2}$. Applying Taylor expansion,
\[
	|\E[A_{n,2}\mid \C_n]| \le \frac{\| g'' \|_\infty}{2} \sum_{i\in N_n} \E[|\tilde Y_{n,i}| (\tilde S_{n,i}-\tilde S_n)^2\mid \C_n].
\]
The last bound is bounded by
\begin{align*}
	& \frac{\| g'' \|_\infty}{2} \sum_{i\in N_n}\sum_{j\in N_n(i;m_n)}\sum_{k\in N_n(i;m_n)} \E[|\tilde Y_{n,i}| \tilde Y_{n,j} \tilde Y_{n,k} \mid \C_n]\\
	&\qquad= \frac{\| g'' \|_\infty}{2} \sum_{s\ge 0} \sum_{i\in N_n}\sum_{j\in N_n(i;m_n)}\sum_{k\in N_n(i;m_n):d_n(k,\{i,j\})=s} \E[|\tilde Y_{n,i}| \tilde Y_{n,j} \tilde Y_{n,k} \mid \C_n]
\end{align*}
Now, since $\E[\tilde Y_{n,k} \mid \C_n] = 0$, by Corollary \ref{corr:cov_ineq_prod} the last sum is bounded by
\begin{align*}
	&\frac{C\overline{\theta}_n\mu_{n,p}^3}{\sigma_n^3} \sum_{s\ge 0} \sum_{i\in N_n}\sum_{k \in N_n^\partial(i;s)}|N_n(i;m_n) \setminus N_n(k;s-1)| \theta_{n,s}^{1-(3/p)}\\
	&\qquad\le \frac{C\overline{\theta}_n\mu_{n,p}^3}{\sigma_n^3} \sum_{s\ge 0} \sum_{i\in N_n}\max_{k \in N_n^\partial(i;s)}|N_n(i;m_n) \setminus N_n(k;s-1)||N_n^\partial(i;s)|^{1-(3/p)}.
\end{align*}
Using H\"{o}lder's inequality and the definitions of $\Delta_n(s,m;k)$, $\delta_n^\partial(s)$, and $c_n(s,m;k)$ in \eqref{eq:delta_n}, \eqref{eq:Delta_n}, and \eqref{eq:c_n}, we can bound the last term by
\[
	\frac{Cn\overline{\theta}_n\mu_{n,p}^3}{\sigma_n^3} \sum_{s\ge 0}c_n(s,m_n;1) \theta_{n,s}^{1-(3/p)}.
\]
Thus, it follows that
\[
	|\E[A_{n,2}\mid \C_n]|\le \frac{ Cn \overline{\theta}_n \| g'' \|_\infty \mu_{n,p}^3}{\sigma_n^3} \sum_{s\ge 0} c_n(s,m_n;1) \theta_{n,s}^{1-(3/p)}.
\]
Let us now turn to $A_{n,1}$. Write
\[
	|\E[A_{n,1}\mid \C_n]| = \E\left[g'(\tilde S_n)\left(1-\sum_{i\in N_n}\sum_{j\in N_n(i;m_n)} \tilde Y_{n,i} \tilde Y_{n,j} \right)\mid \C_n\right].
\]
Since, by the definition of $\sigma_n$,
\[
	1 = \E\left[\sum_{i\in N_n}\sum_{j\in N_n} \tilde Y_{n,i} \tilde Y_{n,j}\mid \C_n\right],
\]
we rewrite
\begin{align*}
	\E[A_{n,1}\mid \C_n] &= - \E\left[g'(\tilde S_n)\left(\sum_{i\in N_n}\sum_{j\in N_n(i;m_n)} (\tilde Y_{n,i} \tilde Y_{n,j} - \E[\tilde Y_{n,i} \tilde Y_{n,j}\mid \C_n]) \right)\mid \C_n\right]\\
	&\quad + \E\left[g'(\tilde S_n)\mid \C_n\right]\left(\sum_{i\in N_n}\sum_{j\in N_n \setminus N_n(i;m_n)} \E[\tilde Y_{n,i} \tilde Y_{n,j}\mid \C_n] \right) = R_{n,1} + R_{n,2}, \textnormal{ say.}
\end{align*}
Using the Cauchy-Schwarz inequality and letting $Z_{n,ij} = \tilde Y_{n,i} \tilde Y_{n,j} - \E[\tilde Y_{n,i} \tilde Y_{n,j}\mid \C_n]$, we bound $|R_{n,1}|$ by
\[
	\sqrt{\E[(g'(\tilde S_n))^2\mid \C_n]} \sqrt{\E\left[\left(\sum_{i\in N_n}\sum_{j\in N_n(i;m_n)} Z_{n,ij} \right)^2\mid \C_n\right]}.
\]

Let us write the last conditional expectation in the preceding display as
\begin{align*}
	& \sum_{i\in N_n}\sum_{j\in N_n(i;m_n)}\sum_{k\in N_n}\sum_{l\in N_n(k;m_n)} \E[Z_{n,ij} Z_{n,kl} \mid \C_n] \\
	&\qquad= \sum_{s\ge 0} \sum_{(i,j,k,l) \in H_n(s,m_n)} \E[Z_{n,ij} Z_{n,kl} \mid \C_n]\\
	&\qquad\le \frac{C\overline{\theta}_n\mu_{n,p}^4}{\sigma_n^4} \sum_{s\ge 0}|H_n(s,m_n)| \theta_{n,m}^{1-(4/p)},
\end{align*}
by Corollary \ref{corr:cov_ineq_prod}, recalling the definition of $H_n(s,m_n)$ in \eqref{eq:H}. In combination with \eqref{eq:H_bound}, this gives the following bound for $|R_{n,1}|$:
\[
	\frac{C\sqrt{n\overline{\theta}_n} \|g'\|_\infty \mu_{n,p}^2}{\sigma_n^2}\sqrt{\sum_{s\ge 0}c_n(s,m_n;2) \theta_{n,s}^{1-(4/p)}}.
\]
Let us turn to $R_{n,2}$. We bound $|R_{n,2}|$ by
\[
	\frac{Cn^2 \overline{\theta}_n \|g'\|_\infty \mu_{n,p}^2 }{\sigma_n^2}\theta_{n,m_n}^{1-(2/p)},
\]
using Theorem \ref{thm:cov_ineq1}. Finally, let us consider $A_{n,3}$. Note that
\begin{align}
\label{eq:bound}
	\E[A_{n,3}|\C_n] \le \sum_{i \in N_n} \left|\Cov\left(\tilde Y_{n,i}, g(\tilde S_{n,i})\mid \C_n\right)\right|.
\end{align}
Using Theorem \ref{thm:cov_ineq1} and Assumption \ref{assu:psi_dep}, we bound the term on the right hand side of \eqref{eq:bound} by
\begin{align*}
	\frac{Cn^2\overline{\theta}_n\left( \|g\|_\infty +\sigma_n^{-1}\|g'\|_\infty\right)\mu_{n,p}}{\sigma_n} \theta_{n,m_n}^{1-(1/p) - (1/q)},
\end{align*}
where $q$ is such that $p^{-1} + q^{-1} < 1$. By taking $q \rightarrow \infty$, we conclude that
\[
	|\E[A_{n,3}\mid \C_n]| \le \frac{Cn^2\overline{\theta}_n\left( \|g\|_\infty + \sigma_n^{-1}\|g'\|_\infty\right)\mu_{n,p}}{\sigma_n}\theta_{n,m_n}^{1-(1/p)}.
\]
Collecting the results, we obtain the desired result of the lemma.
\end{proof}

\begin{lemma}
	\label{lemma:Berry-Esseen}
	Suppose that Assumption \ref{assu:psi_dep} holds, and that $\E[Y_{n,i}\mid \C_n]=0$ a.s. Then there exists $C>0$ such that for all $n \ge 1$,
	\begin{align*}
		 \sup_{t \in \R}\,\left|\PM\left\{\tilde S_n \le t\mid \C_n\right\} - \Phi(t)\right|&\le C \sqrt{n\overline{\theta}_n} \left(\frac{\mu_{n,p}}{\sigma_n}\right)^{3/2} \sqrt{\sum_{s\ge 0} c_n(s,m_n;1) \theta_{n,s}^{1-(3/p)}} \\
		&\quad+ C \sqrt{n\overline{\theta}_n} \left(\frac{\mu_{n,p}}{\sigma_n}\right)^2\sqrt{\sum_{s\ge 0} c_n(s,m_n;2) \theta_{n,s}^{1-(4/p)}}\\
		&\quad + C n^2 \overline{\theta}_n\left(\frac{ \mu_{n,p}^2}{\sigma_n^2} \theta_{n,m_n}^{1-(2/p)} + \frac{\mu_{n,p}}{\sigma_n}\left(1+\frac{1}{\sigma_n}\right) \theta_{n,m_n}^{1-(1/p)}\right) \qtext{a.s.},
	\end{align*}
	where $\Phi$ denotes the distribution function of $\mathcal{N}(0,1)$.
\end{lemma}

\begin{proof}[\textnbf{Proof}]
The proof is an adaptation of the proof of Theorem 2.4 of \cite{Penrose:03:RandomGeometricGraphs} to our set-up. Let $\bar \Delta_n(g)$ be as defined in Lemma \ref{lemma:CLT}. Let us define $h_+(x)=1$ for $x \le t$, $h_+(x)=0$ for $x \ge t + \varepsilon$, and $h_+$ is continuous and linear on $[x,x+\varepsilon]$. Similarly, we also take $h_-(x)=1$ for $x \le t - \varepsilon$, $h_-(x)=0$ for $x \ge t$, and $h_-$ is continuous and linear on $[x-\varepsilon,x]$. Define for any real function $g$,
\[
	\Delta_n(g) = |\E[g(\tilde S_n)\mid \C_n]-\E[g(Z)]|.
\]
Let us find a bound for $\Delta_n(h_+)$ and $\Delta_n(h_-)$. First, note that by Stein's Lemma \citep[e.g.,][p.~15]{Chen/Goldstein/Shao:11:NormalApprox}, for any real valued function $h$ with $\mathbf{E}|h(Z)|<\infty$,
\begin{equation}
\label{eq:stein}
	|\E[g'(\tilde S_n)-\tilde S_ng(\tilde S_n)\mid \C_n]| = \Delta_n(h),
\end{equation}
where
\[
	g(x) = e^{x^2/2}\int_{-\infty}^x (h(w) - \E[h(Z)]) e^{-w^2/2} dw.
\]
Since for $h = h _+$ or $h = h_-$, \citep[see Lemma 2.4 of][]{Chen/Goldstein/Shao:11:NormalApprox}
\begin{equation}
\label{eq:inequalities}
	\begin{aligned}
		\|g\|_\infty &\le \sqrt{\pi/2}\,\|h - \E[h(Z)]\|_\infty \le \sqrt{\pi/2}, \\
		\|g'\|_\infty &\le 2\,\|h - \E[h(Z)]\|_\infty \le 2, \qtext{and} \\
		\|g''\|_\infty &\le 2\,\|h'\|_\infty \le 2/\varepsilon,
	\end{aligned}
\end{equation}
we apply Lemma \ref{lemma:CLT} to \eqref{eq:stein} to deduce that for $h = h_+$ or $h = h_-$,
\[
	\Delta_n(h) \le \bar \Delta_n(g) + \frac{Cn\overline{\theta}_n\| g''\|_\infty \mu_{n,p}^3}{\varepsilon\sigma_n^3} \sum_{s\ge 0}c_n(s,m_n;1) \theta_{n,s}^{1-(3/p)}.
\]
Let us now bound
\begin{align*}
	\PM\{\tilde S_n \le t\mid \C_n\} \le \E[h_+(\tilde S_n)\mid \C_n] &\le \E[h_+(Z)] + \Delta_n(h_+)\\
	&\le \PM\{Z \le t + \varepsilon\} +\Delta_n(h_+)\\
	&\le \PM\{Z \le t\} + \phi(0)\varepsilon + \Delta_n(h_+),
\end{align*}
where $\phi$ is the density of $\mathcal{N}(0,1)$. Similarly, we also bound
\[
	\PM\{\tilde S_n \le t\mid \C_n\} \ge \PM\{Z \le t\} - \phi(0) \varepsilon - \Delta_n(h_-).
\]
Hence, we have
\begin{align*}
	&|\PM\{\tilde S_n \le t\mid \C_n\} - \PM\{Z \le t\}| \\
	&\qquad\le 2 \phi(0) \varepsilon + \frac{Cn\overline{\theta}_n\| g''\|_\infty \mu_{n,p}^3}{\varepsilon\sigma_n^3} \sum_{s\ge 0} c_n(s,m_n;1) \theta_{n,s}^{1-(3/p)} + \bar \Delta_n(g).
\end{align*}
Choose
\[
	\varepsilon = \left(\frac{C n \overline{\theta}_n \| g''\|_\infty \mu_{n,p}^3}{2 \phi(0)\sigma_n^3} \sum_{s\ge 0} c_n(s,m_n;1) \theta_{n,s}^{1-(3/p)} \right)^{1/2}.
\]
Applying the bounds in \eqref{eq:inequalities} to $\bar \Delta_n(g)$, we obtain the desired result.
\end{proof}

\begin{proof}[\textnbf{Proof of Theorem \ref{thm:CLT}}]
The desired result follows from Lemma \ref{lemma:Berry-Esseen} in combination with the conditions given in the theorem. Details are omitted.
\end{proof}
\newpage
\subsection{Notation List}

{%
	\renewcommand{\arraystretch}{1.05}
	\footnotesize
	\setlength\LTleft{-6.9mm}
	\begin{longtable}{lll}
		\hline\hline
		Notation & Description & Place of Definition \\
		\hline
		& & \\
		\endhead
		& & \\
		\hline\hline
		\endfoot
		$\alpha_{n,s}$ &: strong mixing coefficients & \eqref{alpha_ns} \\

		$A_n$ &: the adjacency matrix of network $G_n$ &Above \eqref{eq:distance} \\

		$A_{n,ij}$ &: the $(i,j)$-th entry of $A_n$ &Above \eqref{eq:distance} \\

		$b_n$ &: the bandwidth in HAC estimation & Above \eqref{eq:HAC1} \\

		$c_n(s,m;k)$ &: a network statistic used in Condition ND & \eqref{eq:c_n} \\

		$\C_n$ &: the $\sigma$-field as the common shock w.r.t. which $A_n$ is measurable & Above Definition \ref{def:psi_dep} \\

		$d_n(i,i')$ &: the length of a shortest path between $i$ and $i'$ in $G_n$ & Section \ref{subsec:Network_Topology} \\

		$d_n(A,B)$ &: $\min_{i \in A} \min_{i' \in B} d_n(i,i')$ & \eqref{eq:d_n} \\

		$\mathss{d}_a$ &: distance on $\R^{v,a}$ for $v, a \in \N$ & \eqref{eq:distance} \\

		$\delta_n(s)$ &: the average of $|N_n(i;s)|$ over $i\in N_n$ & \eqref{eq:delta2_n} \\

		$\delta_n^\partial(s)$ &: $\delta_n^\partial (s;1)$ & Below \eqref{eq:delta_n} \\

		$\delta_n^\partial(s;k)$ &: the average of $|N_n^\partial(i;s)|^k$ over $i \in N_n$ & \eqref{eq:delta_n} \\

		$\Delta_n(s,m;k)$ &: a network statistic used in the definition of $c_n(s,m;k)$ & \eqref{eq:Delta_n} \\

		$H_n(s,m)$ &: a set of $(i,j,k,l) \in N_n^4$ used for CLT and HAC proofs & \eqref{eq:H} \\

		$\Lip(f)$ &: Lipschitz constant of $f$ & Footnote \ref{footnote:Lipschitz_constant} \\

		$\L_v$ &: $\{\L_{v,a}: a \in \N\}$ & \eqref{eq:Lv} \\

		$\L_{v,a}$ &: the set of real Lipschitz bounded functions on $\R^{v \times a}$ for $v,a\in \N$ & \eqref{eq:Lva} \\

		$\N$ &: the set of natural numbers & Above \eqref{eq:distance} \\

		$N_n$ &: the set of sample units as the set of nodes in network $G_n$ & Section \ref{subsec:Network_Topology} \\

		$N_n(i;s)$ &: the set of the nodes that are within distance $s$ from node $i$ & \eqref{eq:nodes_sets} \\

		$N_n^\partial(i;s)$ &: the set of the nodes that are at distance $s$ from node $i$ & \eqref{eq:nodes_sets} \\

		$\omega$ &: the kernel function used for HAC estimation & Above \eqref{eq:HAC1} \\

		$\omega_n$ &: $\omega(s/b_n)$ & Below \eqref{eq:HAC1} \\

		$\Omega_n(s)$ &: a normalized sum of $\E[Y_{n,i}Y_{n,j}^\top \mid \C_n]$ over $i \in N_n$ and $j \in N_n^\partial(i;s)$ & \eqref{eq:Omega_ns} \\

		$\hat \Omega_n(s)$ &: a normalized sum of $\left(Y_{n,i}-\bar Y_n\right)\left(Y_{n,j}-\bar Y_n\right)^{\top}$ over $i \in N_n$ and $j \in N_n^\partial(i;s)$ & \eqref{eq:Omega_hat} \\

		$\tilde \Omega_n(s)$ &: a normalized sum of $Y_{n,i}Y_{n,j}^\top$ over $i \in N_n$ and $j \in N_n^\partial(i;s)$ & \eqref{eq:Omega_tilde} \\

		$\P(a,b;s)$ &: the set of pairs $A,B \subset N_n$ such that $|A|=a,|B|=b$ and $d_n(A,B) \ge s$ & \eqref{eq:P(a,b;s)} \\

		$\Phi$ &: the distribution function of $\mathcal{N}(0,1)$ & Theorem \ref{thm:CLT} \\

		$S_n$ &: $\sum_{i \in N_n} Y_{n,i}$ & Below \eqref{eq:sigma_n2} \\

		$\sigma_n^2$ &: $\Var(S_n \mid \C_n)$ & \eqref{eq:sigma_n2} \\

		$\pi_n$ &: the maximum expected degree in network formation model \eqref{eq:graph} & \eqref{eq:pi_n} \\

		$\psi_{a,b}(f,g)$ &: $\psi$ functional in $\psi$-dependence for real functions $f,g$ & \eqref{eq:psi_dep} \\

		$\theta_{n}$ &: $(\theta_{n,s})_{s \in N_n}$ & Definition \ref{def:psi_dep} \\

		$\theta_{n,s}$ &: the dependence coefficient in $\psi$-dependence & \eqref{eq:psi_dep} \\

		$V_n$ &: $\Var(S_n/\sqrt{n} \mid \C_n)$ & \eqref{eq:trueV} \\

		$\tilde V_n$ &: $\sum_{s\ge 0}\omega_n(s)\tilde\Omega_n(s)$, an estimator of $V_n$ when $\E[Y_{n,i} \mid \C_n] = 0$ & \eqref{eq:HAC1} \\

		$\hat V_n$ &: $\sum_{s\ge 0}\omega_n(s)\hat \Omega_n(s)$, an estimator of $V_n$ when $\E[Y_{n,i} \mid \C_n]$ is unknown & \eqref{eq:HAC2} \\

		$Y_{n,A}$ &: $\{Y_{n,i}\}_{i \in A}$, for $A \subset N_n$ & \eqref{eq:Y_A} \\

		$\|\csdot\|$ &: the Euclidean norm, i.e., $\|a\| = \sqrt{a^{\top}a}$. & Below \eqref{eq:distance} \\

		$\|\csdot\|_\infty$ &: the sup norm, i.e., $\|f\|_\infty = \sup_x |f(x)|$ & Below \eqref{eq:Lva} \\

		$\|\csdot\|_{\C_n,p}$ &: $\norm{Y_{n,i}}_{\C_n,p}=(\E[ \vert Y_{n,i} \vert^p \mid \C_n])^{1/p}$ & Above Assumption \ref{assu:Uniform_L1_Integrability} \\

		$\|\csdot\|_F$ &: Frobenius norm, $\|A\|_F = \sqrt{\text{tr}(A^{\top} A)}$ & Footnote \ref{footnote:Frobenius} \\
	\end{longtable}%
}

\end{bibunit}

\clearpage
\setcounter{page}{1}

\begin{bibunit}[elsart-harv]
\vspace*{5ex minus 1ex}
\begin{center}
	\Large \textsc{Supplemental Note to ``Limit Theorems for Network Dependent Random Variables"}
	\bigskip
\end{center}
\thispagestyle{plain}
\date{\today}

\begin{center}
	\textsc{Denis Kojevnikov}\textsuperscript{\textasteriskcentered}\let\thefootnote\relax\footnotetext{\textsuperscript{\textasteriskcentered}Corresponding author. Department of Econometrics and Operations Research, Tilburg University, The Netherlands. Email: D.Kojevnikov@tilburguniversity.edu.}, \textsc{Vadim Marmer\textsuperscript{\S}, and Kyungchul Song\textsuperscript{\S}\let\thefootnote\relax\footnotetext{\textsuperscript{\S}Vancouver School of Economics, University of British Columbia, Canada.}}
\bigskip
\end{center}

This supplemental note consists of five appendices. Appendix \ref{sec:proof_3.1} provides the proof of Lemma \ref{lemma:network_form}. Appendix \ref{sec:proof_4.1} gives the proofs of Propositions \ref{prop:HAC}--\ref{cor:HAC_partially_observed}. Appendix \ref{sec:proof_aux} presents the proofs of Theorems \ref{thm:cov_ineq1} and \ref{thm:cov_ineq2} in Appendix \ref{subsec:covariance_inequalities} of the main paper. Finally, Appendix \ref{sec:more_simulations} reports additional simulation results.

\appendix
\setcounter{section}{1}

\section{Proof of Lemma \ref{lemma:network_form}}\label{sec:proof_3.1}
The following lemma is a variant of Claim 1 in the proof of \cite{Chung/Lu:01:AAM}. We assume the network formation as in \eqref{eq:graph}.

\begin{lemma}
\label{lemma:prob_bound}
For each $n \ge 2$, $s=0,\ldots,n$, $i = 1,\ldots,n$, and $\lambda>0$,
\begin{equation}
\label{eq:prob_bd}
	\PM\left\{|N_n^\partial(i;s)| \le a_s (\log n) (\pi_n \vee 1)^s \mid \varphi_n\right\} \ge 1 - s \exp \left(- \frac{1}{2}\frac{\lambda^2}{1 + \lambda/(3 \sqrt{\log n})} \right),
\end{equation}
where $a_s$ satisfies the recurrence formula as follows: for $s=1,\ldots,n$,
\begin{equation}
\label{eq:recur}
	a_s = a_{s-1} + \frac{\lambda \sqrt{a_{s-1}}}{\sqrt{(\pi_n \vee 1)^s \log n}},
\end{equation}
starting with $a_0 = 1$.
\end{lemma}

\begin{proof}[\textnbf{Proof}]
We use mathematical induction. Since $|N_n^\partial(i;0)| = 1$, the lemma holds for $s = 0$. Suppose that the statement of the lemma holds for $s \ge 0$. Let $\pi_{n,kj} = \E[A_{n,kj}|\varphi_n]$. Note that $N_n^\partial(i;s)$'s are disjoint across $s$'s. Furthermore, once sets $N_n^{\partial}(i;s)$ and $N_n(i;s)$ are determined by $A_{n,ij}, i,j \in N_n$, altering the values of $A_{n,kj}$ for any $k \in N_n^\partial(i;s)$ and $j \in N_n \setminus N_n(i;s)$ does not change the sets $N_n^{\partial}(i;s)$ and $N_n(i;s)$. Since $A_{n,ij}$'s are conditionally independent given $\varphi_n$, this means that $\{A_{n,kj}: k \in N_n^\partial(i;s), j \in N_n \setminus N_n(i;s)\}$ is a set of Bernoulli random variables that are conditionally independent given $(\varphi_n, N_n^\partial(i;s), N_n(i;s))$. Note also that for all $k \in N_n^\partial(i;s), j \in N_n \setminus N_n(i;s)$,
\[
	\E[A_{n,kj}\mid \varphi_n, N_n^\partial(i;s), N_n(i;s)] = \E[A_{n,kj}\mid\varphi_n] = \pi_{n,kj}.
\]
By Bernstein's inequality \citep[e.g., Lemma 2.2.9 of][p.~102]{vanderVaart/Wellner:96:WeakConvg}, we have
\begin{align}
\label{eq:bernstein}
	\begin{aligned}
		&\PM\left\{\sum_{j \in N_n \setminus N_n(i;s)} \sum_{k \in N_n^\partial(i;s)}\left(A_{n,kj} - \pi_{n,kj} \right) > t\mid \varphi_n, N_n^\partial(i;s), N_n(i;s)\right\} \\
		&\qquad\le\exp\left(-\frac{1}{2} \frac{t^2}{\left(\sum_{j \in N_n}\sum_{k \in N_n^\partial(i;s)} \pi_{n,kj}\right) + t/3}\right),
	\end{aligned}
\end{align}
because $A_{n,kj}$'s are Bernoulli random variables that are conditionally independent given $(\varphi_n, N_n^\partial(i;s), N_n(i;s))$, and
\begin{align*}
	\sum_{j \in N_n}\sum_{k \in N_n^\partial(i;s)} \pi_{n,kj} &\ge \sum_{j \in N_n}\sum_{k \in N_n^\partial(i;s)} \pi_{n,kj}(1- \pi_{n,kj})\\
	&\ge \sum_{j \in N_n \setminus N_n(i;s)}\sum_{k \in N_n^\partial(i;s)} \Var(A_{n,kj}\mid \varphi_n).
\end{align*}
Recall the definition of $\pi_n = \max_{1 \le i \le n} \sum_{j \in N_n} \PM \{\varphi_{n,ij} \ge \varepsilon_{ij}|\varphi_n\}$. Since
\[
	|N_n^\partial(i;s+1)| \le \sum_{j \in N_n \setminus N_n(i;s)} \sum_{k \in N_n^\partial(i;s)} A_{n,kj},
\]
and
\[
	\sum_{j \in N_n} \sum_{k \in N_n^\partial(i;s)} \pi_{n,kj} \le \pi_n |N_n^\partial(i;s)|,
\]
the inequality \eqref{eq:bernstein} implies that
\begin{align}
\label{eq:bernstein2}
	\begin{aligned}
		&\PM\left\{|N_n^\partial(i;s+1)| > t + (\pi_n \vee 1) |N_n^\partial(i;s)| \mid \varphi_n, N_n^\partial(i;s), N_n(i;s)\right\} \\
		&\qquad\le \exp\left(-\frac{1}{2} \frac{t^2}{(\pi_n \vee 1) |N_n^\partial(i;s)|+ t/3}\right).
	\end{aligned}
\end{align}
Define
\[
	\mathbb{A}_n = \{|N_n^\partial(i;s)| \le a_s (\pi_n \vee 1)^s \log n \}.
\]
We multiply the left hand side of \eqref{eq:bernstein2} by $1_{\mathbb{A}_n} + 1_{\mathbb{A}_n^c}$ and take the conditional expectation given $\varphi_n$ of both sides to obtain:
\begin{align*}
	&\PM\left\{|N_n^\partial(i;s+1)| > t + a_s (\pi_n \vee 1)^{s+1} \log n \mid \varphi_n\right\} \\
	&\qquad\le \PM\left\{|N_n^\partial(i;s)| \le a_s (\pi_n \vee 1)^s \log n \mid \varphi_n\right\} + \exp\left(-\frac{1}{2} \frac{t^2}{a_s (\pi_n \vee 1)^{s+1} \log n + t/3}\right).
\end{align*}
We fix $\lambda>0$ and take $t = \lambda \sqrt{a_s (\pi_n \vee 1)^{s+1}\log n}$. Then
\begin{align*}
	&\PM\left\{|N_n^\partial(i;s+1)| > \lambda \sqrt{a_s (\pi_n \vee 1)^{s+1}\log n} + a_s (\pi_n \vee 1)^{s+1} \log n \mid \varphi_n\right\} \\
	&\qquad\le \PM\left\{|N_n^\partial(i;s)| > a_s (\pi_n \vee 1)^s \log n \mid \varphi_n \right\} \\
	&\qquad\quad + \exp\left(-\frac{1}{2} \frac{\lambda^2 a_s (\pi_n \vee 1)^{s+1} \log n}{a_s (\pi_n \vee 1)^{s+1} \log n + (\lambda/3)\sqrt{a_s (\pi_n \vee 1)^{s+1} \log n}}\right).
\end{align*}
By the inductive hypothesis and \eqref{eq:recur},
\begin{align*}
	&\PM\left\{|N_n^\partial(i;s+1)| > a_{s+1} (\pi_n \vee 1)^{s+1} \log n \mid \varphi_n\right\} \\
	&\qquad\le s \exp \left(- \frac{1}{2}\frac{\lambda^2}{1 + \lambda/(3 \sqrt{\log n})} \right) \\
	&\qquad\quad + \exp\left(-\frac{1}{2} \frac{\lambda^2 a_s (\pi_n \vee 1)^{s+1} \log n}{a_s (\pi_n \vee 1)^{s+1} \log n + (\lambda/3)\sqrt{a_s (\pi_n \vee 1)^{s+1} \log n}}\right).
\end{align*}
Since the last term above is bounded by
\[
	\exp \left(- \frac{1}{2}\frac{\lambda^2}{1 + \lambda/(3 \sqrt{\log n})} \right),
\]
due to the fact that $a_s \ge 1$, this completes the mathematical induction. \end{proof}

\begin{lemma}
\label{lemma:aux34}
For each $k >0$,
\begin{align*}
	&\PM\left\{\delta_n^\partial(s;k) > \left(5.7 s^2 (\pi_n \vee 1)^s \log n \right)^k, \text{ for some } 1 \le s \le n\right\} \le n^{-1.3}, \text{ and } \\
	&\PM\left\{\delta_n(s;k) > \left(3 s(s+1)^2 (\pi_n \vee 1)^s \log n \right)^k, \text{ for some } 1 \le s \le n\right\} \le n^{-1.3},
\end{align*}
where
\[
	\delta_n(s;k) = \frac{1}{n}\sum_{i \in N_n} |N_n(i;s)|^k.
\]
\end{lemma}

\begin{proof}[\textnbf{Proof}]
We follow the proof of Lemma 1 of \cite{Chung/Lu:01:AAM}. If we take $\lambda = 4.7 \sqrt{\log n}$, we have
\begin{equation}
\label{eq:rate}
	n^3 \exp\left(- \frac{1}{2}\frac{\lambda^2}{1 + \lambda/(3 \sqrt{\log n})} \right) \le n^{-1.3}.
\end{equation}
Now we show that for all $s = 1,\ldots,n$,
\[
	a_s \le 5.7 s^2.
\]
Note that $a_1 = 1 + 4.7/\sqrt{\pi_n \vee 1} \le 5.7$, and hence the above inequality is satisfied when $s=1$. Now for $s \ge 1$, assume that $a_s \le 5.7 s^2$. Then note that
\begin{align*}
	a_{s+1} &= 1 + \frac{\lambda}{\sqrt{\log n}} \sum_{j=0}^s \frac{\sqrt{a_{j}}}{(\pi_n \vee 1)^{(j+1)/2}} \le 1 + 4.7\left( 1 + \sqrt{5.7} \sum_{j=1}^s j \right) \\
	&\le 1 + 4.7 \left( 1 + \frac{\sqrt{5.7}(s^2 + s)}{2} \right) \le 5.7(s+1)^2.
\end{align*}
Thus, applying this with $\lambda = 4.7 \sqrt{\log n}$ to Lemma \ref{lemma:prob_bound}, taking expectation on both sides of \eqref{eq:prob_bd}, and using \eqref{eq:rate}, we obtain that for each $s=1,\ldots,n$,
\begin{equation}
\label{eq:bd1}
	\PM\left\{|N_n^\partial(i;s)| \le 5.7 s^2 (\pi_n\vee 1)^s \log n\right\} \ge 1 - n \times n^{-1.3 - 3}.
\end{equation}
Define the event
\[
	\mathbb{A}_n(i) = \left\{ |N_n^\partial(i;s)| \le 5.7 s^2 (\pi_n \vee 1)^s \log n, \forall 1 \le s \le n \right\}.
\]
Then \eqref{eq:bd1} implies that
\begin{equation}
\label{eq:bd23}
	\PM \mathbb{A}_n^c(i) \le \sum_{s=1}^n \PM\left\{|N_n^\partial(i;s)| > 5.7 s^2(\pi_n \vee 1)^s \log n\right\} \le n^2 \times n^{-1.3-3}.
\end{equation}
Note that
\begin{align*}
	&\PM\left\{\frac{1}{n}\sum_{i \in N_n} |N_n^\partial(i;s)|^k > \left(5.7 s^2 (\pi_n \vee 1)^s \log n \right)^k, \text{ for some } 1 \le s \le n\right\} \\
	&\qquad\le \sum_{i \in N_n}\PM\left\{|N_n^\partial(i;s)| > 5.7 s^2 (\pi_n \vee 1)^s \log n, \text{ for some } 1 \le s \le n\right\} \\
	&\qquad= \sum_{i \in N_n} \PM \mathbb{A}_n^c(i) \le n^{-1.3},
\end{align*}
by \eqref{eq:bd23}. Hence, the first statement follows.

As for the second statement of the lemma, note that in the event $\mathbb{A}_n(i)$, we have for all $1 \le s \le n$,
\begin{align*}
	|N_n(i;s)| &= 1 + \sum_{t=1}^s |N_n^\partial(i;t)| \le 1 + 5.7 \log n \sum_{t=1}^s t^2 (\pi_n\vee 1)^t \\
	&\le 1 + \frac{5.7 (\log n) (\pi_n \vee 1)^s s(s+1)(2s+1)}{6} \le 3 s(s+1)^2 (\pi_n \vee 1)^s \log n.
\end{align*}
Therefore,
\begin{align*}
	&\PM\left\{|N_n(i;s)| \le 3 s(s+1)^2 (\pi_n \vee 1)^s \log n, \forall 1 \le s \le n\right\} \\
 	&\qquad\ge \PM\left\{|N_n(i;s)| \le 3 s(s+1)^2 (\pi_n \vee 1)^s \log n, \forall 1 \le s \le n\right\} \cap \mathbb{A}_{n}(i) = \PM \mathbb{A}_{n}(i).
\end{align*}
Hence,
\begin{align*}
	&\PM\left\{\frac{1}{n}\sum_{i \in N_n} |N_n(i;s)|^k > \left(3 s(s+1)^2 (\pi_n \vee 1)^s \log n \right)^k, \text{ for some } 1 \le s \le n\right\} \\
	&\qquad\le \sum_{i \in N_n} \PM\left\{|N_n(i;s)| > 3 s(s+1)^2 (\pi_n \vee 1)^s \log n, \text{ for some } 1 \le s \le n\right\} \\
	&\qquad\le \sum_{i \in N_n} \PM\mathbb{A}_n^c(i) \le n^{-1.3}. \qedhere
\end{align*}
\end{proof}

\medskip

\begin{proof}[\textnbf{Proof of Lemma \ref{lemma:network_form}}]
Without loss of generality, we will assume that $M=1$ in Condition NF. Let
\begin{align*}
	\mathbb{B}_{1n}(k) &= \left\{\delta_n^\partial(s;k) > \left(5.7 s^2 (\pi_n \vee 1)^s \log n \right)^k, \text{ for some } 1 \le s \le n \right\}, \text{ and } \\
	\mathbb{B}_{2n}(k) &= \left\{\delta_n(s;k) > \left(3 (s+1)^3 (\pi_n \vee 1)^s \log n \right)^k, \text{ for some } 1 \le s \le n \right\}.
\end{align*}
Define a sequence
\[
	q_n(k) = 5.7 \times 3^k (m_n+1)^{3k} (\log n)^{k+1} (\pi_n \vee 1)^{m_n k}.
\]
Then by the definition of $c_n(s,m_n;k)$ in \eqref{eq:c_n}, we have for $\alpha>1$,
\begin{align*}
	&\PM\left(\mathbb{B}_{1n}^c(\alpha/(\alpha - 1)) \cap \mathbb{B}_{2n}^c(k\alpha) \right) \\
	&\qquad\le \PM\left\{c_n(s,m_n;k) \le q_n(k) s^2 (\pi_n \vee 1)^s, \forall 1 \le s \le n, \delta_n(m_n;k) \le q_n(k) \right\}.
\end{align*}
(The inequality follows because $1 \le m_n \le n$.) Since
\begin{align*}
	\PM\left(\mathbb{B}_{1n}^c(\alpha/(\alpha - 1)) \cap \mathbb{B}_{2n}^c(k\alpha) \right)&\ge 1 - \PM\left(\mathbb{B}_{1n}(\alpha/(\alpha - 1)) \right)- \PM\left(\mathbb{B}_{2n}(k\alpha) \right) \\
	&\ge 1 - 2 n^{-1.3},
\end{align*}
by Lemma \ref{lemma:aux34}, we find that
\[
	\PM\left\{c_n(s,m_n;k) \le q_n(k) s^2 (\pi_n \vee 1)^s, \forall 1 \le s \le n, \delta_n(m_n;k) \le q_n(k) \right\} \ge 1 - 2 n^{-1.3}.
\]

We take $\varepsilon' < \varepsilon$ such that
\begin{equation}
\label{eq:bound3}
	(1+ \varepsilon') \log ((\pi_n \vee 1) + \varepsilon') \le \log( (\pi_n \vee 1) + \varepsilon).
\end{equation}
Let
\begin{equation}
\label{eq:m_n}
	m_n = \frac{\log n}{2( 1+ \varepsilon') \log ((\pi_n \vee 1) + \varepsilon')}.
\end{equation}
We first show that Condition ND(b) holds. Note that $c_n(0,m_n;k) = \delta_n(m_n;k)$. Hence,
\begin{equation}
\label{eq:event}
	\frac{1}{n^{k/2}}\sum_{s=0}^n c_n(s,m_n;k) \theta_{n,s}^{1 - \frac{k+2}{p}} \le \frac{q_n(k)}{n^{k/2}}\left( 1 + \sum_{s=1}^n s^2 (\pi_n \vee 1)^s \theta_{n,s}^{1 - \frac{k+2}{p}}\right),
\end{equation}
with probability at least $1 - 2 n^{-1.3}$. (Recall that we have set $\theta_{n,0} = 1$.) Let $\mathbb{B}$ be the event such that
\[
	\frac{1}{n^{k/2}}\sum_{s=0}^n c_n(s,m_n;k) \theta_{n,s}^{1 - \frac{k+2}{p}} \rightarrow 0,
\]
as $n \to \infty$. Then it suffices to show that $\PM\mathbb{B} = 1$. For this we show that
\[
	\PM\bigcap_{n=1}^\infty \bigcup_{n_1 \ge n} \left\{\frac{1}{n_1^{k/2}}\sum_{s=0}^{n_1} c_{n_1}(s,m_{n_1};k) \theta_{n_1,s}^{1 - \frac{k+2}{p}} > \eta \right\} = 0,
\]
for all $\eta>0$. Let $\mathbb{A}_{n_1}$ be the event of the inequality \eqref{eq:event} with $n=n_1$. Note that the probability above is bounded by
\begin{equation}
\label{eq:decomp234}
	\PM\bigcap_{n=1}^\infty \bigcup_{n_1 \ge n} \left\{\frac{1}{n_1^{k/2}}\sum_{s=0}^{n_1} c_{n_1}(s,m_{n_1};k) \theta_{n_1,s}^{1 - \frac{k+2}{p}} > \eta \right\} \cap \mathbb{A}_{n_1} + \PM\bigcap_{n=1}^\infty \bigcup_{n_1 \ge n} \mathbb{A}_{n_1}^c.
\end{equation}
Observe that
\[
	\sum_{n=1}^\infty \PM \mathbb{A}_n^c \le 2 \sum_{n=1}^\infty n^{-1.3} < \infty.
\]
Hence, by Borel-Cantelli Lemma, the last probability in \eqref{eq:decomp234} is zero. As for the leading term in \eqref{eq:decomp234}, observe that
\begin{align}
\label{eq:ineq}
	\begin{aligned}
	& \PM\bigcap_{n=1}^\infty \bigcup_{n_1 \ge n} \left\{\frac{1}{n_1^{k/2}}\sum_{s=0}^{n_1} c_{n_1}(s,m_{n_1};k) \theta_{n_1,s}^{1 - \frac{k+2}{p}} > \eta \right\} \cap \mathbb{A}_{n_1} \\
	&\qquad\le \PM\bigcap_{n=1}^\infty \bigcup_{n_1 \ge n} \left\{\frac{q_n(k)}{n^{k/2}}\left( 1 + \sum_{s=1}^n s^2 (\pi_n \vee 1)^s \theta_{n,s}^{1 - \frac{k+2}{p}}\right) > \eta \right\} \cap \mathbb{A}_{n_1}.
	\end{aligned}
\end{align}

Since
\[
	\frac{m_n k \log (\pi_n \vee 1)}{\log n} \le \frac{k \log (\pi_n \vee 1)}{2(1 + \varepsilon')\log ((\pi_n \vee 1) + \varepsilon')} < \frac{k}{2},
\]
we have
\[
	\frac{q_n(k)}{n^{k/2}} \rightarrow 0,
\]
as $n \to \infty$. On the other hand, observe that by Condition NF,
\[
	(\pi_n \vee 1)^s \theta_{n,s}^{ 1 - 4/p} \le (\pi_n \vee 1)^s ((\pi_n \vee 1) + \varepsilon)^{-\frac{q(p-4)s}{p}},
\]
eventually with probability one, and hence
\begin{align*}
	\sup_{n \ge 1} \sum_{s=1}^n s^2 (\pi_n \vee 1)^s \theta_{n,s}^{1 - \frac{k+2}{p}} &\le \sup_{n \ge 1} \sum_{s=1}^n s^2 (\pi_n \vee 1)^s \theta_{n,s}^{1 - \frac{4}{p}} \\
	&\le \sup_{n \ge 1} \sum_{s=1}^n s^2 (\pi_n \vee 1)^s ((\pi_n \vee 1) + \varepsilon)^{-\frac{q(p-4)s}{p}} < \infty,
\end{align*}
because $q > p/(p-4)$. The first inequality follows because $k \in \{1,2\}$. We have the probability in \eqref{eq:ineq} as zero. We find that $\PM\mathbb{B} = 1$. Thus, Condition ND(b) is satisfied.

Now it is not hard to see that Condition ND(a) is satisfied by the choice of $m_n$ as in \eqref{eq:m_n}. Indeed, note that from Condition NF,
\[
	\theta_{n,m_n}^{1 - 1/p} \le ((\pi_n \vee 1) + \varepsilon)^{-m_n(1 - 1/p)q},
\]
eventually with probability one. Since $q > 3p/(p-1)$, we have
\begin{align*}
	& -m_n \left(1 - \frac{1}{p} \right) q\log ((\pi_n \vee 1) + \varepsilon) + \frac{3}{2} \log n\\
	 &\qquad= (\log n) \left(-\left(1 - \frac{1}{p} \right) \frac{q \log ((\pi_n \vee 1) + \varepsilon)}{2(1 + \varepsilon') \log((\pi_n \vee 1) + \varepsilon')} + \frac{3}{2} \right)\\
	 &\qquad\le (\log n) \left(-\left(1 - \frac{1}{p} \right) \frac{q}{2} + \frac{3}{2} \right) \rightarrow - \infty,
\end{align*}
where the last inequality comes from \eqref{eq:bound3}. We find that
\[
	\theta_{n,m_n}^{1 - 1/p} = o_{a.s.}(n^{-3/2}),
\]
showing Condition ND(a).
\end{proof}

\section{Proofs of Propositions \ref{prop:HAC}--\ref{cor:HAC_partially_observed}}\label{sec:proof_4.1}

\begin{proof}[\textnbf{Proof of Proposition \ref{prop:HAC}}]
For the first implication it suffices to show that for any vector $c\in \R^v$ with $\norm{c}=1$, $\E[\abs{A_n(c)}\mid \C_n]\to 0$ a.s., where $A_n(c)\eqdef c^{\top}(\tilde{V}_n-V_n)c$. Let $\mu\eqdef \sup_n\max_{i\in N_n} \norm{Y_{n,i}}_{\C_n,p}$ and let $y_{n,i}\eqdef c^{\top}Y_{n,i}$. Notice that $\{y_{n,i}\}$ is $(\L_1,\psi)$-dependent with the dependence coefficients $\{\theta_{n,s}\}$. In addition, $\E[y_{n,i}\mid \C_n]=0$ a.s.\ and by Assumption \ref{assu:HAC1}(i),
\[
\sup_{n\geq 1}\max_{i\in N_n}\normin{y_{n,i}}_{\C_n,p}\le\mu<\infty \qtext{a.s.}
\]
Then
\begin{align}
	\begin{aligned}
		\label{eq:var_diff}
		A_n(c) &=\frac{1}{n}\sum_{i\in N_n}\left(y_{n,i}^2-\E[y_{n,i}^2\mid \C_n]\right) \\
		&\quad +\sum_{s\ge 1}\omega_n(s)\times \frac{1}{n}\sum_{i\in N_n}\sum_{j\in N_n^{\partial}(i;s)}\left(y_{n,i}y_{n,j}-\E[y_{n,i}y_{n,j}\mid \C_n]\right) \\
		&\quad +\sum_{s\ge 1}[\omega_n(s)-1]\times \frac{1}{n}\sum_{i\in N_n}\sum_{j\in N_n^{\partial}(i;s)}\E [y_{n,i}y_{n,j}\mid \C_n]\\
		&\equiv R_{n,0}+R_{n,1}+R_{n,2}.
	\end{aligned}
\end{align}

Consider each term in the last line of \eqref{eq:var_diff} separately. Using Theorem \ref{thm:cov_ineq1} for $y_{n,i}$ and $y_{n,j}$ with $d_n(i,j)=s\ge 1$ and Assumption \ref{assu:psi_dep}(b), we obtain
\[
\abs{\E[y_{n,i}y_{n,j}\mid \C_n)]}\le \vartheta_2\theta_{n,s}^{1-\frac{2}{p}} \qtext{a.s.},
\]
where $\vartheta_2=C(\mu\vee 1)^2 \bar \theta $ for some constant $C>0$, and
\begin{equation}\label{eq:bar_theta}
\bar \theta =\sup_{n\ge 1}\max_{s \ge 1} \overline \theta_{n,s}.
\end{equation}
Therefore, 
\begin{align*}
	\abs{R_{n,2}}&\le \sum_{s\ge 1}\abs{\omega_n(s)-1} \times\frac{1}{n}\sum_{i\in N_n}\sum_{j\in N_n^{\partial}(i;s)}\abs{\E[y_{n,i}y_{n,j}\mid \C_n]} \\
	&\le \vartheta_2\sum_{s\ge 1}\abs{\omega_n(s)-1}\theta_{n,s}^{1-\frac{2}{p}}\times \frac{1}{n}\sum_{i\in N_n}\abs{N_n^{\partial}(i;s)} \\
	&= \vartheta_2\sum_{s\ge 1}\abs{\omega_n(s)-1}\delta_n^{\partial}(s)\theta_{n,s}^{1-\frac{2}{p}} \qtext{a.s.},
\end{align*}
and it follows from Assumption \ref{assu:HAC1}(ii) that $\abs{R_{n,2}}=o_{a.s.}(1)$.

Let $z_{n,i,j}\eqdef y_{n,i} y_{n,j}-\E[y_{n,i} y_{n,j}\mid \C_n]$ so that $\E[z_{n,i,j}\mid \C_n]=0$ a.s. Then, using Corollary \ref{corr:cov_ineq_prod} for $z_{n,i,j}$ and $z_{n,k,l}$ with $d_n(\{i,j\},\{k,l\})=s\ge 1$,
\begin{align*}
	&\abs{\E[z_{n,i,j}z_{n,k,l}\mid \C_n]} \le \vartheta_1\theta_{n,s}^{1-\frac{4}{p}} \qtext{a.s.},
\end{align*}
where $\vartheta_1=C (\mu\vee 1)^4 \bar \theta $ for some constant $C>0$. To deal with the case in which $d_n(\{i,j\},\{k,l\})=0$, note that $p>4$ so that
\[
	\abs{\E[z_{n,i,j}z_{n,k,l}\mid \C_n]}\le \left[\Var(y_{n,i}y_{n,j}\mid \C_n)\Var(y_{n,k}y_{n,l}\mid \C_n)\right]^{1/2} \le \mu^4 \qtext{a.s.}
\]
Noticing that $|\omega(\csdot)|\le 1$, we find that
\begin{align*}
	\E[R_{n,1}^2\mid \C_n] &\le \frac{1}{n^2}\sum_{\substack{i,j\in N_n: \\ 1\le d_n(i,j)\le b_n}}\sum_{\substack{k,l\in N_n: \\ 1\le d_n(k,l)\le b_n}}\abs{\E[z_{n,i,j}z_{n,k,l}\mid \C_n]} \\
	&\le \frac{1}{n^2} \sum_{s\ge 0} \sum_{(i,j,k,l)\in H_n(s,b_n)}\abs{\E[z_{n,i,j}z_{n,k,l}\mid \C_n]} \\
	&\le \frac{\vartheta_1}{n}\sum_{s\ge 0}c_n(s,b_n;2)\theta_{n,s}^{1-\frac{4}{p}} \qtext{a.s.},
\end{align*}
where the last inequality is due to \eqref{eq:H_bound}. Hence, it follows from Assumption \ref{assu:HAC1}(iii) that $\E[R_{n,1}^2\mid \C_n]\to 0$ a.s.

Finally, since
\begin{equation}\label{eq:delta_H_bound}
\delta_n^{\partial}(s)\le n^{-1}\abs{H_n(s,b_n)},
\end{equation}
	 it is not hard to show that
\begin{align}\label{eq:like_LLN}
	\begin{aligned}
		\E[R_{n,0}^2\mid \C_n] &\le \frac{1}{n^2}\sum_{s\ge 0}\sum_{i\in N_n}\sum_{j\in N_n^{\partial}(i;s)}\abs{\Cov(y_{n,i}^2,y_{n,j}^2\mid \C_n)} \\
		&\le \frac{\vartheta_0}{n}\sum_{s\ge 0}\delta_n^{\partial}(s)\theta_{n,s}^{1-\frac{4}{p}}\to 0 \qtext{a.s.},
	\end{aligned}
\end{align}
where $\vartheta_0=C(\mu\vee 1)^4 \bar \theta $ for some constant $C>0$.

As for the second implication define $\bar y_n\eqdef c^{\top}\bar Y_n$, $\lambda_n\eqdef c^{\top}\Lambda_n$ and consider the difference between two estimators, $A_n'(c)\eqdef c^{\top}(\hat{V}_n-\tilde{V}_n)c$, which can be written as follows:
\begin{align}
	A_n'(c)&=\sum_{s\ge 0}\omega_n(s) c^{\top}\left(\hat\Omega_n(s)-\tilde\Omega_n(s)\right)c \notag\\
	&=(\bar y_n-\lambda_n)^2\sum_{s\ge 0}\omega_n(s)\times \frac{1}{n}\sum_{i\in N_n}\absin{N_n^{\partial}(i;s)} \label{eq:diff_lambda_1}\\
	&\quad -(\bar y_n-\lambda_n)\sum_{s\ge 0}\omega_n(s)\times\frac{2}{n}\sum_{i\in N_n}\absin{N_n^{\partial}(i;s)}(y_{n,i}-\lambda_n). \label{eq:diff_lambda_2}
\end{align}
First, consider the expression in \eqref{eq:diff_lambda_1}:
\begin{align}
\begin{aligned}\label{eq:feas_infeas_1}
	\norm{(\bar y_n-\lambda_n)^2\sum_{s\ge 0}\omega_n(s)\times \frac{1}{n}\sum_{i\in N_n}\absin{N_n^{\partial}(i;s)} }_{C_n,1} \
&\leq \norm{\bar y_n-\lambda_n}_{\C_n,2}^2 \frac{1}{n} \sum_{i\in N_n}\sum_{s\leq b_n}\absin{N_n^{\partial}(i;s)} \\
	&=O_{a.s.}(n^{-1})\delta_n(b_n)\\
	&=o_{a.s.}(1),
\end{aligned}
\end{align}
where the result in the first line holds because $\vert \omega_n(s)\vert \leq 1$, the $O_{a.s.}(1/n)$ term in the second line is by the same argument as in \eqref{eq:like_LLN}, the $\delta_n(b_n)$ term appears in the second line because $\sum_{s\leq b_n}\absin{N_n^{\partial}(i;s)}=\vert N_n(i;b_n) \vert$, and the result in the last line holds by the assumption $\delta_n(b_n)=o_{a.s.}(n)$.

For the expression in \eqref{eq:diff_lambda_2}, we have:
\begin{align}
	&\norm{(\bar y_n-\lambda_n)\sum_{s\ge 0}\omega_n(s)\times\frac{1}{n}\sum_{i\in N_n}\absin{N_n^{\partial}(i;s)}(y_{n,i}-\lambda_n)}_{\C_n,2}^2 \label{eq:second_bound}\\
	&\qquad\leq \norm{\bar y_n-\lambda_n}_{\C_n,2}^2 \norm{\sum_{s\ge 0}\omega_n(s)\times\frac{1}{n}\sum_{i\in N_n}\absin{N_n^{\partial}(i;s)}(y_{n,i}-\lambda_n)}_{\C_n,2}^2 \notag \\
	&\qquad= O_{a.s.}\left(\frac{1}{n}\right) \times \norm{\frac{1}{n}\sum_{i\in N_n}\left(\sum_{s\ge 0}\omega_n(s)\absin{N_n^{\partial}(i;s)}\right)(y_{n,i}-\lambda_n)}_{\C_n,2}^2
\label{eq:O}\\
	&\qquad\leq O_{a.s.}\left(\frac{1}{n^3}\right) \sum_{s\geq 0} \theta_{n,s}^{1-\frac{2}{p}} \sum_{i \in N_n} \vert N_n(i;b_n)\vert \sum_{j\in N_n^\partial(i;s)} \vert N_n(j;b_n) \vert \label{eq:remove_omegas} \\
	&\qquad=O_{a.s.}\left(\frac{1}{n^3}\right) \sum_{s\geq 0} \theta_{n,s}^{1-\frac{2}{p}} \vert J_n(s;b_n) \vert,\label{eq:with_J}
\end{align}
where
\[
	J_n(s;b_n)=\left\{(i,j,k,l)\in N_n^4: j\in N_n(i;b_n), l\in N_n(k;b_n), d_n(i,k)=s\right\}.
\]
The $O_{a.s.}(1/n)$ term in \eqref{eq:O} is by the same argument as in \eqref{eq:like_LLN}. The result in \eqref{eq:remove_omegas} holds by $\vert \omega_n(s)\vert \leq 1$, $\sum_{s\ge 0}\omega_n(s)\absin{N_n^{\partial}(i;s)}\leq \vert N_n(i;b_n)\vert $, and the same argument as in \eqref{eq:like_LLN}.

Note that $d_n(\{i,j\},\{k,l\})\leq d_n(i,j)$. Moreover, since the tuples $(i,j,k,l)$ in the definitions of $H_n(s;b_n)$ and $J_n(s;b_n)$ are ordered, $(i,j,k,l)\in J_n(s;b_n)$ implies that $(i,j,k,l)\in H_n(t;b_n)$ for a unique $t=d_n(\{i,j\},\{k,l\})\leq s$. Similarly, $(i,j,k,l)\in H_n(t;b_n)$ implies that $(i,j,k,l)\in J_n(s;b_n)$ for a unique $s=d(i,k)\geq t$. We now have:
\begin{align}
	\sum_{s\geq 0} \theta_{n,s}^{1-\frac{2}{p}} \vert J_n(s;b_n) \vert
	&= \sum_{s\geq 0} \theta_{n,s}^{1-\frac{2}{p}} \sum_{t=0}^s\vert J_n(s;b_n) \cap H_n(t;b_n)\vert \label{eq:switch1} \\
	&\leq \bar\theta^q n\sum_{s\geq 0} \frac{\theta_{n,s}^{1-(4/p)}}{s\vee 1} \sum_{t=0}^s\vert J_n(s;b_n) \cap H_n(t;b_n)\vert \notag\\
	&\leq \bar\theta^q n\sum_{s\geq 0} \sum_{t=0}^s \frac{\theta_{n,t}^{1-(4/p)}}{t\vee 1} \vert J_n(s;b_n) \cap H_n(t;b_n)\vert \label{eq:switch2} \\
	&\leq \bar\theta^q n\sum_{t\geq 0} \theta_{n,t}^{1-\frac{4}{p}} \sum_{s\geq 0} \vert J_n(s;b_n) \cap H_n(t;b_n)\vert \notag \\
	&= \bar\theta^q n\sum_{t\geq 0} \theta_{n,t}^{1-\frac{4}{p}} \vert H_n(t;b_n)\vert, \label{eq:J}
\end{align}
where $q=(p-2)/(p-4)$. The inequality in \eqref{eq:switch1} holds by the definition of $\bar \theta$ in \eqref{eq:bar_theta} and because $J(s;b_n)=\varnothing$ for any $s\ge n$, and the inequality in \eqref{eq:switch2} holds by the assumption that $\{\theta_{n,s}/s^{p/(p-4)}\}$ are non-increasing in $s$. By Assumption \ref{assu:psi_dep}(b), \eqref{eq:with_J}, and \eqref{eq:J}, the expression in \eqref{eq:second_bound} is now bounded by
\begin{equation}\label{eq:end_HAC}
	O_{a.s.}\left(\frac{1}{n}\right) \sum_{t\geq 0} \theta_{n,t}^{1-\frac{4}{p}} c_n(t,b_n;2)\to_{a.s.} 0,
\end{equation}
where the convergence holds by Assumption \ref{assu:HAC1}(iii).
\end{proof}

\begin{proof}[\textnbf{Proof of Proposition \ref{prop:kernel}}]
By \eqref{eq:kernel},
\begin{align}
	\sum_{s\ge 1}\abs{\omega_n(s)-1}\delta_n^{\partial}(s) \theta_{n,s}^{1-\frac{2}{p}} & \leq \frac{C}{b_n^{1+\eta}}\sum_{s\geq 1}s^{1+\eta}\delta_n^{\partial}(s) \theta_{n,s}^{1-\frac{2}{p}} \notag \\
	& \leq \frac{5.7\times C\log{n}}{b_n^{1+\eta}}\sum_{s\geq 1}s^{3+\eta} (\pi_n \vee 1)^s \theta_{n,s}^{1-\frac{2}{p}} \label{eq:partial_bound} \\
	& \leq O_{a.s.}\left(b_n^{-\eta}\right)\sum_{s\geq 1}s^{3+\eta} (\pi_n \vee 1)^{s+1} \theta_{n,s}^{1-\frac{2}{p}} \label{eq:bndwdth} \\
	& \leq O_{a.s.}\left(b_n^{-\eta}\right)\sum_{s\geq 1}s^{3+\eta} ((\pi_n \vee 1)+\varepsilon)^{s+1-\frac{qs(p-2)}{p}} \label{eq:byNF} \\
	& \leq O_{a.s.}\left(b_n^{-\eta}\right)\sum_{s\geq 1}s^{3+\eta} ((\pi_n \vee 1)+\varepsilon)^{1-s} , \label{eq:byNF2} 
\end{align}
where \eqref{eq:partial_bound} holds eventually with probability one by Lemma \ref{lemma:aux34} and the Borel-Cantelli Lemma, \eqref{eq:bndwdth} holds since $(\log{n})/b_n=O_{a.s.}(\pi_n\vee 1)$, and \eqref{eq:byNF}--\eqref{eq:byNF2} hold by Condition NF. The result follows because $\sup_{n\geq 1}\sum_{s\geq 1}s^{3+\eta} ((\pi_n \vee 1)+\varepsilon)^{1-s}<\infty$, $\eta>0$, and $b_n\to\infty$ by the assumptions of the proposition.
\end{proof}

\begin{proof}[\textnbf{Proof of Proposition \ref{cor:HAC_partially_observed}}]
Similarly to the proof of Proposition \ref{prop:HAC}, write $c^\top (\tilde V^*_n -V_n) c= R_{n,0}+R_{n,1}+R_{n,2}+R_{n,3}$, where $R_{n,0}$ is the same as in \eqref{eq:var_diff},
\begin{align*}
	R_{n,1}& = \sum_{s\ge 1} \omega_n(s) \times \frac{1}{n} \sum_{i\in N_n} \sum_{j\in N^{*\partial}_n(i;s)}(y_{n,i}y_{n,j}-\E[y_{n,i}y_{n,j} \mid \C_n]), \\
	R_{n,2} &=\sum_{s\ge 1}[\omega_n(s)-1]\times \frac{1}{n}\sum_{i\in N_n}\sum_{j\in N_n^{*\partial}(i;s)}\E [y_{n,i}y_{n,j}\mid \C_n],\\
	R_{n,3} &=\sum_{s\ge 1}\frac{1}{n}\sum_{i\in N_n}\sum_{\substack{j\in N_n^{\partial}(i;s): d^*_n(i,j)=\infty}}\E [y_{n,i}y_{n,j}\mid \C_n].
\end{align*}
By Theorem \ref{thm:cov_ineq1} and Assumption \ref{assu:partial}, 
\[
\vert R_{n,3}\vert\leq \vartheta_2 \sum_{s\ge 1} \theta_{n,s}^{1-(2/p)} \delta^\partial_n(s\mid d^*_n=\infty)=o_{a.s.}(1),
\]
where the random variable $\vartheta_2$ is defined in the proof of Proposition \ref{prop:HAC}. For $R_{n,2}$, write
\begin{align*}
		\vert R_{n,2}\vert & \leq \sum_{s\ge 1}\vert\omega_n(s)-1\vert\times \frac{1}{n}\sum_{i\in N_n}\sum_{s'=1}^s\sum_{j\in N_n^{*\partial}(i;s) \cap N_n^{\partial}(i;s')}\vert\E [y_{n,i}y_{n,j}\mid \C_n] \vert \\
	&\leq \vartheta_2 \sum_{s\ge 1} \vert\omega_n(s)-1\vert \sum_{s'=1}^{s} \theta_{n,s'}^{1-\frac{2}{p}} \times \frac{1}{n} \sum_{i\in N_n} \vert N_n^{*\partial}(i;s) \cap N_n^{\partial}(i;s')\vert\\
	& \leq \vartheta_2 \sum_{s\ge 1} \sum_{s'=1}^{s} \vert\omega_n(s')-1\vert \theta_{n,s'}^{1-\frac{2}{p}} \times \frac{1}{n} \sum_{i\in N_n} \vert N_n^{*\partial}(i;s) \cap N_n^{\partial}(i;s')\vert\\
	& \leq \vartheta_2 \sum_{s'\ge 1} \vert\omega_n(s')-1\vert \theta_{n,s'}^{1-\frac{2}{p}} \times \frac{1}{n} \sum_{i\in N_n} \sum_{s\geq s'} \vert N_n^{*\partial}(i;s) \cap N_n^{\partial}(i;s')\vert\\
	&\leq \vartheta_2 \sum_{s'\ge 1} \vert\omega_n(s')-1\vert \theta_{n,s'}^{1-\frac{2}{p}} \delta^\partial_n(s')\\
	&=o_{a.s.}(1),
	\end{align*}
	where the inequality in the third line holds under $\vert \omega(s) -1\vert \leq \vert \omega(s') -1\vert$ for $s\geq s'$, and the equality in the last line holds by Assumption \ref{assu:HAC1}(ii).
Lastly, as in the proof of Proposition \ref{prop:HAC},
\begin{align*}
	\E[R_{n,1}^2 \mid \C_n] &\leq \frac{1}{n^2} \sum_{s\ge 0}\sum_{\substack{i,j\in N_n:\\ 1\leq d^*_n(i,j)\leq b_n}} \sum_{\substack{k,l\in N_n:1\leq d^*_n(k,l)\leq b_n, \\ d_n(\{i,j\},\{k,l\})=s}} |\E_n [z_{n,i,j} z_{n,k,l} \mid \C_n]| \\
	&\leq \frac{\vartheta_1}{n} \sum_{s\ge 0} c_n(s,b_n;2) \theta^{1-\frac{4}{p}}_{n,s},
\end{align*}
where the second inequality holds because $d_n(i,j)\leq d^*_n(i,j)$, and the random variable $\vartheta_1$ is defined in the proof of Proposition \ref{prop:HAC}.

For the second part of the proposition, as in the proof of the second part of Proposition \ref{prop:HAC}, write
\begin{align}
	&\sum_{s\ge 0}\omega_n(s) c^{\top}\left(\hat\Omega^*_n(s)-\tilde\Omega^*_n(s)\right)c \notag\\
	&\qquad=(\bar y_n-\lambda_n)^2\sum_{s\ge 0}\omega_n(s)\times \frac{1}{n}\sum_{i\in N_n}\absin{N_n^{*\partial}(i;s)} \label{eq:diff_lambda_1*}\\
	&\qquad\quad -(\bar y_n-\lambda_n)\sum_{s\ge 0}\omega_n(s)\times\frac{2}{n}\sum_{i\in N_n}\absin{N_n^{*\partial}(i;s)}(y_{n,i}-\lambda_n). \label{eq:diff_lambda_2*}
	\end{align}
By the same arguments as in \eqref{eq:feas_infeas_1} and since $\delta_n^*(b_n)=o_{a.s.}(n)$ by the assumption in the second part of the proposition, for the expression in \eqref{eq:diff_lambda_1*} we have:
\[
	\norm{(\bar y_n-\lambda_n)^2\sum_{s\ge 0}\omega_n(s)\times \frac{1}{n}\sum_{i\in N_n}\absin{N_n^{*\partial}(i;s)} }_{\C_n,1}=o_{a.s.}(1).
\]
The expression in \eqref{eq:diff_lambda_2*} can be treated similarly to \eqref{eq:second_bound}:
\begin{align}
	& \norm{(\bar y_n-\lambda_n)\sum_{s\ge 0}\omega_n(s)\times\frac{1}{n}\sum_{i\in N_n}\absin{N_n^{*\partial}(i;s)}(y_{n,i}-\lambda_n)}_{\C_n,2}^2 \notag \\
	&\qquad\leq O_{a.s.}\left(\frac{1}{n}\right) \times \norm{\frac{1}{n}\sum_{i\in N_n}\left(\sum_{s\ge 0}\omega_n(s)\absin{N_n^{*\partial}(i;s)}\right)(y_{n,i}-\lambda_n)}_{\C_n,2}^2 \notag\\
	&\qquad\leq O_{a.s.}\left(\frac{1}{n^3}\right) \sum_{s\geq 0} \theta_{n,s}^{1-2/p} \sum_{i \in N_n} \vert N_n^*(i;b_n)\vert \sum_{j\in N_n^\partial(i;s)} \vert N_n^*(j;b_n) \vert \notag\\
	&\qquad\leq O_{a.s.}\left(\frac{1}{n^3}\right) \sum_{s\geq 0} \theta_{n,s}^{1-2/p} \sum_{i \in N_n} \vert N_n(i;b_n)\vert \sum_{j\in N_n^\partial(i;s)} \vert N_n(j;b_n) \vert \label{eq:N*_N}\\
	&\qquad= O_{a.s.}\left(\frac{1}{n^3}\right) \sum_{s\geq 0} \theta_{n,s}^{1-2/p} \vert J_n(s;b_n) \vert,\notag
\end{align}
where the result in \eqref{eq:N*_N} holds because $d^*_n(i,j)\geq d_n(i,j)$ and, therefore, $N^*_n(i;s)\subset N_n(i;s)$. The rest of the proof is similar to that of the second part of Proposition \ref{prop:HAC} (see equations \eqref{eq:with_J}--\eqref{eq:end_HAC}).
\end{proof}

\section{Proofs of Auxiliary Results}\label{sec:proof_aux}
\begin{proof}[\textnbf{Proof of Lemma A.1}] Suppose w.l.o.g.\ that $\E[f(Y_{n,A},Z_1)\mid \C_n]=0$ and $\E[g(Y_{n,B},Z_2)\mid \C_n]=0$ a.s. By Lemma 1.3 in \cite{DaPrato:14:StochasticEq} we can approximate $Z_j$ by a sequence of simple functions $\{Z_{j,m}\}$ s.t.\ $\rho_j(Z_{j,m},Z_j)\searrow 0$ pointwise, and for each $m\ge 1$, $Z_{j,m}=\sum_{k=1}^m z_{j,k}\ind_{A_{j,k}}$, where $z_{j,k} \in \Z_j$, $A_{j,k}\in \C_n$ and $A_{j,k}\cap A_{j,l}=\varnothing$ for $k\ne l$. Then, letting $B_{k,l}\eqdef A_{1,k}\cap A_{2,l}$,
\begin{align*}
	\abs{\E[f(Y_{n,A},Z_{1,m})g(Y_{n,B},Z_{2,m})\mid \C_n]}&\le \sum_{k,l=1}^m \abs{\E[f(Y_{n,A},z_{1,k})g(Y_{n,B},z_{2,l})\mid \C_n]}\ind_{B_{k,l}} \\
	&\le \sum_{k,l=1}^m \psi_{a,b}(f^{z_{1,k}},g^{z_{2,l}})\ind_{B_{k,l}}\theta_{n,s} \\
	&=F(Z_{1,m},Z_{2,m})\theta_{n,s} \qtext{a.s.}
\end{align*}
The second inequality above is due to \eqref{eq:psi_dep_gen}. Consequently, the result follows by the conditional dominated convergence theorem.
\end{proof}

\begin{proof}[\textnbf{Proof of Theorem \ref{thm:cov_ineq1}}]
Fix $\kappa,\lambda\ge 1$ and let $\Xi\eqdef \{(\mu_{\xi,p},\mu_{\zeta,q})\in (0,\infty)^2\}$. Next we define $\xi'\eqdef \mu_{\xi,p}^{-1}\xi\ind_{\Xi}$,
\begin{alignat*}{2}
	&\xi_{\kappa}\eqdef (\varphi_{\kappa}\circ \mu_{\xi,p}^{-1}f)(Y_{n,A})\ind_{\Xi}, \quad &\xi_{\kappa}^*\eqdef \xi_{\kappa}-\E[\xi_{\kappa}\mid \C_n], \\
	&\hat{\xi}_{\kappa}\eqdef \xi'-\xi_{\kappa}, &\hat{\xi}_{\kappa}^*\eqdef \hat{\xi}_{\kappa}-\E[\hat{\xi}_{\kappa}\mid \C_n],
\end{alignat*}
and, similarly, $\zeta',\zeta_{\lambda},\zeta_{\lambda}^*,\hat{\zeta}_{\lambda}$, and $\hat{\zeta}^*_{\lambda}$, where we use $g$, $\mu_{\zeta,q}$, and $\lambda$ instead of $f$, $\mu_{\xi,p}$, and $\kappa$. First,
\begin{align*}
	\absin{\Cov(\xi',\zeta'\mid \C_n)}
	&= \absin{\E[(\xi_{\kappa}^*+\hat{\xi}_{\kappa}^*)(\zeta_{\lambda}^*+\hat{\zeta}_{\lambda}^*)\mid \C_n]} \\
	&\le \absin{\E[\xi_{\kappa}^*\zeta_{\lambda}^*\mid \C_n]}+\absin{\E([\xi_{\kappa}^*\hat{\zeta}_{\lambda}^*\mid \C_n]} \\
	&\quad+ \absin{\E[\hat{\xi}_{\kappa}^*\zeta_{\lambda}^*\mid \C_n]}+\absin{\E[\hat{\xi}_{\kappa}^*\hat{\zeta}_{\lambda}^*\mid \C_n]} \qtext{a.s.}
\end{align*}
Consider each term in the last inequality separately. By Lemma \ref{lemma:psi_dep} and Assumption \ref{assu:psi_dep_gen} we find that\footnote{
	Note that for $x\ge 0$ and $z\ne 0$, $\varphi_{x}\circ z^{-1}f=z^{-1}(\varphi_{xz}\circ f)$.
}
\begin{align*}
	\absin{\E[\xi_{\kappa}^*\zeta_{\lambda}^*\mid \C_n]}&\le \psi_{a,b}(\varphi_{\kappa}\circ \mu_{\xi,p}^{-1}f,\varphi_{\lambda}\circ \mu_{\zeta,q}^{-1}g)\theta_{n,s} \\
	&\le \frac{\kappa\lambda}{\mu_{\xi,p}\mu_{\zeta,q}}\overline{\psi}_{a,b}(\mu_{\xi,p},\mu_{\zeta,q})\theta_{n,s} \qtext{a.s.\ on \ $\Xi$}.
\intertext{As for the other terms, noticing that $\absin{\xi_{\kappa}^*}\le 2\kappa$ a.s., we have}
	\absin{\E[\xi_{\kappa}^*\hat{\zeta}_{\lambda}^*\mid \C_n]}&=\absin{\Cov(\xi_{\kappa}^*,\hat{\zeta}_{\lambda}^*\mid \C_n)}=\absin{\Cov(\xi_{\kappa}^*,\hat{\zeta}_{\lambda}\mid \C_n)} \\
	&\le \E[\absin{\xi_{\kappa}^*}\absin{\hat{\zeta}_{\lambda}}\mid \C_n]\le 2\kappa\E[\absin{\hat{\zeta}_{\lambda}}\mid \C_n] \\
	&\le 4\kappa\lambda^{1-q} \qtext{a.s.\ on \ $\Xi$}
	\intertext{because $\norm{\zeta'}_{C_n,q}=\ind_{\Xi}$ a.s.\ and}
	\E[\absin{\hat{\zeta}_{\lambda}}\mid \C_n]&= \E[\absin{\zeta'-\zeta_{\lambda}}\ind\{\zeta'>\lambda\}\mid \C_n] \\
	&\le \left(\E[\absin{\zeta'-\zeta_{\lambda}}^q\mid \C_n]\right)^{1/q}\left(\PM(\zeta'>\lambda\mid \C_n)\right)^{1-1/q} \\
	&\le 2\norm{\zeta'}_{\C_n,q}(\lambda^{-q}\E[\absin{\zeta'}^q\mid \C_n])^{1-1/q} \\
	&=2\lambda^{1-q} \qtext{a.s.\ on \ $\Xi$}.
\intertext{Similarly,}
	\absin{\E[\hat{\xi}_{\kappa}^* \zeta_{\lambda}^*\mid \C_n]}&\le 4\kappa^{1-p}\lambda \qtext{a.s.\ on \ $\Xi$}.
\intertext{Finally,}
	\absin{\E[\hat{\xi}_{\kappa}^* \hat{\zeta}_{\lambda}^*\mid \C_n]}&= \absin{\Cov(\hat{\xi}_{\kappa}^*, \hat{\zeta}_{\lambda}^*\mid \C_n)}=\absin{\Cov(\hat{\xi}_{\kappa},\hat{\zeta}_{\lambda}\mid \C_n)} \\
	&\le \absin{\E[\hat{\xi}_{\kappa} \hat{\zeta}_{\lambda}\mid \C_n]}+\E[\absin{\hat{\xi}_{\kappa}}\mid \C_n]\E[\absin{\hat{\zeta}_{\lambda}}\mid \C_n] \\
	&\le \absin{\E[\hat{\xi}_{\kappa} \hat{\zeta}_{\lambda}\mid \C_n]}+4\kappa^{1-p}\lambda^{1-q} \qtext{a.s.\ on \ $\Xi$},
\intertext{and for $p',q'$ s.t.\ $1/p'+1/q'=1-1/p-1/q$ we find that}
	\absin{\E[\hat{\xi}_{\kappa} \hat{\zeta}_{\lambda}\mid \C_n]}&\le \E[\absin{\hat{\xi}_{\kappa} \hat{\zeta}_{\lambda}}\mid \C_n] \\
	&\le \left(\E[\absin{\xi'-\xi_{\kappa}}^p\mid \C_n]\right)^{1/p}\left(\PM(\xi'>\kappa\mid \C_n)\right)^{1/p'} \\
	&\quad\times \left(\E[\absin{\zeta'-\zeta_{\lambda}}^q\mid \C_n]\right)^{1/q}\left(\PM(\zeta'>\lambda\mid \C_n)\right)^{1/q'} \\
	&\le 4\kappa^{-p/p'}\lambda^{-q/q'} \qtext{a.s.\ on \ $\Xi$}.
\end{align*}
Combining these inequalities and multiplying by $\mu_{\xi,p}\mu_{\zeta,q}$, we get
\begin{align}
	\label{eq:cov_bound1}
	\begin{aligned}
		\abs{\Cov(\xi,\zeta\mid \C_n)}&\le \overline{\psi}_{a,b}(\mu_{\xi,p}\mu_{\zeta,q}) \kappa\lambda\theta_{n,s}+4\mu_{\xi,p}\mu_{\zeta,q} \\
		&\quad\times \left(\kappa\lambda^{1-q}+\kappa^{1-p}\lambda+\kappa^{-p/p'}\lambda^{-q/q'}+\kappa^{1-p}\lambda^{1-q}\right) \qtext{a.s.\ on \ $\Xi$}.
	\end{aligned}
\end{align}

Since \eqref{eq:cov_bound1} holds for all $\kappa,\lambda\ge 1$ a.s.\ on $\Xi$, it also holds for random $\kappa$ and $\lambda$ a.s.\ on $\Xi'=\Xi\cap \{(\kappa,\lambda)\in [1,\infty)^2\}$. Thus, setting $\kappa=\underline{\theta}_{n,s}^{-1/p}$ and $\lambda=\underline{\theta}_{n,s}^{-1/q}$ we get \eqref{eq:cov_ineq1} on $\Xi'$. As for the set $\Xi\cap\Xi'^c$, note that $\Cov(\xi,\zeta\mid \C_n)=0$ a.s.\ on $\{\theta_{n,s}=0\}$. Similarly, $\Cov(\xi,\zeta\mid \C_n)=0$ a.s.\ on $\{\mu_{\xi,p}=0\}\cup \{\mu_{\zeta,q}=0\}$, and $\{\mu_{\xi,p}=\infty\}$ and $\{\mu_{\zeta,q}=\infty\}$ are null sets.
\end{proof}

\begin{proof}[\textnbf{Proof of Theorem \ref{thm:cov_ineq2}}] We reuse the notation and bounds established in the proof of Theorem \ref{thm:cov_ineq1}. In addition, let
\begin{alignat*}{2}
	&\xi_{\dbl\kappa}\eqdef \mu_{\xi,p}^{-1}f_{(\kappa\mu_{\xi,p},\kappa\gamma_1)}(Y_{n,A})\ind_{\Xi}, \quad &\xi_{\dbl\kappa}^*\eqdef \xi_{\dbl\kappa}-\E[\xi_{\dbl\kappa}\mid \C_n], \\
	&\hat{\xi}_{\dbl\kappa}\eqdef \xi_{\kappa}-\xi_{\dbl\kappa}, &\hat{\xi}_{\dbl\kappa}^*\eqdef \hat{\xi}_{\dbl\kappa}-\E[\hat{\xi}_{\dbl\kappa}\mid \C_n],
\end{alignat*}
and, similarly, $\zeta_{\dbl\lambda},\zeta_{\dbl\lambda}^*,\hat{\zeta}_{\dbl\lambda}$, and $\hat{\zeta}^*_{\dbl\lambda}$, where $f$, $\mu_{\xi,p}$, and $\gamma_1$ are replaced by $g$, $\mu_{\zeta,q}$, and $\gamma_2$, respectively. Then
\begin{align*}
	\absin{\E[\xi_{\kappa}^*\zeta_{\lambda}^*\mid \C_n]}&\le \absin{\E[\xi_{\dbl\kappa}^*\zeta_{\dbl\lambda}^*\mid \C_n]}+\absin{\E[\xi_{\dbl\kappa}^*\hat{\zeta}_{\dbl\lambda}^*\mid \C_n]} \\
	&\quad+ \absin{\E[\hat{\xi}_{\dbl\kappa}^*\zeta_{\dbl\lambda}^*\mid \C_n]}+\absin{\E[\hat{\xi}_{\dbl\kappa}^*\hat{\zeta}_{\dbl\lambda}^*\mid \C_n]} \qtext{a.s.}
\end{align*}

Let $\tilde{\Xi}\eqdef\Xi\cap \{(\gamma_1,\gamma_2)\in (0,\infty)^2\}$. By Lemma \ref{lemma:psi_dep} and Assumption \ref{assu:psi_dep_gen} we find that
\begin{align*}
	\absin{\E[\xi_{\dbl\kappa}^*\zeta_{\dbl\lambda}^*\mid \C_n]}&\le \frac{\kappa\lambda}{\mu_{\xi,p}\mu_{\zeta,q}}\widetilde{\psi}_{a,b}(\mu_{\xi,p},\mu_{\zeta,q},\gamma_1,\gamma_2)\theta_{n,s} \qtext{a.s.\ on \ $\tilde{\Xi}$}.
\end{align*}
Second, noticing that $\{\absin{\hat{\zeta}_{\dbl\lambda}}>0\}\subseteq \bigcup_{i\in B}\bigcup_{1\le k\le v}\{\absin{[Y_{n,i}]_k}>h_2(\lambda\gamma_2)\}$ ($\because \zeta_{\lambda}\ne \zeta_{\lambda\lambda}$ only if $Y_{n,B}\ne \varphi_{h_2(\lambda\gamma_2)}(Y_{n,B})$) and using the conditional Markov inequality,
\begin{align*}
	\absin{\E[\xi_{\dbl\kappa}^*\hat{\zeta}_{\dbl\lambda}^*\mid \C_n]}&= \absin{\E[\xi_{\dbl\kappa}^*\hat{\zeta}_{\dbl\lambda}\mid \C_n]}\le 2\kappa\E[\absin{\hat{\zeta}_{\dbl\lambda}}\mid \C_n] \\
	&\le 4\kappa\lambda\sum_{i\in B}\sum_{1\le k\le v}\PM\left(\absin{[Y_{n,i}]_k}>h_2(\lambda\gamma_2)\mid \C_n\right) \\
	&\le 4bv\cdot\kappa\lambda^{1-q} \qtext{a.s.\ on \ $\tilde{\Xi}$}.
\end{align*}
Similarly,
\begin{align*}
	\absin{\E[\hat{\xi}_{\dbl\kappa}^*\zeta_{\dbl\lambda}^*\mid \C_n]}&\le 4av\cdot\kappa^{1-p}\lambda \qtext{a.s.\ on \ $\tilde{\Xi}$},
\end{align*}
and for $p',q'$ s.t.\ $1/p'+1/q'=1-1/p-1/q$,
\begin{align*}
	\absin{\E[\hat{\xi}_{\dbl\kappa}^*\hat{\zeta}_{\dbl\lambda}^*\mid \C_n]}\le 4abv^2\left(\kappa^{-p/p'}\lambda^{-q/q'}+ \kappa^{1-p}\lambda^{1-q}\right) \qtext{a.s.\ on \ $\tilde{\Xi}$}.
\end{align*}
Finally, the result follows by modifying the inequality \eqref{eq:cov_bound1} established in the proof of Theorem \ref{thm:cov_ineq1} and choosing $\kappa=\underline{\theta}_{n,s}^{-1/p}$ and $\lambda=\underline{\theta}_{n,s}^{-1/q}$.
\end{proof}
\section{Additional Simulation Results}\label{sec:more_simulations}
In this section, we report additional simulation results. Table \ref{tab:coverage_1_5_1_7} reports the simulated coverage probabilities for the 95\% HAC-based confidence intervals with the constant for the bandwidth selection rule in equation \eqref{eq:hac_const} set to $1.7$, $1.8$, and $1.9$. Table \ref{tab:coverage_1_8_2_2} reports the same results for the constant set to $2.0$, $2.1$, and $2.2$. For both tables, $\varepsilon$ in \eqref{eq:hac_const} is set to $0.05$.

The results show similar patterns to those reported in the main text: while in the majority of the cases (and across all the considered bandwidth constants) the simulated coverage is close to the nominal $0.95$,  the performance of the HAC-based confidence intervals deteriorates for larger values of the denseness parameter $\lambda$ and the dependence parameter $\gamma$. The coverage improves with the sample size for all considered values of the constant in the bandwidth selection rule. The worst results are observed in the cases of $\lambda=5$ and $\gamma=0.5$. For example, in the case of the $\text{constant}=1.7$, the coverage probability is only $0.800$ for $n=500$, but it improves to $0.870$ for $n=5000$.

Figure \ref{fig:power} reports the simulated rejection probabilities for the two-sided HAC-based $t$-test of the null hypothesis  of a zero mean. We report the results for $\lambda=3$ and the bandwidth constant of $2.0$, and the figure shows that the probability of Type II error increases with the dependence parameter $\gamma$.

\begin{landscape}
\renewcommand{\arraystretch}{1.3}
\setlength{\tabcolsep}{5pt}
\begin{table}[htbp]
  \centering
  \caption{Simulated coverage probabilities of the 95\% HAC-based confidence intervals for different values of the denseness parameter $\lambda$, sample size $n$, the dependence parameter $\gamma$, and the bandwidth  $\text{constant}=1.7, 1.8,1.9$}
\small
    \resizebox{\columnwidth}{!}{\begin{tabular}{rcccrrrrrrrrrrrrrrrrrrrrr}
       \toprule                 &       &       &       & \multicolumn{6}{c}{$\gamma$}                     &       & \multicolumn{6}{c}{$\gamma$}                     &       & \multicolumn{6}{c}{$\gamma$}                     &  \\
\cmidrule{5-10}\cmidrule{12-17}\cmidrule{19-24}          & $\lambda$ & $n$     &       & \multicolumn{1}{c}{0.0} & \multicolumn{1}{c}{0.1} & \multicolumn{1}{c}{0.2} & \multicolumn{1}{c}{0.3} & \multicolumn{1}{c}{0.4} & \multicolumn{1}{c}{0.5} &       & \multicolumn{1}{c}{0.0} & \multicolumn{1}{c}{0.1} & \multicolumn{1}{c}{0.2} & \multicolumn{1}{c}{0.3} & \multicolumn{1}{c}{0.4} & \multicolumn{1}{c}{0.5} &       & \multicolumn{1}{c}{0.0} & \multicolumn{1}{c}{0.1} & \multicolumn{1}{c}{0.2} & \multicolumn{1}{c}{0.3} & \multicolumn{1}{c}{0.4} & \multicolumn{1}{c}{0.5} &  \\
\midrule   &       &       &       &       &       &       &       &       &       &       &       &       &       &       &       &       &       &       &       &       &       &       &       &  \\
     \noalign{\vskip-10pt}        &       &       &       & \multicolumn{6}{c}{\underline{$\text{constant}=1.7$}}              &       & \multicolumn{6}{c}{\underline{$\text{constant}=1.8$}}              &       & \multicolumn{6}{c}{\underline{$\text{constant}=1.9$}}            &  \\
   \noalign{\vskip-10pt}        &       &       &       &       &       &       &       &       &       &       &       &       &       &       &       &       &       &       &       &       &       &       &       &  \\
             & 1     & 500   &       & 0.948 & 0.944 & 0.948 & 0.944 & 0.947 & 0.944 &       & 0.948 & 0.944 & 0.948 & 0.944 & 0.947 & 0.944 &       & 0.948 & 0.944 & 0.948 & 0.944 & 0.947 & 0.944 &  \\
          & 1     & 1000  &       & 0.946 & 0.949 & 0.947 & 0.947 & 0.949 & 0.944 &       & 0.946 & 0.949 & 0.947 & 0.947 & 0.949 & 0.944 &       & 0.946 & 0.949 & 0.947 & 0.947 & 0.949 & 0.944 &  \\
          & 1     & 5000  &       & 0.949 & 0.948 & 0.952 & 0.951 & 0.950 & 0.947 &       & 0.949 & 0.948 & 0.952 & 0.951 & 0.950 & 0.947 &       & 0.949 & 0.948 & 0.952 & 0.951 & 0.950 & 0.947 &  \\
          &       &       &       &       &       &       &       &       &       &       &       &       &       &       &       &       &       &       &       &       &       &       &       &  \\
          & 2     & 500   &       & 0.939 & 0.934 & 0.934 & 0.932 & 0.925 & 0.910 &       & 0.938 & 0.933 & 0.933 & 0.931 & 0.925 & 0.911 &       & 0.937 & 0.932 & 0.934 & 0.931 & 0.924 & 0.912 &  \\
          & 2     & 1000  &       & 0.944 & 0.943 & 0.941 & 0.940 & 0.933 & 0.918 &       & 0.944 & 0.943 & 0.940 & 0.940 & 0.933 & 0.920 &       & 0.943 & 0.943 & 0.940 & 0.938 & 0.933 & 0.921 &  \\
          & 2     & 5000  &       & 0.947 & 0.946 & 0.946 & 0.946 & 0.943 & 0.935 &       & 0.947 & 0.946 & 0.946 & 0.946 & 0.944 & 0.937 &       & 0.947 & 0.947 & 0.946 & 0.946 & 0.944 & 0.938 &  \\
          &       &       &       &       &       &       &       &       &       &       &       &       &       &       &       &       &       &       &       &       &       &       &       &  \\
          & 3     & 500   &       & 0.940 & 0.931 & 0.932 & 0.921 & 0.903 & 0.865 &       & 0.937 & 0.930 & 0.930 & 0.920 & 0.903 & 0.867 &       & 0.937 & 0.928 & 0.928 & 0.920 & 0.904 & 0.868 &  \\
          & 3     & 1000  &       & 0.940 & 0.941 & 0.939 & 0.925 & 0.913 & 0.882 &       & 0.939 & 0.940 & 0.939 & 0.925 & 0.914 & 0.884 &       & 0.938 & 0.939 & 0.937 & 0.924 & 0.915 & 0.887 &  \\
          & 3     & 5000  &       & 0.946 & 0.949 & 0.944 & 0.936 & 0.930 & 0.911 &       & 0.946 & 0.948 & 0.944 & 0.936 & 0.930 & 0.914 &       & 0.945 & 0.947 & 0.944 & 0.937 & 0.931 & 0.916 &  \\
          &       &       &       &       &       &       &       &       &       &       &       &       &       &       &       &       &       &       &       &       &       &       &       &  \\
          & 4     & 500   &       & 0.936 & 0.928 & 0.923 & 0.908 & 0.884 & 0.826 &       & 0.933 & 0.925 & 0.921 & 0.907 & 0.884 & 0.829 &       & 0.932 & 0.923 & 0.920 & 0.907 & 0.885 & 0.831 &  \\
          & 4     & 1000  &       & 0.941 & 0.935 & 0.930 & 0.925 & 0.898 & 0.853 &       & 0.940 & 0.935 & 0.930 & 0.925 & 0.899 & 0.856 &       & 0.938 & 0.934 & 0.929 & 0.924 & 0.900 & 0.858 &  \\
          & 4     & 5000  &       & 0.944 & 0.945 & 0.942 & 0.936 & 0.917 & 0.888 &       & 0.944 & 0.945 & 0.942 & 0.937 & 0.919 & 0.892 &       & 0.944 & 0.944 & 0.942 & 0.937 & 0.921 & 0.894 &  \\
          &       &       &       &       &       &       &       &       &       &       &       &       &       &       &       &       &       &       &       &       &       &       &       &  \\
          & 5     & 500   &       & 0.933 & 0.927 & 0.917 & 0.900 & 0.864 & 0.800 &       & 0.932 & 0.926 & 0.916 & 0.900 & 0.866 & 0.805 &       & 0.930 & 0.923 & 0.915 & 0.898 & 0.865 & 0.808 &  \\
          & 5     & 1000  &       & 0.942 & 0.935 & 0.933 & 0.917 & 0.884 & 0.828 &       & 0.940 & 0.932 & 0.931 & 0.915 & 0.886 & 0.834 &       & 0.937 & 0.930 & 0.931 & 0.915 & 0.887 & 0.839 &  \\
          & 5     & 5000  &       & 0.950 & 0.944 & 0.942 & 0.930 & 0.914 & 0.870 &       & 0.950 & 0.944 & 0.942 & 0.931 & 0.916 & 0.875 &       & 0.949 & 0.943 & 0.942 & 0.930 & 0.919 & 0.881 &  \\
          &       &       &       &       &       &       &       &       &       &       &       &       &       &       &       &       &       &       &       &       &       &       &       &  \\
         \bottomrule
    \end{tabular}}%
  \label{tab:coverage_1_5_1_7}%
\end{table}%
\end{landscape}

\begin{landscape}
\renewcommand{\arraystretch}{1.3}
\setlength{\tabcolsep}{5pt}
\begin{table}[htbp]
  \centering
  \caption{Simulated coverage probabilities of the 95\% HAC-based confidence intervals for different values of the denseness parameter $\lambda$, sample size $n$, the dependence parameter $\gamma$, and the bandwidth  $\text{constant}=2.0, 2.1,2.2$}
\small
    \resizebox{\columnwidth}{!}{\begin{tabular}{rrrrrrrrrrrrrrrrrrrrrrr}
    \toprule
          &       &       & \multicolumn{6}{c}{$\gamma$}                     &       & \multicolumn{6}{c}{$\gamma$}                     &       & \multicolumn{6}{c}{$\gamma$} \\
\cmidrule{4-9}\cmidrule{11-16}\cmidrule{18-23}    \multicolumn{1}{c}{$\lambda$} & \multicolumn{1}{c}{$n$} &       & \multicolumn{1}{c}{0.0} & \multicolumn{1}{c}{0.1} & \multicolumn{1}{c}{0.2} & \multicolumn{1}{c}{0.3} & \multicolumn{1}{c}{0.4} & \multicolumn{1}{c}{0.5} &       & \multicolumn{1}{c}{0.0} & \multicolumn{1}{c}{0.1} & \multicolumn{1}{c}{0.2} & \multicolumn{1}{c}{0.3} & \multicolumn{1}{c}{0.4} & \multicolumn{1}{c}{0.5} &       & \multicolumn{1}{c}{0.0} & \multicolumn{1}{c}{0.1} & \multicolumn{1}{c}{0.2} & \multicolumn{1}{c}{0.3} & \multicolumn{1}{c}{0.4} & \multicolumn{1}{c}{0.5} \\
    \midrule
          &       &       &       &       &       &       &       &       &       &       &       &       &       &       &       &       &       &       &       &       &       &  \\
      \noalign{\vskip-10pt}       &       &       & \multicolumn{6}{c}{\underline{$\text{constant}=2.0$}}              &       & \multicolumn{6}{c}{\underline{$\text{constant}=2.1$} }             &       & \multicolumn{6}{c}{\underline{$\text{constant}=2.2$}} \\
    \noalign{\vskip-10pt}         &       &       &       &       &       &       &       &       &       &       &       &       &       &       &       &       &       &       &       &       &       &  \\
   1     & 500   &       & 0.948 & 0.944 & 0.948 & 0.944 & 0.947 & 0.944 &       & 0.948 & 0.944 & 0.948 & 0.944 & 0.947 & 0.944 &       & 0.948 & 0.944 & 0.948 & 0.944 & 0.947 & 0.944 \\
    1     & 1000  &       & 0.946 & 0.949 & 0.947 & 0.947 & 0.949 & 0.944 &       & 0.946 & 0.949 & 0.947 & 0.947 & 0.949 & 0.944 &       & 0.946 & 0.949 & 0.947 & 0.947 & 0.949 & 0.944 \\
    1     & 5000  &       & 0.949 & 0.948 & 0.952 & 0.951 & 0.950 & 0.947 &       & 0.949 & 0.948 & 0.952 & 0.951 & 0.950 & 0.947 &       & 0.949 & 0.948 & 0.952 & 0.951 & 0.950 & 0.947 \\
          &       &       &       &       &       &       &       &       &       &       &       &       &       &       &       &       &       &       &       &       &       &  \\
    2     & 500   &       & 0.936 & 0.931 & 0.933 & 0.931 & 0.924 & 0.912 &       & 0.935 & 0.931 & 0.932 & 0.930 & 0.924 & 0.912 &       & 0.935 & 0.930 & 0.932 & 0.930 & 0.923 & 0.912 \\
    2     & 1000  &       & 0.943 & 0.942 & 0.940 & 0.938 & 0.933 & 0.922 &       & 0.943 & 0.942 & 0.939 & 0.938 & 0.933 & 0.922 &       & 0.942 & 0.941 & 0.939 & 0.938 & 0.933 & 0.922 \\
    2     & 5000  &       & 0.947 & 0.947 & 0.945 & 0.946 & 0.944 & 0.938 &       & 0.947 & 0.947 & 0.945 & 0.946 & 0.945 & 0.939 &       & 0.947 & 0.946 & 0.945 & 0.947 & 0.945 & 0.939 \\
          &       &       &       &       &       &       &       &       &       &       &       &       &       &       &       &       &       &       &       &       &       &  \\
    3     & 500   &       & 0.936 & 0.926 & 0.926 & 0.918 & 0.903 & 0.869 &       & 0.933 & 0.923 & 0.924 & 0.917 & 0.902 & 0.870 &       & 0.931 & 0.921 & 0.922 & 0.916 & 0.900 & 0.870 \\
    3     & 1000  &       & 0.938 & 0.938 & 0.936 & 0.923 & 0.915 & 0.889 &       & 0.937 & 0.935 & 0.935 & 0.922 & 0.915 & 0.890 &       & 0.936 & 0.933 & 0.934 & 0.921 & 0.914 & 0.891 \\
    3     & 5000  &       & 0.944 & 0.947 & 0.943 & 0.937 & 0.932 & 0.919 &       & 0.943 & 0.946 & 0.942 & 0.937 & 0.932 & 0.920 &       & 0.943 & 0.946 & 0.942 & 0.936 & 0.932 & 0.921 \\
          &       &       &       &       &       &       &       &       &       &       &       &       &       &       &       &       &       &       &       &       &       &  \\
    4     & 500   &       & 0.929 & 0.921 & 0.918 & 0.905 & 0.885 & 0.833 &       & 0.926 & 0.918 & 0.915 & 0.902 & 0.884 & 0.834 &       & 0.922 & 0.915 & 0.912 & 0.901 & 0.882 & 0.836 \\
    4     & 1000  &       & 0.936 & 0.931 & 0.927 & 0.923 & 0.900 & 0.860 &       & 0.934 & 0.929 & 0.925 & 0.923 & 0.900 & 0.863 &       & 0.932 & 0.928 & 0.923 & 0.921 & 0.900 & 0.864 \\
    4     & 5000  &       & 0.943 & 0.943 & 0.941 & 0.937 & 0.921 & 0.898 &       & 0.942 & 0.943 & 0.941 & 0.937 & 0.922 & 0.900 &       & 0.941 & 0.942 & 0.940 & 0.937 & 0.924 & 0.903 \\
          &       &       &       &       &       &       &       &       &       &       &       &       &       &       &       &       &       &       &       &       &       &  \\
    5     & 500   &       & 0.928 & 0.919 & 0.910 & 0.894 & 0.864 & 0.810 &       & 0.924 & 0.915 & 0.908 & 0.892 & 0.864 & 0.811 &       & 0.922 & 0.913 & 0.906 & 0.891 & 0.862 & 0.813 \\
    5     & 1000  &       & 0.935 & 0.929 & 0.930 & 0.913 & 0.889 & 0.842 &       & 0.935 & 0.928 & 0.929 & 0.913 & 0.889 & 0.845 &       & 0.932 & 0.927 & 0.927 & 0.912 & 0.889 & 0.846 \\
    5     & 5000  &       & 0.949 & 0.942 & 0.942 & 0.931 & 0.920 & 0.885 &       & 0.948 & 0.942 & 0.941 & 0.931 & 0.921 & 0.888 &       & 0.948 & 0.941 & 0.941 & 0.931 & 0.922 & 0.890 \\
          &       &       &       &       &       &       &       &       &       &       &       &       &       &       &       &       &       &       &       &       &       &  \\
       \bottomrule
    \end{tabular}}%
  \label{tab:coverage_1_8_2_2}%
\end{table}%

\end{landscape}

\begin{figure}[t]
	\includegraphics{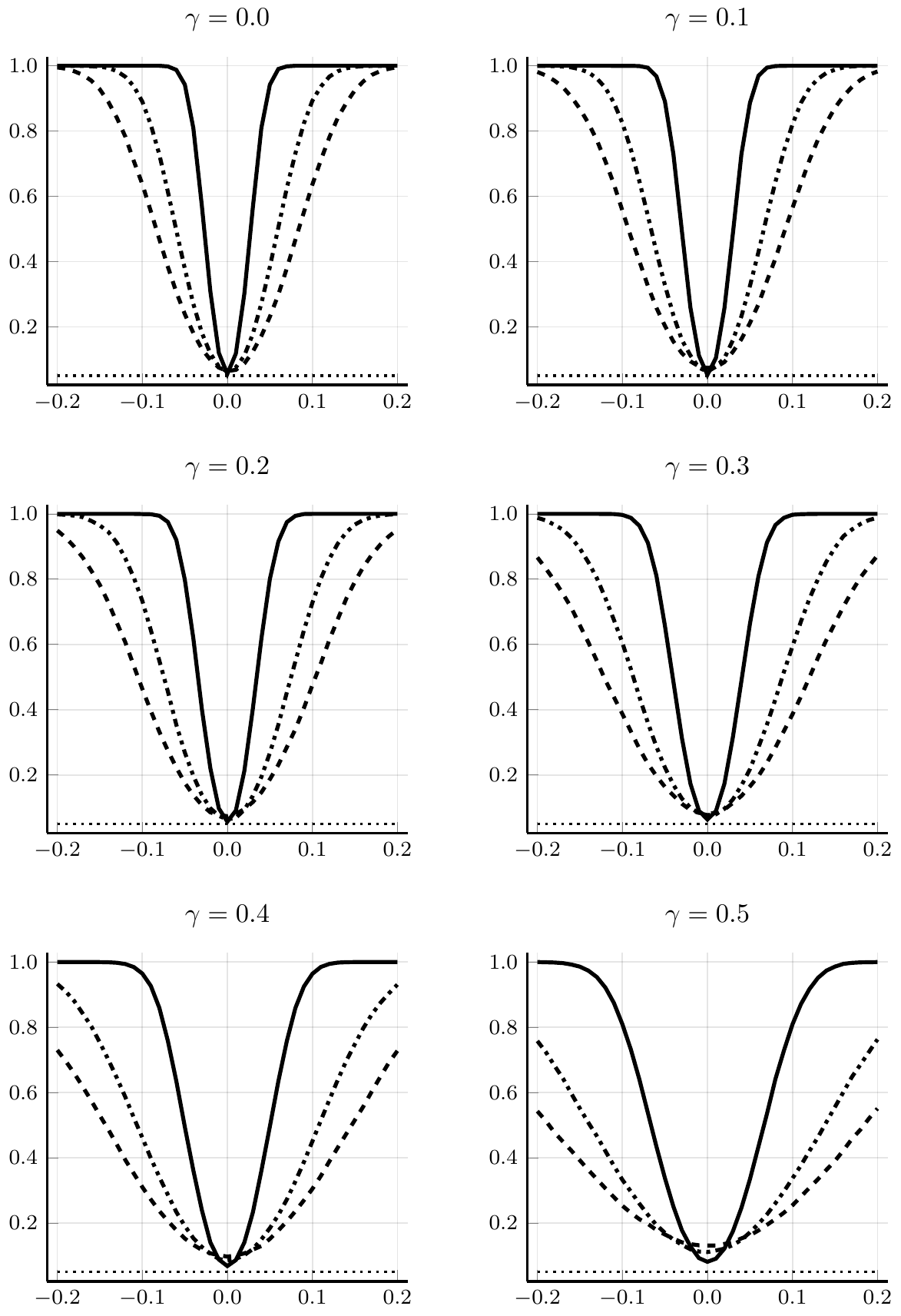}
	\caption{Simulated rejection probabilities for the 5\% two-sided HAC-based $t$-test of the null hypothesis of a zero mean for different values of the dependence parameter $\gamma$ and sample sizes $n=500$ (\textit{dashed line}), $n=1000$ (\textit{dot-dashed line}), and $n=5000$ (\textit{solid line}). The denseness parameter $\lambda=3$, and the bandwidth constant is set to $2.0$.\label{fig:power}}
\end{figure}

\clearpage
\putbib[network_dep]

\end{bibunit}


\begin{thebibliography}{49}
\expandafter\ifx\csname natexlab\endcsname\relax\def\natexlab#1{#1}\fi
\expandafter\ifx\csname url\endcsname\relax
  \def\url#1{\texttt{#1}}\fi
\expandafter\ifx\csname urlprefix\endcsname\relax\def\urlprefix{URL }\fi

\bibitem[{Andrews(1991)}]{Andrews:91}
Andrews, D. W.~K., 1991. Heteroskedasticity and autocorrelation consistent
  covariance matrix estimation. Econometrica 59~(3), 817--858.

\bibitem[{Aronow and Samii(2017)}]{Aronow&Samii:17}
Aronow, P., Samii, C., 2017. Estimating average causal effects under general
  interference, with application to a social network experiment. Annals of
  Applied Statistics 11~(4), 1912--1947.

\bibitem[{Baldi and Rinott(1989)}]{Baldi/Rinott:89:AP}
Baldi, P., Rinott, Y., 1989. On normal approximations of distributions in terms
  of dependency graphs. Annals of Probability 17, 1646--1650.

\bibitem[{Birkel(1988)}]{Birkel:88:AP}
Birkel, T., 1988. On the convergence rate in the central limit theorem for
  associated processes. Annals of Probability 16, 1685--1698.

\bibitem[{Bl\"{a}sius et~al.(2018)Bl\"{a}sius, Friedrich, and
  Krohmer}]{Blasius/Friedrich/Krohmer:17:Algorithmica}
Bl\"{a}sius, T., Friedrich, T., Krohmer, A., 2018. Cliques in hyperbolic random
  graphs. Algorithmica 80, 2324--2344.

\bibitem[{Blume et~al.(2015)Blume, Brock, Durlauf, and
  Jayaraman}]{Blume/Brock/Durlauf/Jayaraman:15:JPE}
Blume, L.~E., Brock, W.~A., Durlauf, S.~N., Jayaraman, R., 2015. Linear social
  interactions models. Journal of Political Economy 123, 444--496.

\bibitem[{Borg and Groenen(2005)}]{Borg/Groenen:05:MDS}
Borg, I., Groenen, P. J.~F., 2005. Modern Multidimensional Scaling. Springer
  Verlag, New York.

\bibitem[{Boucher and Mourifie(2017)}]{Boucher/Mourifie:17:EJ}
Boucher, V., Mourifie, I., 2017. My friend far, far away: A random field
  approach to exponential random graph models. Econometrics Journal 20,
  S14--S46.

\bibitem[{Chen et~al.(2011)Chen, Goldstein, and
  Shao}]{Chen/Goldstein/Shao:11:NormalApprox}
Chen, L. H.~Y., Goldstein, L., Shao, Q.-M., 2011. Normal approximation by
  {S}tein's method. Springer-Verlag, New York, USA.

\bibitem[{Chen and Shao(2004)}]{Chen/Shao:04:AP}
Chen, L. H.~Y., Shao, Q.-M., 2004. Normal approximation under local dependence.
  Annals of Probability 32, 1985--2028.

\bibitem[{Chung and Lu(2001)}]{Chung/Lu:01:AAM}
Chung, F., Lu, L., 2001. The diameter of sparse random graphs. Advances in
  Applied Mathematics 26, 257--279.

\bibitem[{Comets and Jan\v{z}ura(1998)}]{Comets/Janzura:98:JAP}
Comets, F., Jan\v{z}ura, M., 1998. A central limit theorem for conditionally
  centered random fields with an application to markov fields. Journal of
  Applied Probability 35, 608--621.

\bibitem[{Conley(1999)}]{Conley:99:JOE}
Conley, T.~G., 1999. {GMM} estimation with cross-sectional dependence. Journal
  of Econometrics 92, 1--45.

\bibitem[{Conley and Molinari(2007)}]{Conley/Molinary:07:JOE}
Conley, T.~G., Molinari, F., 2007. Spatial correlation robust inference with
  errors in location or distance. Journal of Econometrics 140, 76--96.

\bibitem[{Davidson(1994)}]{Davidson:94:StochasticLimitTheory}
Davidson, J., 1994. Stochastic Limit Theory: An Introduction for
  Econometricians. Oxford University Press, New York.

\bibitem[{Doukhan and Louhichi(1999)}]{Doukhan/Louhichi:99}
Doukhan, P., Louhichi, S., 1999. A new weak dependence condition and
  applications to moment inequalities. Stochastic Processes and their
  Applications 84~(2), 313--342.

\bibitem[{Gaetan and Guyon(2010)}]{Gaetan/Guyon:10:SpatialStatistics}
Gaetan, C., Guyon, X., 2010. Spatial Statistics and Modeling. Springer, New
  York.

\bibitem[{Graham(2017)}]{Graham:17}
Graham, B.~S., 2017. An econometric model of network formation with degree
  heterogeneity. Econometrica 85~(4), 1033--1063.

\bibitem[{Hall and Heyde(1980)}]{HallHeyde:80:MLT}
Hall, P.~G., Heyde, C.~C., 1980. Martingale limit theory and its applications.
  Academic Press, New York; London.

\bibitem[{Janson(1988)}]{Janson:88:AP}
Janson, S., 1988. Normal convergence by higher semiinvariants with applications
  to sums of dependent random variables and random graphs. Annals of
  Probability 16, 305--312.

\bibitem[{Jenish and Prucha(2009)}]{Jenish/Prucha:09:JOE}
Jenish, N., Prucha, I.~R., 2009. Central limit theorems and uniform laws of
  large numbers for arrays of random fields. Journal of Econometrics 150~(1),
  86--98.

\bibitem[{Jenish and Prucha(2012)}]{Jenish/Prucha:12:JOE}
Jenish, N., Prucha, I.~R., 2012. On spatial processes and asymptotic inference
  under near-epoch dependence. Journal of Econometrics 170, 178--190.

\bibitem[{Jia(2008)}]{Jia:08:Eca}
Jia, P., 2008. What happens when {W}al-{M}art comes to town: An empirical
  analysis of the discount retailing industry. Econometrica 170, 1263--1316.

\bibitem[{Jiménez-Sevilla and Sánchez-González(2011)}]{JimenezSevilla:11}
Jiménez-Sevilla, M., Sánchez-González, L., 2011. Smooth extension of
  functions on a certain class of non-separable {B}anach spaces. Journal of
  Mathematical Analysis and Applications 378~(1), 173--183.

\bibitem[{Johnsson and Moon(2019)}]{Johnsson/Moon:19:ReStatForthComing}
Johnsson, I., Moon, H.~R., 2019. Estimation of peer effects in endogenous
  social networks: Control function approach, {R}eview of Economics and
  Statistics, forthcoming.

\bibitem[{Kelejian and Prucha(2007)}]{Kelejian/Prucha:07}
Kelejian, H.~H., Prucha, I.~R., 2007. {HAC} estimation in a spatial framework.
  Journal of Econometrics, 131--154.

\bibitem[{Kim and Sun(2011)}]{Kim/Sun:11:JOE}
Kim, M.~S., Sun, Y., 2011. Spatial heteroskedasticity and autocorrelation
  consistent estimation of covariance matrix. Journal of Econometrics 160,
  349--371.

\bibitem[{Kojevnikov(2019)}]{Kojevnikov:19:WP}
Kojevnikov, D., 2019. The bootstrap for network dependent processes, {W}orking
  Paper.

\bibitem[{Kuersteiner(2019)}]{Kuersteiner:19:WP}
Kuersteiner, G.~M., 2019. Limit theorems for data with network structure,
  arXiv:1908.02375v1 [math.PR].

\bibitem[{Kuersteiner and Prucha(2015)}]{Kuersteiner/Pucha:15:WP}
Kuersteiner, G.~M., Prucha, I.~R., 2015. Dynamic spatial panel models:
  Networks, common shocks, and sequential exogeneity, {W}orking Paper.

\bibitem[{Lauritzen(1996)}]{Lauritzen:96:GraphicalModels}
Lauritzen, S.~L., 1996. Graphical Models. Clarendon Press, Oxford.

\bibitem[{Lee and Song(2019)}]{Lee/Song:18:Bernoulli}
Lee, J.~H., Song, K., 2019. Stable limit theorems for empirical processes under
  conditional neighborhood dependence. Bernoulli 25, 1189--1224.

\bibitem[{Lee(2004)}]{Lee:04:Eca}
Lee, L.-F., 2004. Asymptotic distributions of quasi-maximum likelihood
  estimators for spatial autoregressive models. Econometrica 72, 1899--1925.

\bibitem[{Lee et~al.(2010)Lee, Liu, and Lin}]{Lee/Liu/Lin:10:EJ}
Lee, L.-F., Liu, X., Lin, X., 2010. Specification and estimation of social
  interaction models with network structures. Econometrics Journal 13,
  145--176.

\bibitem[{Leung(2019{\natexlab{a}})}]{Leung:19:WP}
Leung, M.~P., 2019{\natexlab{a}}. Inference in models of discrete choice with
  social interactions using network data, arXiv:1911.07106 [econ.EM].

\bibitem[{Leung(2019{\natexlab{b}})}]{Leung:19:JOE}
Leung, M.~P., 2019{\natexlab{b}}. A weak law for moments of pairwise stable
  networks. Journal of Econometrics 210, 310--326.

\bibitem[{Leung(2020)}]{Leung:20:ReStat}
Leung, M.~P., 2020. Treatment and spillover effects under network interference,
  {R}eview of Economics and Statistics.

\bibitem[{Leung and Moon(2019)}]{Leung/Moon:19:WPb}
Leung, M.~P., Moon, H.~R., 2019. Normal approximation in large network models,
  arXiv:1904.11060v1 [econ.EM].

\bibitem[{Murphy(2012)}]{Murphy:12:ML}
Murphy, K.~P., 2012. Machine Learning: A Probabilistic Perspective. The MIT
  Press, New York, USA.

\bibitem[{Pearl(2009)}]{Pearl:09:Causality}
Pearl, J., 2009. Causality. Cambridge University Press, New York.

\bibitem[{Penrose(2003)}]{Penrose:03:RandomGeometricGraphs}
Penrose, M., 2003. Random Geometric Graphs. Oxford University Press, New York,
  USA.

\bibitem[{Petty(1971)}]{Petty:71}
Petty, C.~M., 1971. Equilateral sets in {M}inkowski spaces. Proc. Amer. Math.
  Soc. 29, 369--374.

\bibitem[{Prakasa~Rao(2013)}]{PrakasaRao:13}
Prakasa~Rao, B. L.~S., 2013. Conditional independence, conditional mixing and
  conditional association. Annals of the Institute of Statistical Mathematics
  61, 441--460.

\bibitem[{Ridder and Sheng(2019)}]{Ridder/Sheng:19:WP}
Ridder, G., Sheng, S., 2019. Estimation of large network formation games,
  working paper.

\bibitem[{Rinott and Rotar(1996)}]{Rinott/Rotar:96:JMA}
Rinott, Y., Rotar, V., 1996. A multivariate {CLT} for local dependence with
  $n^{1/2} \log n$ rate and applications to multivariate graph related
  statistics. Journal of Multivariate Analysis 56, 333--350.

\bibitem[{Song(2018)}]{Song:17:ReStat}
Song, K., 2018. Measuring the graph concordance of locally dependent
  observations. Review of Economics and Statistics 100, 535--549.

\bibitem[{Stein(1972)}]{Stein:72:BerkeleySymp}
Stein, C., 1972. A bound for the error in the normal approximation to the
  distribution of a sum of dependent random variables. Proceedings in the Sixth
  Berkeley Symposium on Mathematical Statistics 2, 583--602.

\bibitem[{Stein(1986)}]{Stein:86}
Stein, C., 1986. Approximate computation of expectations. Lecture
  Notes-Monograph Series 7, i--164.

\bibitem[{Wendland(2004)}]{Wendland:04:ScatteredDataApprox}
Wendland, H., 2004. Scattered Data Approximation. Cambridge Monographs on
  Applied and Computational Mathematics. Cambridge University Press.

\end{thebibliography}


\begin{thebibliography}{3}
\expandafter\ifx\csname natexlab\endcsname\relax\def\natexlab#1{#1}\fi
\expandafter\ifx\csname url\endcsname\relax
  \def\url#1{\texttt{#1}}\fi
\expandafter\ifx\csname urlprefix\endcsname\relax\def\urlprefix{URL }\fi

\bibitem[{Chung and Lu(2001)}]{Chung/Lu:01:AAM}
Chung, F., Lu, L., 2001. The diameter of sparse random graphs. Advances in
  Applied Mathematics 26, 257--279.

\bibitem[{Da~Prato and Zabczyk(2014)}]{DaPrato:14:StochasticEq}
Da~Prato, G., Zabczyk, J., 2014. Stochastic Equations in Infinite Dimensions,
  2nd Edition. Encyclopedia of Mathematics and Its Applications. Cambridge
  University Press.

\bibitem[{{van der Vaart} and Wellner(1996)}]{vanderVaart/Wellner:96:WeakConvg}
{van der Vaart}, A.~W., Wellner, J.~A., 1996. Weak Convergence and Empirical
  Processes. Springer, New York, USA.

\end{thebibliography}
\end{document}